\documentclass[11pt]{article}
\usepackage[utf8]{inputenc}
\usepackage{tikz}
\usetikzlibrary{shapes.geometric, arrows.meta, positioning, calc, shadows.blur}

\usepackage{amsmath, mathrsfs, amssymb, amsthm, fullpage, graphicx, mathtools, bbm, hyperref}
\usepackage[ruled,vlined]{algorithm2e}

\newtheorem{theorem}{Theorem}
\newtheorem{lemma}{Lemma}

\newtheorem*{example*}{Example}

\makeatletter
\def\BState{\State\hskip-\ALG@thistlm}
\makeatother

\newtheorem{definition}{Definition}
\newtheorem{proposition}{Proposition}
\newtheorem{corollary}{Corollary}
  {
      \theoremstyle{plain}
      \newtheorem{assumption}{Assumption}
  }

\numberwithin{equation}{section}

\usepackage{caption}
\usepackage{chngcntr}

\newcommand{\mme}[0]{\mathbb{E}}
\newcommand{\mmp}[0]{\mathbb{P}}
\newcommand{\mmr}[0]{\mathbb{R}}
\newcommand{\mmn}[0]{\mathbb{N}}

\newcommand{\bone}[0]{\mathbbm{1}}

\providecommand{\argmin}{\mathop{\rm argmin}}

\newcommand{\indic}[1]{\bone\left\{#1\right\}} 

\DeclarePairedDelimiterX{\inp}[2]{\langle}{\rangle}{#1, #2}

\usepackage[nameinlink, noabbrev]{cleveref}

\DeclareRobustCommand{\abbrevcrefs}{%
\Crefname{theorem}{Thm.}{Thms.}%
\Crefname{example}{Ex.}{Exs.}%
\crefname{equation}{eqn.}{eqns.}%
}

\DeclareRobustCommand{\Cshref}[1]{{\abbrevcrefs\Cref{#1}}}

\usepackage{multirow,float}

\usepackage{relsize}

\usepackage{accents}

\usepackage[round,comma,authoryear]{natbib}

\definecolor{LightDarkRed}{RGB}{235,170,170}
\definecolor{LightDarkBlue}{RGB}{150,180,235}
\definecolor{LightDarkOrange}{RGB}{255,200,150}
\definecolor{newOrange}{RGB}{255,130,0}
\definecolor{newBlue}{RGB}{135,206,250}

\newcommand{\footremember}[2]{%
    \footnote{#2}
    \newcounter{#1}
    \setcounter{#1}{\value{footnote}}%
}
\newcommand{\footrecall}[1]{%
    \footnotemark[\value{#1}]%
}

\title{Conformal Prediction with Conditional Guarantees}
\author{Isaac Gibbs\footremember{stats}{Department of Statistics, Stanford University.}\footremember{address}{Address for correspondence: Sequoia Hall, Stanford University, 390 Serra Mall, Stanford, CA, 94305, USA.
 Emails: \href{mailto:igibbs@stanford.edu}{igibbs@stanford.edu}, \href{mailto:jcherian@stanford.edu}{jcherian@stanford.edu}.} \and John J. Cherian\footrecall{stats}{} \and Emmanuel J. Cand\`{e}s\footremember{mathstats}{Departments of Mathematics and Statistics, Stanford University.}}
\date{}

\makeatletter
\@addtoreset{footnote}{section}
\makeatother

\allowdisplaybreaks
\providecommand{\keywords}[1]
{
  \small	
  \textbf{\textit{Keywords---}} #1
}

\begin{document}

\maketitle

\begin{abstract}
We consider the problem of constructing distribution-free prediction sets with finite-sample conditional guarantees. Prior work has shown that it is impossible to provide exact conditional coverage universally in finite samples. Thus, most popular methods only guarantee marginal coverage over the covariates or are restricted to a limited set of conditional targets, e.g. coverage over a finite set of pre-specified subgroups. This paper bridges this gap by defining a spectrum of problems that interpolate between marginal and conditional validity. We motivate these problems by reformulating conditional coverage as coverage over a class of covariate shifts. When the target class of shifts is finite-dimensional, we show how to simultaneously obtain exact finite-sample coverage over all possible shifts. For example, given a collection of subgroups, our prediction sets guarantee coverage over each group. For more flexible, infinite-dimensional classes where exact coverage is impossible, we provide a procedure for quantifying the coverage errors of our algorithm. Moreover, by tuning interpretable hyperparameters, we allow the practitioner to control the size of these errors across shifts of interest. Our methods can be incorporated into existing split conformal inference pipelines, and thus can be used to quantify the uncertainty of modern black-box algorithms without distributional assumptions.
\end{abstract}


\section{Introduction}\label{sec:intro}
Consider a training dataset $\{(X_i,Y_i)\}_{i=1}^{n}$, and a test point $(X_{n + 1}, Y_{n + 1})$, all drawn i.i.d.~from an unknown, arbitrary distribution $P$. We study the problem of using the observed data $\{(X_i,Y_i)\}_{i=1}^{n} \cup \{X_{n+1}\}$ to construct a prediction set $\hat{C}(X_{n + 1})$ that includes $Y_{n + 1}$ with probability $1 - \alpha$ over the randomness in the training and test points. 

Ideally, we would like this prediction set to meet three competing goals: it should (1) make no assumptions on the underlying data-generating mechanism, (2) be valid in finite samples, and (3) satisfy the \emph{conditional coverage} guarantee, $\mmp(Y_{n + 1} \in \hat{C}(X_{n + 1}) \mid X_{n + 1} = x) = 1 - \alpha$. Unfortunately, prior work has shown that is impossible to obtain all three of these conditions simultaneously (\cite{Vovk2012, Barber2020}). Perhaps the closest method to achieving these goals is conformal prediction, which relaxes the third criterion to the \emph{marginal coverage} guarantee, $\mmp(Y_{n + 1} \in \hat{C}(X_{n + 1})) = 1 - \alpha$. 

\begin{figure}
\centering
\begin{tikzpicture}[
    node distance=2.5cm,
    startstop/.style={
        rectangle, rounded corners, text width=3cm, minimum height=1.5cm, draw=black, align=center,
        blur shadow, thick, font=\small
    },
    arrow/.style={thick,->,>=Stealth, shorten >=2pt, shorten <=2pt},
    downarrow/.style={thick,->,>=Stealth, shorten >=1pt, shorten <=1pt},
]

\node (mlmodel) [startstop, fill=black!80, text=white] {ML model};
\node (score) [startstop, right=of mlmodel, fill=gray!30, text=black] {Conformity score};
\node (conformal) [startstop, right=of score, yshift=1.7cm, fill=LightDarkBlue, text=black] {Split conformal calibration};
\node (ourmethod) [startstop, right=of score, yshift=-1.7cm, fill=LightDarkOrange, text=black] {Conditional calibration};

\node[below=0.5cm of mlmodel, text width=3.5cm, align=center, font=\footnotesize] (mlmodelannotation) {Accurate point prediction};
\node[below=0.5cm of score, text width=3.5cm, align=center, font=\footnotesize] (scoreannotation) {Prediction errors on hold-out set};
\node[below=0.5cm of conformal, text width=3cm, align=center, font=\footnotesize] (conformalannotation) {Marginal coverage \\ \Cshref{thm:split_marginal_cov}};
\node[below=0.5cm of ourmethod, text width=3cm, align=center, font=\footnotesize] (ourmethodannotation) {Conditional coverage \\\Cshref{thm:finite_dim_result,thm:infinite_dim_result}};

\draw [arrow] (mlmodel) -- (score);
\draw [arrow] (score) -- (conformal);
\draw [arrow] (score) -- (ourmethod);

\draw [downarrow] (mlmodel) -- (mlmodelannotation);
\draw [downarrow] (score) -- (scoreannotation);
\draw [downarrow] (conformal) -- (conformalannotation);
\draw [downarrow] (ourmethod) -- (ourmethodannotation);

\end{tikzpicture}
\caption{Predicting with finite-sample guarantees: our conditionally valid pipeline vs. split conformal prediction.}
\label{fig:conformal_pipeline}
\end{figure}

The gap between conditional and marginal coverage can be extremely consequential in high-stakes decision-making. Marginal validity does not preclude substantial variation in coverage among relevant subpopulations. For example, a conformal prediction set for predicting a drug candidate's binding affinity achieves marginal $1 - \alpha$ coverage even if it underestimates the predictive uncertainty over the most promising subset of compounds. In human-centered applications, marginally valid prediction sets can be untrustworthy for certain legally protected groups (e.g., those defined by sensitive attributes such as race, gender, age, etc.) (\cite{Romano2020}). 

Achieving practical, finite-sample results requires weakening our desideratum from exact conditional coverage. In this article, we pursue a goal that is motivated by the following equivalence:
\begin{gather*}
    \mmp(Y_{n + 1} \in \hat{C}(X_{n + 1}) \mid X_{n + 1} = x) = 1 - \alpha, \quad \text{for all $x$}\\
    \iff \\
    \mme [ f(X_{n + 1})  ( \mathbf{1} \{Y_{n + 1} \in \hat{C}(X_{n + 1})\} - (1 - \alpha) ) ] = 0, \quad \text{for all measurable $f$}.
\end{gather*}
In particular, we define a relaxed coverage objective by replacing ``all measurable $f$'' with ``all $f$ belonging to some (potentially infinite) class $\mathcal{F}$.'' At the least complex end, taking $\mathcal{F} = \{x \mapsto 1\}$ recovers marginal validity, while more intermediate choices interpolate between marginal and conditional coverage. 

This generic approach to relaxing conditional validity was first popularized by the ``conditional-to-marginal'' moment testing literature (\cite{andrews2013inference}). Our relaxation is also referred to as a ``multi-accuracy'' objective in theoretical computer science (\cite{hebert2018multicalibration, kim2019multiaccuracy}). We remark that \citet{Deng2023} have concurrently proposed the same objective for the conditional coverage problem. However, their results focus on the infinite data regime, where the distribution of $(X,Y)$ is known exactly, and their algorithm requires access to an unspecified black-box optimization oracle.

By contrast, our proposed method preserves the attractive assumption-free and computationally efficient properties of split conformal prediction. Emulating split conformal, we design our procedure as a wrapper that takes any black-box machine learning model as input. We then compute conformity scores that measure the accuracy of this model's predictions on new test points. Finally, by calibrating bounds for these scores, we obtain prediction sets with finite-sample conditional guarantees. \Cref{fig:conformal_pipeline} displays our workflow.

Unlike split conformal, our approach adaptively compensates for poor initial modeling decisions. In particular, if the prediction rule and conformity scores were well-designed at the outset, our procedure may only make small adjustments. This could happen, for instance, if the scores were derived from a well-specified parametric model. More often, however, the user will begin with an inaccurate or incomplete model that fails to fully capture the distribution of $Y \mid X$. In these cases, our procedure will improve on split conformal by recalibrating the conformity score to provide exact conditional guarantees. In the predictive inference literature, conformal inference is often described as a protective layer that lies on top of a black-box machine learning model and transforms its point predictions into valid prediction sets. With this in mind, one might view our method as an additional protective layer that lies on top of a conformal method and transforms its (potentially poor) marginally valid outputs into richer, conditionally valid, prediction sets.

A number of prior works have also considered either modifying the split conformal calibration step (\cite{Lei2014, Leying2022, Barber2023}) or the initial prediction rule (\cite{Romano2019, Romano2021, Chernozhukov2021DistConf}) to better model the distribution of $Y \mid X$. Crucially, despite heuristic improvements in the quality of the resulting prediction sets, all of the aforementioned approaches obtain at most weak asymptotic guarantees that rely on slow, non-parametric convergence rates. 

Perhaps the only procedures to obtain a practical coverage guarantee are those of \citet{Barber2020}, \citet{Vovk2003}, and \citet{Jung2023}. All of these methods guarantee a form of group-conditional coverage, i.e. $\mmp(Y_{n+1} \in \hat{C}(X_{n+1}) \mid X_{n+1} \in G) \geq 1-\alpha$ for all sets $G$ in some pre-specified class $\mathcal{G}$.  However, the approach of \citet{Barber2020} can be computationally infeasible and severely conservative, yielding wide intervals with coverage probability far above the target level. On the other hand, the method of \citet{Vovk2003}, Mondrian conformal prediction, provides exact coverage in finite samples, but does not allow the groups in $\mathcal{G}$ to overlap. Finally, \citet{Jung2023} propose running quantile regression over the linear function class $\{\sum_{G \in \mathcal{G}} \beta_G\bone\{x \in G\} : \beta \in \mathbb{R}^{|G|}\}$. This method is both practical and allows for overlapping groups. In this work, we propose a new method for achieving conditional coverage that improves upon the method of \citet{Jung2023} in three ways; 1) by providing tighter finite-sample coverage, 2) by requiring no assumptions on the data-generating distribution (in particular, unlike \citet{Jung2023} we allow for discrete outcomes), and 3) by providing conditional coverage guarantees far beyond the group setting. 

\subsection{Preview of contributions}

To motivate and summarize the main contributions of our paper, we preview some applications. In particular, we show how our method can be used to satisfy two popular coverage desiderata: group-conditional coverage and coverage under covariate shift. Additionally, we demonstrate the improved finite sample performance of our method compared to previous approaches.

\subsubsection{Group-conditional coverage}
\begin{figure}[ht!]
    \centering
    \includegraphics[width=\textwidth]{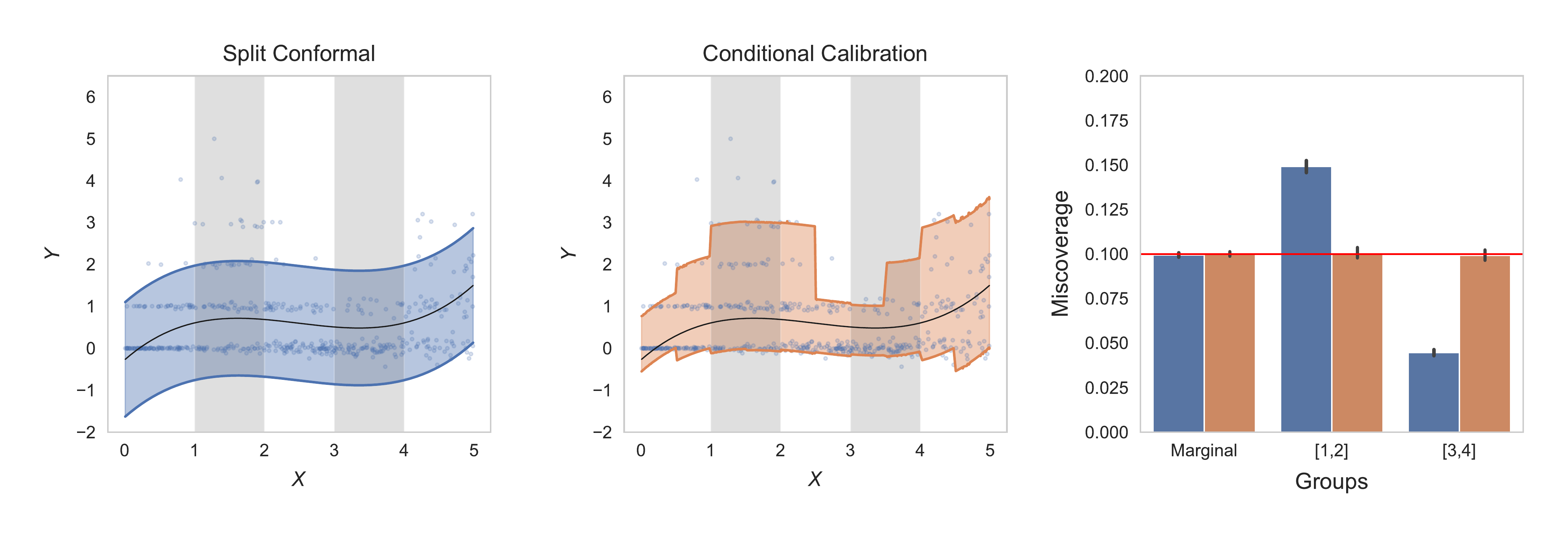}
    \caption{Comparison of split conformal prediction (blue, left-most panel) and the randomized implementation of our method (orange, center panel) on a simulated dataset first considered by \citet{Romano2019}. Black curves denote an estimate of the conditional mean, while the blue and orange shaded regions indicate the fitted prediction intervals. For this experiment, our method is implemented using the procedure outlined in \Cref{sec:finite_dim} with $\mathcal{F} := \{\sum_{G \in \mathcal{G}} \beta_G \bone\{x \in G\}: \beta \in \mmr^{|\mathcal{G}|}\}$. The rightmost panel shows the miscoverage of the two methods marginally over the x-axis and conditionally on x falling in the two grey shaded bands; the red line indicates the target level of $\alpha = 0.1$.}
    \label{fig:cqr_data}
\end{figure}
\emph{Group-conditional coverage} requires that $\hat{C}(\cdot)$ satisfy $\mmp (Y_{n + 1} \in \hat{C}(X_{n + 1}) \mid X_{n + 1} \in G) = 1 - \alpha$ for all $G$ belonging to some collection of pre-specified (potentially overlapping) groups $\mathcal{G} \subseteq 2^{\text{Domain}(X)}$ (\cite{Barber2020}). This corresponds to a special case of our guarantee in which $\mathcal{F} = \{\sum_{G \in \mathcal{G}} \beta_G \bone\{x \in G\}: \beta \in \mmr^{|\mathcal{G}|}\}$.

\Cref{fig:cqr_data} illustrates the coverage guarantee on a simulated dataset. Here, $x$ is univariate and we have taken the groups $\mathcal{G}$ to be the collection of all sub-intervals with endpoints belonging to $\{0,0.5,1,\dots,5\}$. Two of these sub-intervals, $[1,2]$ and $[3,4]$, are shaded in grey. We compare two procedures, split conformal prediction and our conditional calibration method. As is standard, we implement split conformal using conformity score $S(x,y) := |y - \hat{\mu}(x)|$ where $\hat{\mu}(x)$ is an estimate of the conditional mean $\mme[Y \mid X]$, while for our method, we take a two-sided approach in which upper and lower bounds on $y - \hat{\mu}(x)$ are computed separately. We see that while split conformal prediction only provides marginal validity, our method returns prediction sets that are adaptive to the shape of the data and thus obtain exact coverage over all subgroups.

\begin{figure}[ht!]
    \centering
    \includegraphics[width=\textwidth]{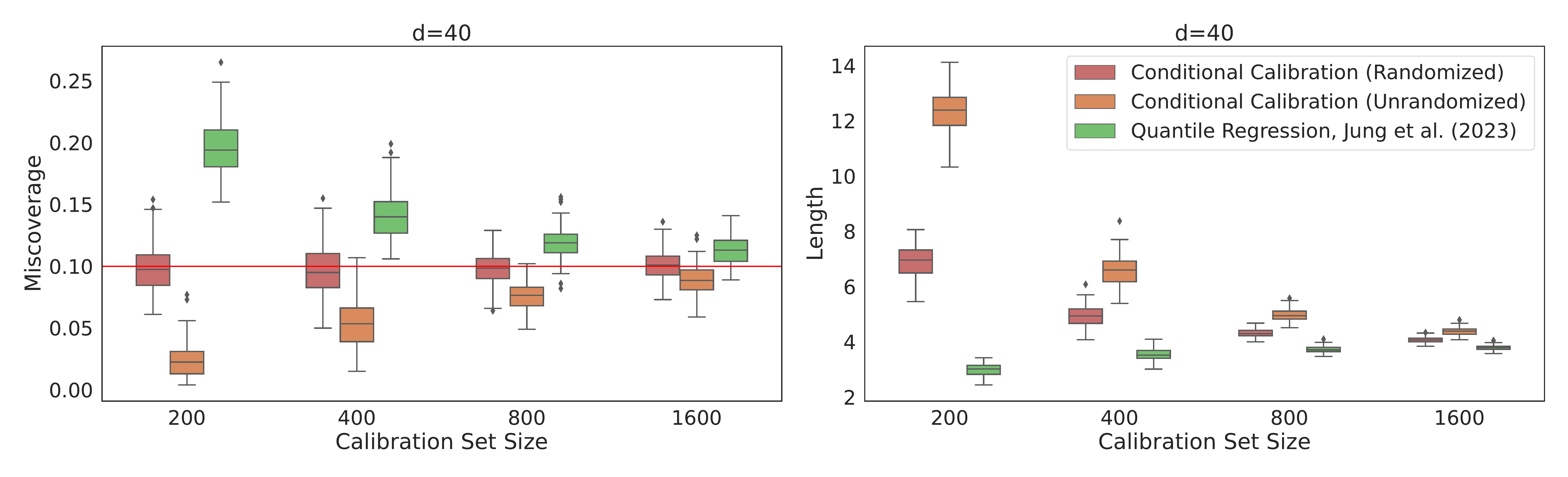}
    \caption{Marginal calibration-conditional miscoverage (left panel) and length (right panel) of quantile regression (green) and the randomized (red) and unrandomized (orange) implementations of our conditional-calibration method on a simulated dataset. All methods are implemented in their two-sided form with conformity score $S(x,y) = y$, i.e. we estimate the $\alpha/2$ and $1-\alpha/2$ quantiles of $Y \mid X$ separately and define the prediction set to be the values of $y$ that fall between the two bounds (see \Cref{sec:app_two-sided} for details). Data for this simulation are generated i.i.d. from $Y_i = X_i^\top w + \epsilon_i$ where $X_i \sim \mathcal{N}(0,I_d)$, $\epsilon_i \sim \mathcal{N}(0,1)$, and $w \sim \text{Unif}(\mathcal{S}^{d-1})$. We implement both vanilla quantile regression \citep{Jung2023} and our conditional-calibration methods on the function class $\mathcal{F} := \{\beta_0 + \sum_{i=1}^d \beta_i \bone\{x_i > 0\} : \beta \in \mmr^d\}$. Boxplots show empirical estimates obtained by averaging over 1000 test points for each of 100 calibration datasets. The red line in the left panel indicates the target coverage level of $1-\alpha = 0.9$.}
    \label{fig:cc_cov_and_length}
\end{figure}

In this paper, we improve upon existing group-conditional coverage results in two crucial aspects: (1) we obtain tighter finite sample coverage and (2) we make no assumptions on the distribution of $(X_i, Y_i)$ or the overlap of the groups $\mathcal{G}$. Concretely, given an arbitrary finite collection of groups our randomized conditional calibration method guarantees exact coverage, 
\[
\mmp(Y_{n+1} \in \hat{C}(X_{n+1}) \mid X_{n+1} \in G) = 1-\alpha,\ \forall G \in \mathcal{G},
\]
while our unrandomized procedure obeys the inequalities, 
\[
1-\alpha \leq \mmp(Y_{n+1} \in \hat{C}(X_{n+1}) \mid X_{n+1} \in G) \leq 1-\alpha + \frac{|\mathcal{G}|}{(n+1)\mmp(X \in G)},\ \forall G \in \mathcal{G}.
\]
\Cref{fig:cc_cov_and_length} shows the improved finite sample coverage of these method. For simplicity, this plot only displays the marginal miscoverage. Boxplots in the figure show the estimated distributions of the calibration-conditional miscoverage and length, i.e. the quantities 
\[
\mmp(Y_{n+1} \notin \hat{C}(X_{n+1}) \mid \{(X_i,Y_i)\}_{i=1}^n) \qquad \text{and} \qquad \mme[\, \text{length}(\hat{C}(X_{n+1})) \mid \{(X_i,Y_i)\}_{i=1}^n],
\] 
as the sample size varies. In agreement with our theory, we find that the unrandomized version of our procedure guarantees conservative coverage, while our randomized variant offers exact coverage regardless of the sample size. On the other hand, the method of \citet{Jung2023}, i.e. linear quantile regression over the set of subgroup-inclusion indicator functions, can severely undercover even at what might be considered large sample sizes.

\subsubsection{Coverage under covariate shift}

Given an appropriate choice of $\mathcal{F}$, our prediction set also achieves \emph{coverage under covariate shift}. To define this objective, fix any non-negative function $f$ and let $\mmp_{f}$ denote the setting in which $\{(X_i,Y_i)\}_{i=1}^n$ is sampled i.i.d.~from $P$, while $(X_{n+1},Y_{n+1})$ is sampled independently from the distribution in which $P_X$ is ``tilted'' by $f$, i.e., 
\[
X_{n+1} \sim \frac{f(x)}{\mme_P[f(X)]}\cdot dP_X(x),\quad Y_{n+1}\mid X_{n+1} \sim P_{Y\mid X}.
\]
Then, our method guarantees coverage under $\mmp_{f}$ so long as $f \in \mathcal{F}$. For example, when $\mathcal{F}$ is a finite-dimensional linear function class, our prediction set satisfies
\begin{align*}
   \mmp_{f} ({Y}_{n + 1} \in \hat{C}({X}_{n + 1})) = 1 - \alpha, \qquad \text{for all non-negative functions $f \in \mathcal{F}$}.
\end{align*}

There have been a number of previous works in this area that also establish coverage under covariate shift. However, these works assume that there is a \emph{single} covariate shift of interest that is either known a-priori (\cite{Tibs2019}), or estimated from unlabeled data (\cite{Qiu2022, yang2022doubly}). Our method captures this setting as a special case in which $\mathcal{F}$ is chosen to be a singleton. On the other hand, when $\mathcal{F}$ is non-singleton, our guarantee is more general and ensures coverage over all shifts $f \in \mathcal{F}$, simultaneously.

\begin{figure}[ht]
    \centering
    \includegraphics[width=\textwidth]{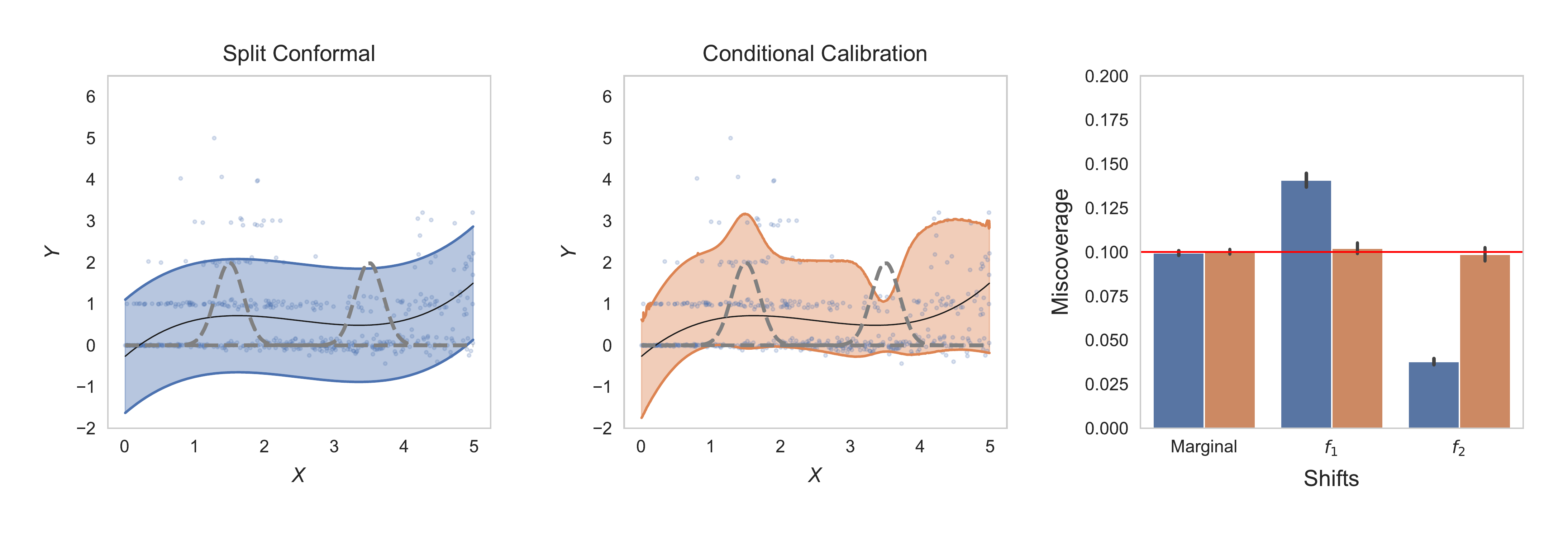}
    \caption{Comparison of split conformal prediction (blue, left-most panel) and the randomized implementation of our method (orange, center panel) on a simulated dataset first considered by \citet{Romano2019}. Black curves denote an estimate of the conditional mean, while the blue and orange shaded regions indicate the fitted prediction intervals. We consider coverage under three scenarios; marginally over the whole x-axis and locally under two Gaussian tilts (denoted by $f_1$ and $f_2$ and plotted as grey dotted lines). For this experiment, our method is implemented using the procedure outlined in \Cref{sec:finite_dim} with $\mathcal{F} := \{\beta_0 + \sum_{i = 1}^5 \beta_i w_i(x) : \beta \in \mmr^6\}$, where the $w_i$ corresponds to the Gaussian tilts with parameters $(\mu,\sigma) \in \{(0.5,1),(1.5,0.2),(2.5,1),(3.5,0.2),(4.5,1)\}$. The rightmost panel indicates the miscoverage of both methods under all three settings with a red line denoting the target level of $\alpha = 0.1$.}
    \label{fig:cqr_data_2}
\end{figure}

\Cref{fig:cqr_data_2} illustrates a simple example of this guarantee on a synthetic dataset. Once again, the covariate $x$ is a scalar and the conformity score is taken to be $S(x,y) = |\hat{\mu}(x) - y|$; following the previous example, we implement the two-sided version of our method. We consider five covariate shifts in which $P_X$ is tilted by the Gaussian density $f_{\mu,\sigma}(x) = \exp(-\frac{1}{2\sigma^2}(x-\mu)^2)$ for $(\mu,\sigma) \in \{(0.5,1),(1.5,0.2),(2.5,1),(3.5,0.2),(4.5,1)\}$ . The shifts centered at $1.5$ and $3.5$ are plotted in grey and denoted as $f_1$ and $f_2$ in the figure. As the left-most panels show, split conformal gives a constant width prediction band over the entire x-axis, while our method, adapts to the shape of the data around the covariate shifts. The right-most panel of the figure validates that this correction is sufficient to expand the marginal coverage guarantee of split conformal inference to exact coverage under all three scenarios: no shift, shift by $f_1$, and shift by $f_2$.

\begin{figure}[H]
    \centering
    \includegraphics[width=0.8\textwidth]{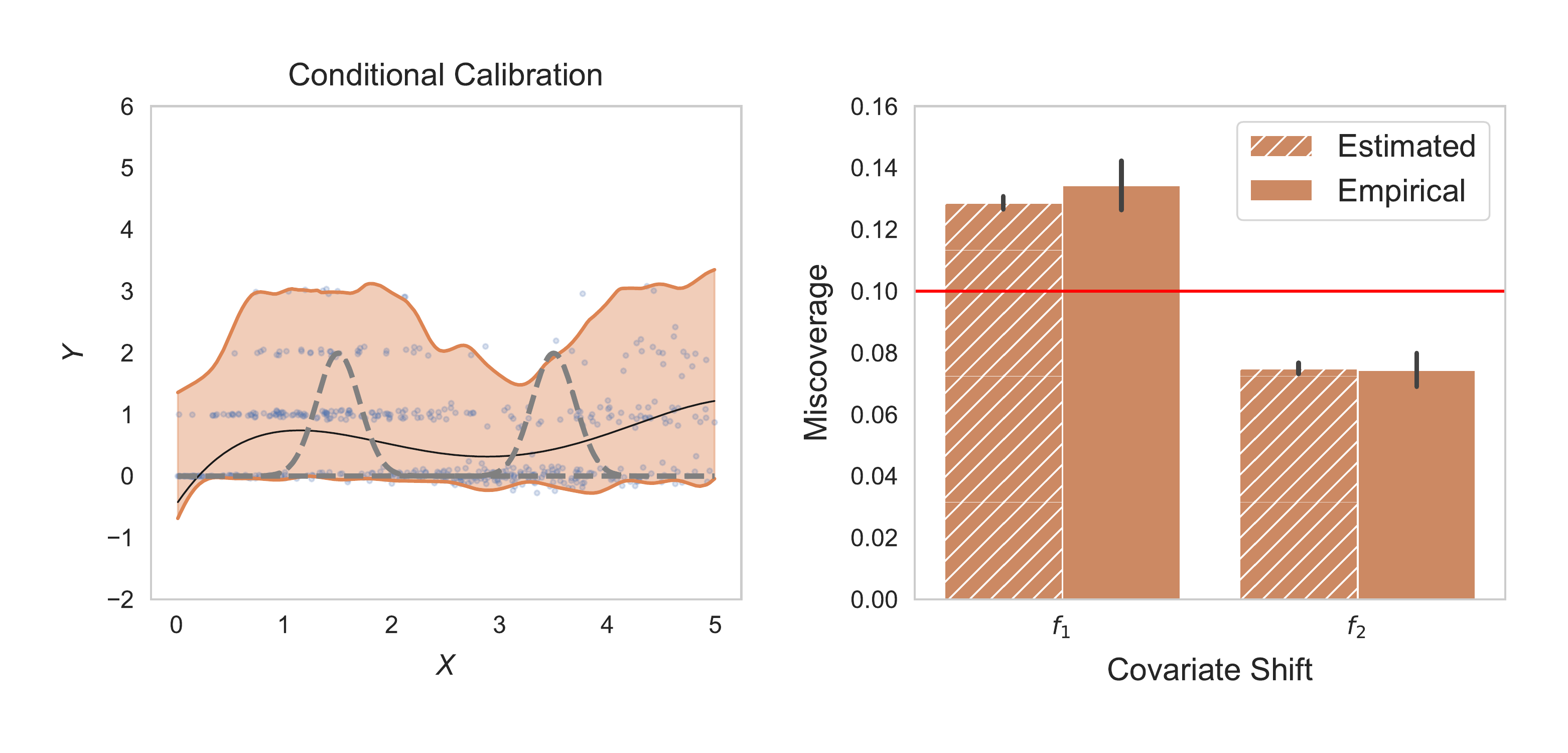}
    \caption{Demonstration of the unrandomized implementation of our shift-agnostic method on a simulated dataset first considered by \citet{Romano2019}. The orange shaded region in the left panel depicts the prediction interval output by our method when $\mathcal{F}$ is chosen to be the Gaussian reproducing kernel Hilbert space given by kernel $K(x,y) = \exp(-12.5|x-y|^2)$ and the hyperparameter $\lambda$ is set equal to $0.005$ (see \Cref{sec:infinite_dim} for details). Hatched and solid bars in the right panel show the estimated coverage returned by our method and the true realized empirical coverage, respectively. Finally, the red line indicates the target level of $\alpha = 0.1$.}
    \label{fig:cqr_data_3}
\end{figure}

Extending even further beyond existing approaches, our method can also provide guarantees even when no prior information is known about the shift. By fitting with a so-called ``universal'' function class, e.g., $\mathcal{F}$ is a suitable reproducing kernel Hilbert space, we provide coverage guarantees under \emph{any} covariate shift. Due to the complexity of these classes, our coverage guarantee is no longer exactly $1 - \alpha$ for all tilts $f$. Instead, in its place, we obtain an exact estimate of the (improved) finite-sample coverage of our method. For example, as seen in \Cref{fig:cqr_data_3}, if we run our shift-agnostic method for the two plotted Gaussian tilts, our estimated coverage precisely matches the empirically observed values. Thus, even though we cannot guarantee exact $1-\alpha$ coverage for all shifts in this example, we are still able to accurately report the true performance that users can expect from our method. For more information on how these estimates are computed we refer the interested reader to \Cref{sec:infinite_dim} and, in particular, to \eqref{eq:rkhs_cov} and \Cref{prop:rkhs_inner_prod_est}.

\subsection{Outline}

The remainder of this article is structured as follows. In \Cref{sec:finite_dim}, we introduce our method and give coverage results for the case in which $\mathcal{F}$ is finite dimensional. These results are expanded on in \Cref{sec:infinite_dim}, where we consider infinite dimensional classes. Computational difficulties that arise in both the finite and infinite dimensional cases are addressed in \Cref{sec:computation}, and an efficient implementation of our method is given. In \Cref{sec:real_data}, we apply our method to two datasets and show that our approach attains tighter finite-sample conditional coverage and more predictable failure modes than competing alternatives. Finally, we conclude in Section \ref{sec:fun_choice} with a discussion of considerations that arise when choosing the function class (and associated regularization) to use in our method.

A Python package, \textsf{conditionalconformal}, implementing our methods is available on PyPI, and notebooks reproducing the experimental results in this paper can be found at \href{https://www.github.com/jjcherian/conditional-conformal}{github.com/jjcherian/conditional-conformal}.

\textbf{Notation}: In this paper, we consider two settings. In the first, we take $\{(X_i, Y_i)\}_{i = 1}^{n + 1} \stackrel{i.i.d.}{\sim} P$. In the second, $\{(X_i, Y_i)\}_{i = 1}^{n} \stackrel{i.i.d.}{\sim} P$, while $(X_{n + 1}, Y_{n + 1})$ is sampled independently from the tilted distribution $X_{n+1} \sim (f(x)/\mme_P[f(X)]) \cdot dP_X(x)$ and $Y_{n+1} \mid X_{n+1} \sim P_{Y \mid X}$. We write $\mmp$ and $\mme$ with no subscript when referring to the first scenario, while we use the subscript $f$ to denote the second. Additionally, note that throughout this article we use $(\mathcal{X},\mathcal{Y})$ to denote the domain of the $(X,Y)$ pairs.
 
\section{Protection against finite dimensional shifts}\label{sec:finite_dim}
\subsection{Warm-up: marginal coverage} \label{sec:warm-up}
As a starting point to motivate our approach, we show that split conformal prediction is a special case of our method. Before explaining this result in detail, it is useful to first review the details of the split conformal algorithm. 

Recall that the conformity score function, $S : \mathcal{X} \times \mathcal{Y} \to \mmr$ measures how well the prediction of some model at $X$ ``conforms'' to the target $Y$. For instance, given an estimate $\hat{\mu}(\cdot)$ of $\mme[Y \mid X]$, we may take $S(\cdot,\cdot)$ to be the absolute residual $S(x,y) = |y - \hat{\mu}(x) |$. In a typical implementation of split conformal, we would need to split the training data $\{(X_i,Y_i)\}_{i=1}^n$ into two parts, using one part to train $\hat{\mu}$ and reserving the second part as the calibration set. Because our method provides the same coverage guarantees regardless of the initial choice of $S(\cdot,\cdot)$, we will not discuss this first step in detail. Instead, we will assume that the conformity score function is fixed, and we are free to use the entire dataset $\{(X_i,Y_i)\}_{i=1}^n$ to calibrate the scores. In practice, the initial step of fitting the conformity score can be critical for getting a good baseline predictor. For example, in our experiment in \Cref{sec:cell_data} , we use a neural network to obtain an initial prediction of the genetic treatment given to cells based on fluorescent microscopy images.

Given a conformity score function and calibration set, the split conformal algorithm outputs the set of values $y$ for which $S(X_{n+1},y)$ is sufficiently small, i.e., the set of values $y$ that conform with the prediction at $X_{n+1}$. The threshold for this prediction set, which we denote by $S^*$, is set to be the $(\lceil (n + 1) \cdot (1 - \alpha) \rceil/n)$-quantile of the conformity scores evaluated on the calibration set. In summary, the split conformal prediction set is formally defined as 
\begin{equation}\label{eq:split_conformal_set}
\hat{C}_{\text{split}}(X_{n + 1}) = \{y : S(X_{n + 1}, y) \leq S^*\}.
\end{equation}

\sloppy The standard method for proving the marginal coverage of $\hat{C}_{\text{split}}(\cdot)$ is to appeal to the exchangeability of the conformity scores. Namely, let $S_1,\dots,S_{n+1}$ denote the scores $S(X_1,Y_1),\dots,S(X_{n+1},Y_{n+1})$. Since the $(n + 1)$-th conformity score is drawn i.i.d.~from the same distribution as the first $n$ scores, the location of $S_{n + 1}$ among the order statistics of $(S_1,\dots,S_{n+1})$ is drawn uniformly at random from each of the $n + 1$ possible indices. So, recalling that $S^*$ is chosen to be the the $(\lceil (n + 1) \cdot (1 - \alpha) \rceil/n)$-quantile, i.e., the smallest order statistic satisfying $\mmp(S_{n + 1} \leq S^*) \geq 1 - \alpha$, we arrive at the coverage guarantee $\mmp(Y_{n+1} \in \hat{C}_{\text{split}}(X_{n + 1})) = \mmp(S_{n + 1} \leq S^*) \geq 1 - \alpha$. The following theorem summarizes the formal consequences of these observations.

\begin{theorem}[\citet{Romano2019}, Theorem 1, see also \citet{VovkBook}]\label{thm:split_marginal_cov}
Assume that $\{(X_i,Y_i)\}_{i=1}^{n+1}$ are independent and identically distributed. Then, the split conformal prediction set \eqref{eq:split_conformal_set} satisfies,
\[
\mmp(Y_{n+1} \in \hat{C}_{\textup{split}}(X_{n+1})) \geq 1-\alpha.
\]
If $S(X_{n+1},Y_{n+1})$ has a continuous distribution, it also holds that
\[
\mmp(Y_{n+1} \in \hat{C}_{\textup{split}}(X_{n+1})) \leq 1-\alpha + \frac{1}{n + 1}.
\]
\end{theorem}

Our first insight is that this prediction set and marginal coverage guarantee can also be obtained by re-interpreting split conformal as an intercept-only quantile regression. Recall the definition of the ``pinball'' loss,
\begin{align*}
\ell_{\alpha}(\theta,S) := \begin{cases}
(1-\alpha)(S - \theta)\ \text{ if } S \geq \theta,\\
\alpha(\theta - S) \ \text{ if } S < \theta.
\end{cases}
\end{align*}
It is well-known that minimizing $\ell_{\alpha}$ over $\theta$ will produce a $(1-\alpha)$-quantile of the training data, i.e., $\theta^* = \argmin_{\theta \in \mathbb{R}} \sum_{i = 1}^n \ell_\alpha(\theta, S_i)$ is a $(1-\alpha)$-quantile of $\{S_i\}_{i = 1}^n$ (\cite{koenker1978regression}). 

In our exchangeability proof, recall that we upper bounded $S_{n + 1}$ not by the $(1 - \alpha)$-quantile of $\{S_i\}_{i = 1}^n$, but by the $(\lceil (n + 1) \cdot (1 - \alpha) \rceil/n)$-quantile. The latter value was obtained by considering an augmented dataset that included all of the scores $S_1,\dots,S_{n+1}$. To similarly account for the (unobserved) conformity score in a quantile regression, we will now fit $\theta^*$ using a dataset that includes a guess for $S_{n + 1}$. Namely, let $\hat{\theta}_S$ be a solution to the quantile regression problem in which we impute $S$ for the unknown conformity score, i.e., $$\hat{\theta}_{S} := \argmin_{\theta \in \mathbb{R}} \frac{1}{n + 1} \sum_{i = 1}^n \ell_{\alpha}(\theta, S_i) + \frac{1}{n + 1}  \ell_\alpha(\theta, S).$$
Then, one can verify that 
\begin{equation}\label{eq:split_as_quant}
    \hat{C}_{\text{split}}(X_{n + 1}) = \{y : S_{n + 1}(X_{n + 1}, y) \leq \hat{\theta}_{S_{n + 1}(X_{n + 1}, y)}\},
\end{equation}
or said more informally, $\hat{C}_{\text{split}}(X_{n + 1})$ includes any $y$ such that $S(X_{n + 1}, y)$ is smaller than the $(1 - \alpha)$-quantile of the augmented calibration set $\{S_i\}_{i = 1}^n \cup \{S(X_{n + 1}, y)\}$. As an aside, we note there is some small subtlety here due to the non-uniqueness of $\hat{\theta}_S$. To get exact equality in \eqref{eq:split_as_quant}, one should choose $\hat{\theta}_S$ to be the largest minimizer of the quantile regression. Readers familiar with conformal inference will also recognize this method of imputing a guess for the missing $(n+1)$-th datapoint as a type of full conformal prediction (\cite{VovkBook}).

Having established that split conformal prediction can be derived via quantile regression, our generalization of this procedure to richer function classes naturally follows. Namely, we will replace the single score threshold $\theta$ with a function $f(X)$ that estimates the conditional quantiles of $Y \mid X$. We then prove a generalization of \Cref{thm:split_marginal_cov} showing that the resulting prediction set attains a conditional coverage guarantee. 
\subsection{Finite dimensional classes} \label{sec:finite_dim_main_result}

Recall our objective:
\begin{equation}\label{eq:cov_goal}
\mme[f(X_{n+1})(\bone\{Y_{n+1} \in \hat{C}(X_{n+1})\} - (1-\alpha))] = 0,\ \forall f \in \mathcal{F}.
\end{equation}
In the previous section, we constructed a prediction set with marginal coverage, i.e., \eqref{eq:cov_goal} for $\mathcal{F} = \{\theta : \theta \in \mathbb{R}\}$, by fitting an augmented quantile regression over \textit{the same} function class $\mathcal{F} = \{\theta : \theta \in \mathbb{R}\}$. Here, we generalize this observation to any finite-dimensional linear class. 

To formally define our method, let $\mathcal{F} = \{\Phi(\cdot)^\top\beta : \beta \in \mmr^d\}$ denote the class of linear functions over the basis $\Phi: \mathcal{X} \to \mathbb{R}^d$. Our goal is to construct  a $\hat{C}$ satisfying \eqref{eq:cov_goal} for this choice of $\mathcal{F}$. Imitating our re-derivation of split conformal prediction, we define the augmented quantile regression estimate $\hat{g}_S$ as 
\begin{equation}\label{eq:finite_dim_reg}
\hat{g}_S := \argmin_{g \in \mathcal{F}} \frac{1}{n+1}\sum_{i=1}^n \ell_{\alpha}(g(X_i),S_i) + \frac{1}{n+1}  \ell_{\alpha}(g(X_{n+1}),S).
\end{equation}
Then, we take our prediction set to be 
\begin{equation}\label{eq:finite_dim_pred_set}
\hat{C}(X_{n+1}) := \{y : S(X_{n+1},y) \leq \hat{g}_{S(X_{n+1},y)}(X_{n+1})\}.
\end{equation}
Critically, we emphasize that $\hat{g}_S$ is fit using the same function class $\mathcal{F}$ that appears in our coverage target. This fact will be crucial to the theoretical results that follow. To keep our notation clear under this recycling of $\mathcal{F}$, we will always use $g$ to denote quantile estimates and $f$ to denote re-weightings.

Before discussing the coverage properties of this method there are two technical issues that we must address. First, astute readers may have noticed that \eqref{eq:finite_dim_pred_set} appears to be intractable. Indeed, a naive computation of $\hat{C}(X_{n+1})$ would require us to compute $\hat{g}_S$ for all $S \in \mmr$. In \Cref{sec:computation}, we will give an efficient algorithm for computing the prediction set that overcomes this naive approach. To ease exposition we defer the details of this method for now. The second issue that we must address is the non-uniqueness of the estimate $\hat{g}_S$. In all subsequent results of this article, we will assume that $\hat{g}_S$ is computed using an algorithm that is invariant under re-orderings of the input data. This assumption is relevant because quantile regression can admit multiple optima; in theory, the selected optimum might systematically depend on the indices of the scores. In practice, this assumption is inconsequential because any commonly used algorithm, e.g., an interior point solver, satisfies this invariance condition.

With these issues out of the way, we are now ready to state the main result of this section, \Cref{thm:finite_dim_result}, which summarizes the coverage properties of \eqref{eq:finite_dim_pred_set}. When interpreting this theorem it may be useful to recall that for non-negative $f$, $\mmp_f(\cdot)$ denotes the setting in which $(X_1,Y_1),\dots,(X_{n},Y_{n}) \stackrel{i.i.d.}{\sim} P$, while $(X_{n + 1}, Y_{n + 1})$ is sampled independently from $X_{n+1} \sim \frac{f(x)}{\mme_P[f(X)]}dP_X(x)$ and $Y_{n+1} \mid X_{n+1} \sim P_{Y \mid X}$.

\begin{theorem}\label{thm:finite_dim_result}
    Let $\mathcal{F} = \{\Phi(\cdot)^\top\beta : \beta \in \mmr^d\}$ denote the class of linear functions over the basis $\Phi: \mathcal{X} \to \mathbb{R}^d$. Then, for any non-negative $f \in \mathcal{F}$ with $\mme_P[f(X)] > 0$, the prediction set given by \eqref{eq:finite_dim_pred_set} satisfies
    \begin{equation}\label{eq:finite_dim_pos_func_cov}
    \mmp_f(Y_{n+1} \in \hat{C}(X_{n+1})) \geq 1-\alpha.
    \end{equation}
    On the other hand, if $(X_1, Y_1), \dots, (X_{n + 1}, Y_{n + 1}) \stackrel{i.i.d.}{\sim} P$ and the distribution of $S \mid X$ is continuous, then for all $f \in \mathcal{F}$, we additionally have the two-sided bound,
    \[
    \left|\mme[f(X_{n+1})(\bone\{Y_{n+1} \in \hat{C}(X_{n+1})\} - (1-\alpha))] \right| \leq \frac{d}{n+1}\mme \left [\max_{1 \leq i \leq n+1}|f(X_{i})| \right ].
    \]   
\end{theorem}

This type of two-part result is typical in conformal inference. Namely, while the assumption that the distribution of $S \mid X$ is continuous may seem overly restrictive, it is standard in conformal inference that upper bounds require a mild continuity assumption, while lower bounds are fully distribution-free. For example, the canonical coverage guarantee for split conformal described in \Cref{thm:split_marginal_cov} also gives separate upper and lower bounds for continuous and discrete data. Notably, in the case of split conformal inference this two-part result can be avoided and replaced by an exact $1 - \alpha$ coverage guarantee by randomizing the prediction set. We will show in \Cref{sec:computation} that an analogous result also holds for our method: without any assumptions on the continuity of $S \mid X$, we show that randomizing $\hat{C}(X_{n+1})$ yields $\mme[f(X_{n+1})(\bone\{Y_{n+1} \in \hat{C}(X_{n+1})\} - (1-\alpha))]  = 0$  for all $f \in \mathcal{F}$.  Because the randomization scheme we employ leverages the  algorithms developed in \Cref{sec:computation}, we defer a precise statement of this result for now.


Our next result, \Cref{cor:group_coverage} relates the more abstract guarantee of \Cref{thm:finite_dim_result} to the group-conditional coverage example previewed in the introduction.

\begin{corollary}\label{cor:group_coverage}
    Suppose $\{(X_i,Y_i)\}_{i=1}^{n+1}$ are independent and identically distributed and the prediction set given by \eqref{eq:finite_dim_pred_set} is implemented with $\mathcal{F} = \{x \mapsto \sum_{G \in \mathcal{G}}\beta_G\bone\{x \in G\} : \beta_G \in \mmr,\ \forall G \in \mathcal{G}\}$ for some finite collection of groups $\mathcal{G} \subseteq 2^{\mathcal{X}}$. Then, for any $G \in \mathcal{G}$,
    \[
    \mmp(Y_{n+1} \in \hat{C}(X_{n+1}) \mid X_{n+1} \in G) \geq 1-\alpha.
    \]
    If the distribution of $S \mid X$ is continuous, then we have the matching upper bound,
    \[
    \mmp(Y_{n+1} \in \hat{C}(X_{n+1}) \mid X_{n+1} \in G) \leq 1-\alpha + \frac{|\mathcal{G}|}{(n+1) \cdot \mathbb{P}(X_{n + 1} \in G)}.
    \]
\end{corollary}

The methods described above only estimate the upper $(1-\alpha)$-quantile of the conformity score. If desired, our procedure can also be generalized to give both lower and upper bounds on $S(X_{n+1},Y_{n+1})$. In particular, letting $\hat{g}_S^\tau(\cdot)$ denote our estimate of the $\tau$-th quantile, we can define the two-sided prediction set 
\begin{equation}\label{eq:two-sided_set}
\hat{C}_{\text{two-sid.}}(X_{n+1}) := \{y : \hat{g}_{S_{n+1}(X_{n+1},y)}^{\alpha/2}(X_{n + 1}) \leq S_{n+1}(X_{n+1},y) \leq \hat{g}_{S_{n+1}(X_{n+1},y)}^{1-\alpha/2}(X_{n + 1})\}.
\end{equation}
As an example of this, \Cref{fig:cqr_data,fig:cqr_data_2,fig:cqr_data_3} show results from an implementation of our method in which we fit the lower and upper quantiles of $y - \hat{\mu}(x)$ separately. Because these two-sided prediction sets have identical coverage properties to their one-sided analogues we will for simplicity focus in the remainder of this article on the one-sided version. Readers interested in the two-sided instantiation of our approach should see \Cref{sec:app_two-sided} for additional information about the implementation and formal coverage guarantees of these methods. 

We conclude this section by giving a brief proof sketch of \Cref{thm:finite_dim_result}, leaving formal details to the Appendix. The main idea is to examine the first order conditions of the quantile regression \eqref{eq:finite_dim_reg} and then exploit the fact that this regression treats the test point identically to the calibration data. This connection between the derivative of the pinball loss and coverage was first made by \citet{Jung2023}.

\begin{proof}[Proof sketch of \Cref{thm:finite_dim_result}]
We examine the first order conditions of \eqref{eq:finite_dim_reg}. By a direct computation we have that for any $f \in \mathcal{F}$,
\[
\frac{d}{d\epsilon} \ell_{\alpha}(\hat{g}_S(X) + \epsilon f(X_i),S_i) \bigg|_{\epsilon = 0} = \begin{cases}
-(1-\alpha)f(X_i), \text{ if } S_i > \hat{g}_S(X_i),\\
\alpha f(X_i), \text{ if } S_i < \hat{g}_S(X_i),\\
\text{undefined, if } S_i = \hat{g}_S(X_i).
\end{cases}
\]
For simplicity, suppose that for all $i$, $S_i \neq \hat{g}_S(X_i)$. This assumption does not hold in general and by adding it here we will obtain the stronger result $\mme[f(X_{n+1})(\bone\{Y_{n+1} \in \hat{C}(X_{n+1})\} - (1-\alpha))] = 0$. In the full proof of \Cref{thm:finite_dim_result}, we remove this simplification and incur an additional error term. 

For now, making this assumption gives the first order condition
\begin{align*}
& \frac{1}{n+1} \sum_{i =1}^{n+1} \alpha f(X_i)\bone\{S_i \leq \hat{g}_{S_{n+1}}(X_i)\} - (1-\alpha) f(X_i)\bone\{S_i > \hat{g}_{S_{n+1}}(X_i)\} = 0\\
& \iff \frac{1}{n+1} \sum_{i =1}^{n+1}  f(X_i)(\bone\{S_i \leq \hat{g}_{S_{n+1}}(X_i)\} - (1-\alpha)) = 0.
\end{align*}
Taking expectations, we arrive at our coverage guarantee
\begin{align*}
\mme[f(X_{n+1})(\bone\{Y_{n+1} \in \hat{C}(X_{n+1})\} - (1-\alpha))] & = \mme[f(X_{n+1})(\bone\{S_{n+1} \leq \hat{g}_{S_{n+1}}(X_{n+1})\} - (1-\alpha))]\\
& = \mme[\frac{1}{n+1} \sum_{i=1}^{n+1} f(X_{i})(\bone\{S_{i} \leq \hat{g}_{S_{n+1}}(X_{i})\} - (1-\alpha))]\\
& = 0,
\end{align*}
where the first equality uses the definition of $\hat{C}(X_{n+1})$ and the second equality applies the fact that the triples $(X_1,S_1,\hat{g}_{S_{n+1}}(X_{1})),\dots,(X_{n+1},S_{n+1},\hat{g}_{S_{n+1}}(X_{n+1}))$ are exchangeable.

\end{proof}

\subsection{Related work}

The method proposed above can be briefly summarized as a modified quantile regression procedure in which the new test point is incorporated into the fit. Given the popularity of vanilla quantile regression, one might reasonably ask how this compares to the more standard approach in which one fits
\[
\hat{g}_{\text{qr}} := \argmin_{g \in \mathcal{F}} \frac{1}{n} \sum_{i=1}^n \ell_{\alpha}(g(X_i),S_i),
\]
on the training data and then forms the prediction set,
\[
\hat{C}_{\text{qr}}(X_{n+1}) := \{y : S(X_{n+1},y) \leq \hat{g}_{\text{qr}}(X_{n+1})\}.
\]

\citet{Jung2023} analyze this approach in the case where $\mathcal{F} := \{\sum_{i=1}^d \beta_i \bone\{x \in G_i\} : \beta \in \mmr^d\}$ is the space of linear combinations of subgroup indicator functions. Under appropriate assumptions on the distribution of $(X_i,S(X_i,Y_i))$, they show that for all groups, $1 \leq j \leq d$ and constants $\delta > 0$ this prediction set satisfies the PAC coverage guarantee,
\[
\mmp\left(\left| \mmp\left(Y_{n+1} \in \hat{C}_{\text{qr}}(X_{n+1}) \mid X_{n+1} \in G_j,\{(X_i,Y_i)\}_{i=1}^n \right) - (1-\alpha)  \right| \leq O\left( \left(\frac{\log(1/\delta) + d\log(n)}{n \mmp(X_{n+1} \in G_j)^2} \right)^{1/4}\right) \right) \geq 1-\delta.
\]
On the other hand, in the same setting, \Cref{cor:group_coverage} states the following guarantee for our method,
\[
1-\alpha \leq \mmp(Y_{n+1} \in \hat{C}(X_{n+1}) \mid X_{n+1} \in G_j) \leq 1-\alpha + \frac{d}{(n+1)\mmp(X_{n+1}\in G_j)}.
\]
The proofs of both of these results are based on the first order conditions of quantile regression which, as we showed above, can be exploited to guarantee conditional coverage. That said, the approaches differ substantially and the final results are not directly comparable. At a high level, we find that the former result targets a stronger notion of coverage, ensuring concentration conditional on the calibration data, while the second result has a much faster convergence rate ($d/n$ versus $(d\log(n)/n)^{1/4}$). 

While these methods are difficult to compare theoretically, a much clearer picture emerges on simulated data. In particular, the results in \Cref{fig:cc_cov_and_length} show that our method is far more robust to small sample-sizes or large dimensions. Although we do not provide a PAC guarantee, we find that the coverage of our procedure concentrates tightly at the target level across a wide range of values for $d/n$. On the other hand, the vanilla quantile regression approach taken in \citet{Jung2023} shows notable undercoverage at moderate dimensions, e.g., $d/n \in \{0.05, 0.1\}$.

\section{Extension to infinite dimensional classes}\label{sec:infinite_dim}

Turning back to our method, we now consider settings in which we do not have a small, finite-dimensional function class of interest. In particular, if we view the coverage target \eqref{eq:cov_goal} as an interpolation between marginal and conditional coverage, then it is natural to ask what guarantees can be provided when $\mathcal{F}$ is a rich, and potentially even infinite dimensional, function class. We know from previous work that exact coverage over an arbitrary infinite dimensional class is impossible (\cite{Vovk2012, Barber2020}). Thus, just as we relaxed the definition of conditional coverage above, here we will construct prediction sets that satisfy a relaxed version of \eqref{eq:cov_goal}.

First, note that we cannot directly implement our method over an infinite dimensional class. Indeed, running quantile regression in dimension $d \geq n+1$ will simply interpolate the input data. In our context, this means that every value $S \in \mmr$ will satisfy $S = \hat{g}_S(X_{n+1})$ and our method will always output $\hat{C}(X_{n+1}) = \mmr$. To circumvent this issue and obtain informative prediction sets, we must add regularization. This leads us to the definition
\begin{equation}\label{eq:general_method}
\hat{g}_S := \argmin_{g\in\mathcal{F}} \frac{1}{n+1} \sum_{i=1}^{n} \ell_{\alpha}(g(X_i),S_i) + \frac{1}{n+1}  \ell_{\alpha}(g(X_{n+1}),S) +  \mathcal{R}(g),
\end{equation}
for some appropriately chosen penalty $\mathcal{R}(\cdot)$. Having made this adjustment, we may now proceed identically to the previous section. Namely, we set 
\begin{equation}\label{eq:infinite_dim_set}
\hat{C}(X_{n+1}) := \{y : S(X_{n+1},y) \leq \hat{g}_{S(X_{n+1},y)}(X_{n+1})\}, 
\end{equation}
and by examining the first order conditions of \eqref{eq:general_method}, we obtain the following generalization of \Cref{thm:finite_dim_result}.
\begin{theorem}\label{thm:infinite_dim_result}
    Let $\mathcal{F}$ be any vector space, and assume that for all $f,g \in \mathcal{F}$, the derivative of $\epsilon \mapsto \mathcal{R}(g + \epsilon f)$ exists. If $f$ is non-negative with $\mme_P[f(X)] > 0$, then the prediction set given by \eqref{eq:infinite_dim_set} satisfies the lower bound
    \[
    \mmp_f(Y_{n+1} \in \hat{C}(X_{n+1})) \geq 1- \alpha  - \frac{1}{\mme_P[f(X)]} \mme\left[  \frac{d}{d\epsilon} \mathcal{R}(\hat{g}_{S_{n+1}} + \epsilon f) \bigg|_{\epsilon = 0} \right]. \]
    On the other hand, suppose $(X_1,Y_1),\dots,(X_{n+1},Y_{n+1}) \stackrel{i.i.d.}{\sim} P$. Then, for all $f \in \mathcal{F}$, we additionally have the two-sided bound,
    \begin{equation}    \label{eq:infinite_dim_cov}
    \begin{split}
    \mme[f(X_{n+1})(\bone\{Y_{n+1} \in \hat{C}(X_{n+1})\} - (1-\alpha))] = - \mme\left[ \frac{d}{d\epsilon} \mathcal{R}(\hat{g}_{S_{n+1}} + \epsilon f) \bigg|_{\epsilon = 0} \right] + \epsilon_{\textup{int}},
    \end{split}
    \end{equation}
    where $\epsilon_{\textup{int}}$ is an interpolation error term satisfying $|\epsilon_{\textup{int}}| \leq \mme[|f(X_i)| \bone\{S_i = \hat{g}_{S_{n+1}}(X_i)\}]$.
\end{theorem}
Similar to our results in the previous section, the interpolation term $\epsilon_{\text{int}}$ can be removed if we allow the prediction set to be randomized (see \Cref{sec:computation} for a precise statement).

To more accurately interpret \Cref{thm:infinite_dim_result}, we will need to develop additional understanding of the two quantities appearing on the right-hand side of \eqref{eq:infinite_dim_cov}. The following two sections are devoted to this task for two different choices of $\mathcal{F}$. At a high level, the results in these sections will show that the interpolation error $\epsilon_{\text{int}}$ is of negligible size and thus the coverage properties of our method are primarily governed by the derivative term $- \mme[ \frac{d}{d\epsilon} \mathcal{R}(\hat{g}_{S_{n+1}} + \epsilon f) |_{\epsilon = 0} ]$. Informally, we interpret this derivative as providing a quantitative estimate of the difficulty of achieving conditional coverage in the direction $f$. More practically, we will see that the derivative can be used to obtain accurate estimates of the coverage properties of our method. Critically, these estimates adapt to both the specific choice of tilt $f$ and the distribution of $(X,Y)$. Thus, they determine the difficulty of obtaining conditional coverage in a way that is specific to the dataset at hand.

\subsection{Specialization to functions in a reproducing kernel Hilbert space} \label{sec:rkhs_specialization}

Our first specialization of \Cref{thm:infinite_dim_result} is to the case where $\mathcal{F}$ is constructed using functions from a reproducing kernel Hilbert space (RKHS). More precisely, let $K : \mathcal{X} \times \mathcal{X} \to \mmr$ be a positive definite kernel and $\mathcal{F}_K$ denote the associated RKHS with inner product $\langle \cdot, \cdot \rangle_K$ and norm $\|\cdot\|_K$. Let $\Phi : \mathcal{X} \to \mmr^d$ denote any finite dimensional feature representation of $\mathcal{X}$. Then, we consider implementing our method with function class $\mathcal{F} = \{f_K(\cdot) + \Phi(\cdot)^\top\beta : f_K \in \mathcal{F}_K,\ \beta \in \mmr^d\}$ and penalty $\mathcal{R}(f_K(\cdot) + \Phi(\cdot)^\top\beta) = \lambda \|f_K\|_K^2$. Here, $\lambda > 0$ is a hyperparameter that controls the flexibility of the fit. For now, we take this hyperparameter to be fixed, although later in our practical experiments in \Cref{sec:candc_data}, we will choose it by cross-validation. 

Some examples of RKHSes that may be of interest include the space of radial basis functions given by $K(x,y) = \exp(-\gamma \|x - y\|_2^2)$, which allows us to give coverage guarantees over localizations of the covariates, and the polynomial kernel $K(x,y) = (x^\top y + c)^m$ for $m \in \mmn$, $c \geq 0$, which allows us to investigate coverage over smooth polynomial re-weightings. Additional examples and background material on reproducing kernel Hilbert spaces can be found in \citet{Paulsen2016}.

To obtain a coverage guarantee for $\mathcal{F}$, we must understand the two terms appearing on the right-hand side of \eqref{eq:infinite_dim_cov}. Let $f(\cdot) =  f_K(\cdot) + \Phi(\cdot)^\top\beta$ denote the re-weighting of interest and $\hat{g}_{S_{n+1}}(\cdot) =  \hat{g}_{S_{n+1},K}(\cdot) + \Phi(\cdot)^\top\hat{\beta}_{S_{n+1}}$ denote the fitted quantile estimate. Then, a short calculation shows that $\mme[\frac{d}{d\epsilon} \mathcal{R} (\hat{g}_{S_{n+1}} + \epsilon f)|_{\epsilon = 0} ]  = 2\lambda \mme[\langle \hat{g}_{S_{n+1},K}, f_K \rangle]$. So, applying \Cref{thm:infinite_dim_result}, we find that for all non-negative $f \in \mathcal{F}$ with $\mme_P[f(X)]>0$,
\begin{equation}\label{eq:rkhs_cov}
\begin{split}
& \mmp_f( Y_{n+1} \in \hat{C}(X_{n+1}) ) \geq 1-\alpha- 2\lambda\frac{ \mme[\langle \hat{g}_{S_{n+1},K}, f_K \rangle_K] }{\mme_P[f(X)]},\\
\text{and } & \mmp_f( Y_{n+1} \in \hat{C}(X_{n+1}) ) \leq  1-\alpha- 2\lambda \frac{\mme[\langle \hat{g}_{S_{n+1},K}, f_K \rangle_K]}{\mme_P[f(X)]} +\frac{\epsilon_{\text{int}}}{\mme_P[f(X)]}. 
\end{split}
\end{equation}
Controlling the interpolation error is more challenging and is done by the following proposition.
\begin{proposition}\label{prop:rkhs_bounds}
    Assume that $(X_1,Y_1),\dots,(X_{n+1},Y_{n+1}) \stackrel{i.i.d.}{\sim} P$ and that $K$ is uniformly bounded.
    Furthermore, suppose $(\Phi(X),S)$ has uniformly upper and lower bounded first three moments (Assumption~\ref{ass:moment_conditions} in the Appendix) and that the distribution of $S \mid X$ is continuous with a uniformly bounded density. Then, for any $f \in \mathcal{F}$,
    \[
    \frac{|\epsilon_{\textup{int}}|}{\mme_P[|f(X)|]} \leq   O\left( \frac{d\log(n)}{\lambda n} \right)\frac{\mme \left [\max_{1 \leq i \leq n+1}|f(X_{i})| \right ]}{\mme_P[|f(X)|]}.
    \]
\end{proposition}

Critically, the interpolation error term decays to zero at a faster-than-parametric rate. As a result, for even moderately large $n$, we expect this term to be of small size. In support of this intuition, we will show an experiment in \Cref{sec:candc_data} in which with sample size $n = 650$, linear dimension $d = 5$, and non-linear hyperparameter $\lambda \cong 1$, we observe interpolation error $\frac{\epsilon_{\textup{int}}}{\mme_P[|f(X)|]} \cong 0.005$. Thus, at appropriate sample sizes, the interpolation error has minimal effect on the coverage. 

We remark that achieving this faster-than-parametric rate requires technical insights beyond existing tools. Two standard ways to establish interpolation bounds are to either to exploit the finite-dimensional character of the model class $\mathcal{F}$ or the algorithmic stability.  Unfortunately, quantile regression with both kernel and linear components satisfies neither property. To get around this problem, we give a three-part argument in which we first separate out the linear component of the fit by discretizing over $\beta$. Then, with $\beta$ fixed, we are able to exploit known stability results to control the kernel component of the fit (\cite{Bousquet2002}). Finally, we combine the two previous steps by giving a smoothing argument that shows that the discretization can be extended to the entire function class. This result may be useful in other applications. For instance, the interpolation error determines the derivative of the loss at the empirical minimizer and, thus, may play a key role in central limit theorems for quantile regressors of this type.

Moving away from these technical issues and returning to our coverage guarantee, \eqref{eq:rkhs_cov}, we find that once the interpolation error is removed, the conditional coverage is completely dictated by the inner product between $f_K$ and $\hat{g}_{S_{n+1}}$. Critically, this implies that in the special case where the target re-weighting $f$ lies completely in the unpenalized part of the function class (i.e. when $f_K = 0$) we have $\mme[\langle \hat{g}_{S_{n+1},K}, f_K \rangle_K] = 0$ and thus $\hat{C}(X_{n+1})$ obtains (nearly) exact coverage under $f$. On the other hand, our next proposition shows that when $f_K \neq 0$, we can use a plug-in estimate to accurately estimate $\mme[\langle \hat{g}_{S_{n+1}}, f_K \rangle_K] $. Thus, even when exact coverage is impossible, a simple examination of the quantile regression fit is sufficient to determine the degradation in coverage under any re-weighting of interest.
\begin{proposition}\label{prop:rkhs_inner_prod_est}
    Assume that $(X_1,S_1),\dots,(X_{n+1},S_{n+1}) \stackrel{i.i.d.}{\sim} P$ and that  $K$ is uniformly bounded. Suppose further that the population loss is locally strongly convex near its minimizer (Assumption \ref{ass:pop_strong_convex} in the Appendix) and $(\Phi(X_i),S_i)$ has uniformly bounded upper and lower first and second order moments (Assumption \ref{ass:moment_cond_for_quant_conv} in the Appendix). Define the $n$-sample quantile regression estimate
    \[
    (\hat{g}_{n,K},\hat{\beta}_n) := \argmin_{g_K \in \mathcal{F}_K,\ \beta \in \mmr^d} \frac{1}{n} \sum_{i=1}^n \ell_{\alpha}(g_K(X_i) + \Phi(X_i)^{\top}\beta,S_i) + \lambda \|g_K\|^2_K,
    \]
    and for any $\delta > 0$, let $\mathcal{F}_{\delta} := \{f(\cdot) = f_K(\cdot) + \Phi(\cdot)^\top\beta   \in \mathcal{F} : \|f\|_K + \|\beta\|_2  \leq 1,\ \mme_P[|f(X)|] \geq \delta\}$. Then,
    \begin{equation}\label{eq:inner_prod_est_bound}
    \sup_{f \in \mathcal{F}_{\delta}} \left| 2\lambda \frac{\langle \hat{g}_{n,K}, f_K \rangle_K}{\frac{1}{n} \sum_{i=1}^{n}|f(X_i)|} - 2\lambda \frac{\mme[\langle \hat{g}_{S_{n+1},K}, f_K \rangle_K]}{\mme_P[|f(X)|]} \right| \leq O\left(\sqrt{\frac{d\log(n)}{n}}\right).
    \end{equation}
\end{proposition}

\subsection{Specialization to the class of Lipschitz functions} 

As a second specialization of \Cref{thm:infinite_dim_result}, we now aim to provide valid coverage over all sufficiently smooth re-weightings of the data. We will do this by examining the set of all Lipschitz functions on $\mathcal{X}$. Namely, suppose $\mathcal{X} \subseteq \mmr^p$ and define the Lipschitz norm of functions $f : \mathcal{X} \to \mmr$ as
\[
\text{Lip}(f) := \sup_{x,y \in \mathcal{X},\ x \neq y} \frac{|f(x) - f(y)|}{\|x - y\|_2}.
\]
Analogous to the previous section, let $\mathcal{F}_L := \{f : \text{Lip}(f) < \infty\}$, $\Phi : \mathcal{X} \to \mmr^d$ be any finite dimensional feature representation of $\mathcal{X}$, and consider implementing our method with the function class $\mathcal{F} = \{f_L(\cdot) + \Phi(\cdot)^\top\beta : f_L \in \mathcal{F}_L,\ \beta \in \mmr^d\}$ and penalty $\mathcal{R}(f_L(\cdot) + \Phi(\cdot)^\top\beta) = \lambda\text{Lip}(f_L)$.

The astute reader may notice that the Lipschitz norm is not differentiable and thus \Cref{thm:infinite_dim_result} is not directly applicable to this setting. Nevertheless, it is not difficult to show that \Cref{thm:infinite_dim_result} can be extended by replacing $\frac{d}{d\epsilon}\mathcal{R}(\hat{g}_{S_{n+1}} + \epsilon f)$ with a subgradient. So, after observing that $|\partial_{\epsilon} \mathcal{R} (\hat{g}_{S_{n+1}} + \epsilon f)|_{\epsilon = 0} |  \leq \lambda \text{Lip}(f)$, we can apply this analogue of \Cref{thm:infinite_dim_result} to find that for any non-negative $f \in \mathcal{F}$ with $\mme_P[f(X)]>0$,
\[
1-\alpha- \lambda\frac{ \text{Lip}(f) }{\mme_P[f(X)]}
 \leq \mmp_f( Y_{n+1} \in \hat{C}(X_{n+1}) ) \leq  1-\alpha+ \lambda \frac{\text{Lip}(f)}{\mme_P[f(X)]} +\frac{\epsilon_{\text{int}}}{\mme_P[f(X)]}. 
\]
Control of the interpolation error is handled in the following proposition.
\begin{proposition} \label{prop:lip_bounds}
    Assume that $(X_1,Y_1),\dots,(X_{n+1},Y_{n+1}) \in \mmr^{p} \times \mmr$ are i.i.d.~and that $X$, $\Phi(X)$, and $S$ have bounded domains and uniformly upper and lower bounded first and second moments (Assumption \ref{ass:lip_tech_conditions} in the Appendix). Furthermore, assume that the distribution of $S \mid X$ is continuous with a uniformly bounded density and that $\Phi(\cdot)$ contains an intercept term. Then for any $f \in \mathcal{F}$,
    \[
    \frac{|\epsilon_{\textup{int}}|}{\mme_P[|f(X)|]} \leq O\left( \left(\frac{p\log(n)}{\lambda n^{\min\{1/2,1/p\}}}\right)^{\frac{1}{2}}  + \left(\frac{dp}{\lambda^2 n} \right)^{\frac{1}{4}} \right) .
    \]
\end{proposition}

This result is considerably weaker than our RKHS bound. While our proof for RKHS function classes made careful use of the stability of RKHS fitting \citep{Bousquet2002}, here we take a more brute force approach and directly examine the uniform concentration properties of the number of interpolated points, $\frac{1}{n+1} \sum_{i=1}^{n+1} \bone\{S_i = g_L(X_i) + \phi(X_i)^\top\beta\}$. We defer a detailed description of this approach to the Appendix. For now, we simply remark that we do not believe this proof technique yields a tight bound and it is possible that significant improvements could be made to \Cref{prop:lip_bounds} with more careful arguments. 

Regardless of the tightness of the bound, we still find that when $X$ is low-dimensional, the interpolation error will be small and the miscoverage of $\hat{C}(X_{n+1})$ under $f$ will be primarily driven by its Lipschitz norm. In light of the impossibility of exact conditional coverage, this result gives a natural interpolation between marginal coverage (in which $\text{Lip}(f) = \text{Lip}(1) = 0)$) and conditional coverage (in which  $\text{Lip}(f)$ can be arbitrarily large). On the other hand, in moderate to high dimensions the interpolation error term will not be negligible and the coverage can be highly conservative. For this reason, in our real data examples, we will prefer to use RKHS functions for which we have much faster convergence rates.

\section{Computing the prediction set}\label{sec:computation}

In order to practically implement any of the methods discussed above, we need to be able to efficiently compute $\hat{C}(X_{n+1}) = \{y : S(X_{n+1},y) \leq \hat{g}_{S(X_{n+1},y)}(X_{n+1})\}$. Naively, this recursive definition requires us to fit $\hat{g}_S$ for all possible values of $S \in \mmr$. We will now show that by exploiting the monotonicity properties of quantile regression, this naive computation can be overcome and a valid prediction set can be computed efficiently using only a small number of fits.

The main subtlety that we will have to contend with is that $\hat{g}_S$ may not be uniquely defined. For example, consider computing the median of the dataset $\{S_1,S_2,S_3,S_4\} = \{1,2,3,4\}$. It is easy to show that any value in the interval $[2,3]$ is a valid solution to the median quantile regression $\text{minimize}_{\theta} \sum_{i=1}^4 \ell_{1/2}(\theta,i)$. Critically, this means that it is ambiguous whether or not $3$ lies below or above the median. More generally, in our context, it can be ambiguous whether or not $S \leq \hat{g}_S$. In the earlier sections of this article we have elided such non-uniqueness in the definition of $\hat{g}_S$. We do this because the choice of $\hat{g}_S$ is not critical to the theory and, in particular, all sensible definitions will give the same coverage properties (recall that all of our theoretical results go through so long as $\hat{g}_{S}$ is computed using an algorithm that is invariant under re-orderings of the input data). However, while not theoretically relevant, this ambiguity can cause practical issues to arise in the computation.

The main insight of this section is that these technical issues can be resolved by re-defining $\hat{C}(X_{n+1})$ in terms of the dual formulation of the quantile regression. This will give us a new prediction set, $\hat{C}_{\text{dual}}(X_{n+1})$, that can be computed efficiently and satisfies the same coverage guarantees as $\hat{C}(X_{n+1})$. At a high level, $\hat{C}_{\text{dual}}(X_{n+1})$ is obtained from $\hat{C}(X_{n+1})$ by removing a small portion of the points $y$ that lie on the interpolation boundary $\{y : S(X_{n+1},y) = \hat{g}_{S(X_{n+1},y)}(X_{n+1})\}$. Thus, one should simply think of $\hat{C}_{\text{dual}}(X_{n+1})$ as a trimming of the original prediction set that removes some extraneous edge cases. 

To define our dual optimization more formally, recall that throughout this article we have considered quantile regressions of the form,
\[
\underset{{g \in \mathcal{F}}}{\text{minimize}}\ \frac{1}{n+1}\sum_{i=1}^{n} \ell_{\alpha}(g(X_i),S_i) + \frac{1}{n+1}\ell_{\alpha}(g(X_{n+1}),S) + \mathcal{R}(g).
\]
Instead of directly computing the dual of this program, we first re-formulate this optimization into the identical procedure,
\begin{equation}\label{eq:generic_convex_opt}
    \begin{split}
          \underset{p,q \in \mmr^{n+1},\ g\in \mathcal{F}}{\text{minimize}} \quad &  \sum_{i = 1}^{n+1} (1-\alpha)p_i + \alpha q_i +  (n + 1) \cdot \mathcal{R}(g),\\
          \text{subject to} \quad &  S_i - g(X_i) - p_i + q_i = 0,\\
        &  S - g(X_{n+1}) - p_{n+1} + q_{n+1} = 0,\\
        &  p_i,q_i \geq 0. 
    \end{split}
    \end{equation}
Then, after some standard calculations, this yields the desired dual formulation,
\begin{equation}\label{eq:generic_dual}
\begin{split}
    \underset{\eta \in \mmr^{n+1}}{\text{maximize}} \quad & \sum_{i = 1}^n \eta_i S_i + \eta_{n+1} S - \mathcal{R}^*\left( \eta \right)   \\
    \text{subject to} \quad & -\alpha \leq \eta_i \leq 1 - \alpha,
 \end{split}
\end{equation}
\sloppy where $\mathcal{R}^* (\cdot)$ denotes the function $\mathcal{R}^*(\eta) := -\min_{g \in \mathcal{F}} \,\{(n+1)\mathcal{R}(g) - \sum_{i=1}^{n+1}\eta_i g(X_i)\}$; heuristically, we can think of $\mathcal{R}^*(\cdot)$ as the convex conjugate for $\mathcal{R}(\cdot)$. 

Crucially, the KKT conditions for \eqref{eq:generic_convex_opt} allow for a more tractable definition of our prediction set. Letting $\eta^S$ denote any solution to \eqref{eq:generic_dual} and applying the complementary slackness conditions of this primal-dual pair we find that
\begin{align*}
    \eta^S_{n + 1} \in \begin{cases}
        - \alpha &\text{if $S < \hat{g}_S(X_{n + 1}),$} \\
        [-\alpha, 1 - \alpha] & \text{if $S = \hat{g}_S(X_{n + 1}),$} \\
        1 - \alpha &\text{if $S > \hat{g}_S(X_{n + 1}).$}
    \end{cases}
\end{align*}
As a consequence, checking whether $\eta_{n+1}^S < 1-\alpha$ is nearly equivalent to checking that $S \leq \hat{g}_S(X_{n + 1})$, albeit with a minor discrepancy on the interpolation boundary. This enables us to define the efficiently computable prediction set,
\begin{equation}
\hat{C}_{\text{dual}}(X_{n+1}) := \{y :  \eta_{n+1}^{S(X_{n+1},y)} < 1 - \alpha\}. \label{eq:dual_pred_set}
\end{equation}
In practice, $\hat{C}_{\text{dual}}(X_{n+1})$ can be mildly conservative. If we are willing to allow the prediction set to be randomized then exact coverage can be obtained using the prediction set,
\[
\hat{C}_{\text{dual, rand.}} := \{y : \eta_{n+1}^{S(X_{n+1},y)} < U\},
\]
where $U \sim \text{Unif}([-\alpha,1-\alpha])$ is drawn independent of the data.

Our first result verifies that $\hat{C}_{\text{dual}}(X_{n + 1})$ obtains the same coverage guarantees as our non-randomized primal set $\hat{C}(X_{n + 1})$, while $\hat{C}_{\text{dual, rand.}}$ realizes exact, non-conservative coverage.
\begin{proposition}\label{prop:cov_of_dual}
    Assume that the primal-dual pair \eqref{eq:generic_convex_opt}-\eqref{eq:generic_dual} satisfies strong duality and the dual solutions $\{\eta^S\}_{S \in \mmr}$ are computed using an algorithm that is symmetric in the input data. Then, 
    \begin{enumerate}
    \item The statement of \Cref{thm:infinite_dim_result} is valid with $\hat{C}(X_{n+1})$ replaced by $\hat{C}_{\textup{dual}}(X_{n+1})$,
    \item The statement of \Cref{thm:infinite_dim_result} is valid with $\hat{C}(X_{n+1})$ replaced by $\hat{C}_{\textup{dual, rand.}}(X_{n+1})$ and $\epsilon_{\text{int}}$ set to $0$.
    \end{enumerate}
\end{proposition}

The assumption that the primal-dual pair satisfies strong duality is very minor and we verify in \Cref{sec:app_comp_set-up} that it holds for all the function classes and penalties considered in this article. Moreover, we also verify in \Cref{sec:app_comp_set-up} that for all the function classes and penalties considered in this article, $\mathcal{R}^*(\cdot)$ is tractable and thus solutions to \eqref{eq:generic_dual} can be computed efficiently.
    
The main result of this section is \Cref{thm:pred_set_is_mon}, which states that $S \mapsto \eta^S_{n+1} $ is non-decreasing. Critically, this implies that membership in the dual prediction set \eqref{eq:dual_pred_set} is monotone in the imputed score $S(X_{n + 1}, y)$.

\begin{theorem}\label{thm:pred_set_is_mon}
     For all maximizers $\{\eta^S_{n+1}\}_{S \in \mmr}$ of \eqref{eq:generic_dual}, $S \mapsto \eta^S_{n+1}$ is non-decreasing in $S$.
\end{theorem}

Leveraging \Cref{thm:pred_set_is_mon}, we may compute $\hat{C}_{\text{dual}}(X_{n + 1})$ (or $\hat{C}_{\text{dual, rand.}}(X_{n + 1})$) using the following two-step procedure. First, we identify the largest value of $S$ such that $\eta^S_{n + 1} < 1 - \alpha$ (or $\eta^S_{n+1} < U$). Second, denoting this upper bound by $S^*_{n + 1}$, we output all $y$ such that $S(X_{n + 1}, y) \leq S^*_{n + 1}$. The second step is straightforward for all commonly used conformity scores. For example, if $S(X_{n+1},y) = |\hat{\mu}(X_{n+1}) - y|$, the prediction set becomes $\hat{\mu}(X_{n+1}) \pm S^*_{n + 1}$. 

The monotonicity of the dual variable in $S$ allows for more than one approach to the first step of this procedure. If we have no additional information about the structure of the optimization problem over $\eta$, it is always possible to run a binary search over $S$ to find the largest value such that $\eta^S_{n + 1}$ is less than the targeted cutoff.  On the other hand, for the typical use-case of this method, i.e., when we fit an unregularized quantile regression over a finite-dimensional function class, it is considerably more efficient to compute this cutoff for $S$ by applying standard tools from linear program sensitivity analysis. We defer a detailed description of our approach to \Cref{alg:lp_sensitivity}. 

\begin{figure}[ht]
    \centering
    \includegraphics[width=\textwidth]{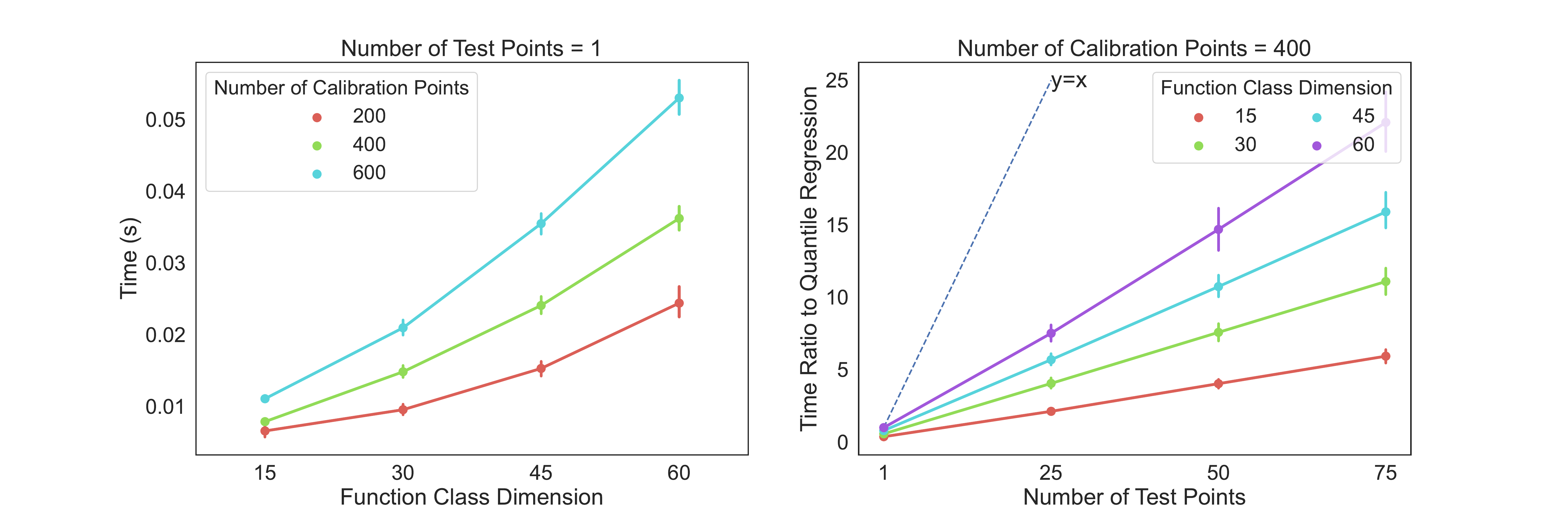}
    \caption{Evaluation of the computational efficiency of our prediction set construction. The left panel displays the time taken to fit a single test point, while the right panel shows the ratio between the time taken by our method versus the cost of fitting a single quantile regression per test point. Both plots display results for the unrandomized version of our conditional calibration method implemented in its one-sided form with conformity score $S(x, y) = y$ and $1-\alpha=0.9$, i.e., we estimate the $0.9$-quantile of $Y \mid X$. Data for this simulation is generated from the Gaussian linear model, $Y_i = X_i^\top w + \epsilon_i$ where $X_i \sim \mathcal{N}(0, I_d)$, $\epsilon_i \sim \mathcal{N}(0, 1)$, and $w \sim \text{Unif}(\mathcal{S}^{d - 1})$, and our method is implemented using the linear function class $\mathcal{F} := \{\beta_0 + X^\top \beta_1: \beta_0 \in \mmr,\ \beta_1 \in \mathbb{R}^d\}$. Dots and error bars show means and confidence intervals from 25 trials.}
    \label{fig:comp_comparison}
\end{figure}

\Cref{fig:comp_comparison} displays the computational efficiency of this sensitivity analysis method on a 2020 MacBook Pro with an Intel Core i5 processor and 16GB of RAM. The left panel of the figure displays  the estimated time to construct one prediction set over varying function classes and calibration set sizes. We find that our method is computationally efficient, and that its time complexity depends much more on the dimension than the sample size. In general, the bulk of the computational cost consists of fitting a single quantile regression on the calibration set. From there, predictions for new test points are obtained by an efficient procedure for updating the quantile fit (see \Cref{alg:lp_sensitivity} in the Appendix for details). The computational benefit of this approach is displayed in the right panel of \Cref{fig:comp_comparison}, which compares the cost of running our procedure against the time it would take to run a single quantile regression for each test point. We see that our updating procedure is substantially faster across a wide range of calibration set sizes and function class dimensions.

\section{Real data experiments}\label{sec:real_data}

\subsection{Communities and crime data}\label{sec:candc_data}

We now illustrate our methods on two real datasets. For our first experiment, we consider the Communities and Crime dataset (\cite{Dua2019, Redmond2002}). In this task, the goal is to use the demographic features of a community to predict its per capita violent crime rate. To make our analysis as interpretable as possible, we use only a small subset of the covariates in this dataset: population size, unemployment rate, median income, racial make-up (given as four columns indicating the percentage of the population that is Black, White, Hispanic, and Asian), and age demographics (given as two columns indicating the percentage of people in the 12-21 and 65+ age brackets). 

Our goal in this section is not to give the best possible prediction intervals. Instead, we perform a simple expository analysis that is designed to demonstrate some possible use cases of our method. For this purpose, we assume that the practitioner's primary concern is that standard prediction sets will provide unequal coverage across communities with differing racial make-ups. Additionally, we suppose that the practitioner does not have any particular predilections for achieving coverage conditional on age, unemployment rate, or median income. We encode these preferences as follows: let $\Phi(X_i)$ denote the length five vector consisting of an intercept term and the four racial features and define $\mathcal{F}_K$ to be the RKHS given by the Gaussian kernel $K(X_i,X_j) = \exp(-4\|X_i - X_j\|_2^2)$ (note that since all variables in this dataset have been previously normalized to lie in $[0,1]$, we do not make any further modifications before computing the kernel). Then, proceeding exactly as in \Cref{sec:rkhs_specialization} we run our method with the function class $\mathcal{F} := \{f_K(\cdot) + \Phi(\cdot)^\top\beta : f_K \in \mathcal{F}_K,\ \beta \in \mmr^5\} $ and penalty $\mathcal{R}(f_K(\cdot) + \Phi(\cdot)^\top\beta) := \lambda \|f_K\|^2_K$. To get $\lambda$, we run cross-validation on the calibration set. While this does not strictly follow the theory of \Cref{sec:infinite_dim}, our results indicate that choosing $\lambda$ in this manner does not negatively impact the coverage properties of our method in practice (see Figures \ref{fig:cov_heatmap} and \ref{fig:cross_val_cov_comp} below). Finally, we set the conformity scores to be the absolute residuals of a linear regression of $Y$ on $X$ and the target coverage level to be $1-\alpha = 0.9$. The sizes of the training, calibration, and test sets are taken to be 650, 650, and 694, respectively and we use the unrandomized prediction set, $\hat{C}_{\text{dual}}$ throughout.

\begin{figure}[ht]
    \centering
    \includegraphics[scale=0.35]{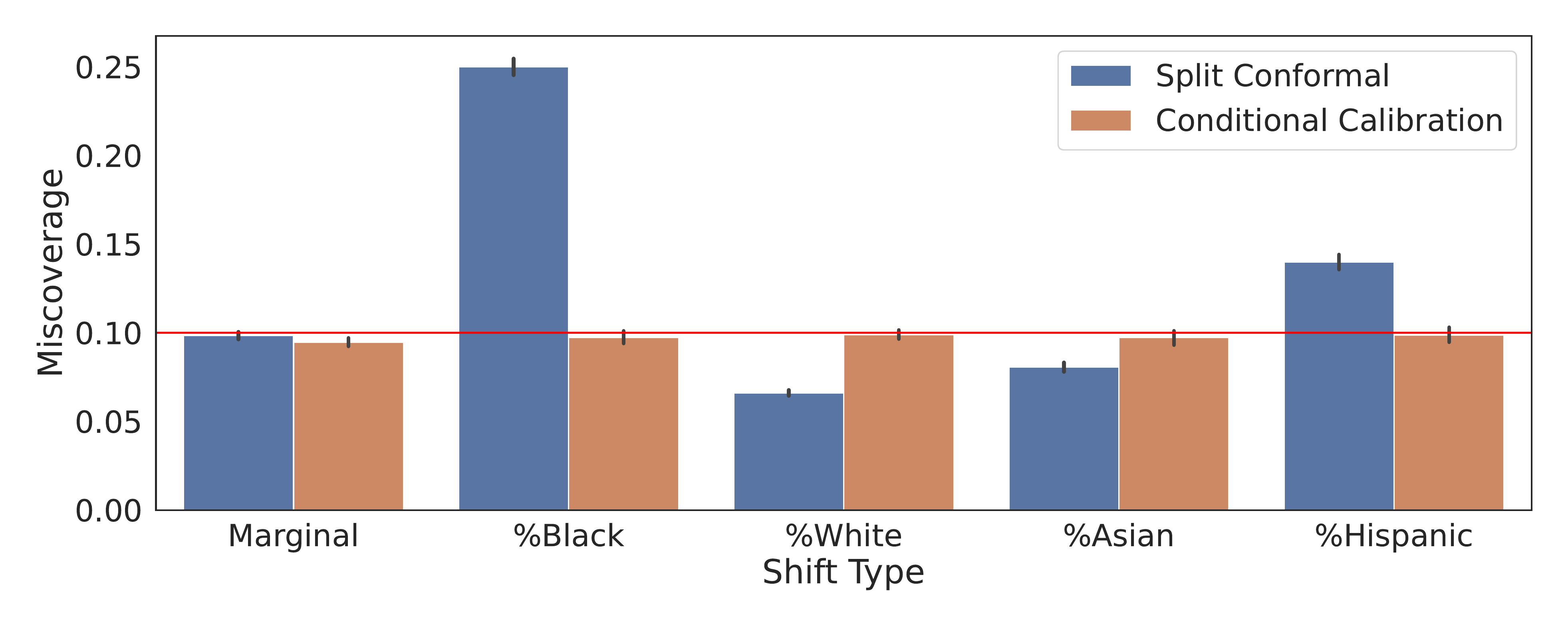}
    \caption{Empirical miscoverage achieved by split conformal inference (blue) and the unrandomized version of our method (orange) both marginally over all data points and across linear re-weightings of the four racial categories. The red line shows the target level of $\alpha = 0.1$ and the error bars indicate 95\% confidence intervals for the value of $\mmp(Y_{n+1} \notin \hat{C}(X_{n+1}))$ obtained by averaging over 200 train-calibration-test splits.}
    \label{fig:racial_cov}
\end{figure}

\Cref{fig:racial_cov} shows the empirical miscoverages obtained by our algorithm and split conformal prediction under the linearly re-weighted distributions $\mmp_f$, for $f \in \{x \mapsto 1,\ x \mapsto x_{\rm \%Black},\ x \mapsto x_{\rm \%White},\ x \mapsto x_{\rm \%Hispanic},\ x \mapsto x_{\rm \%Asian}\}$. As expected, our method obtains the desired coverage level under all five settings, while split conformal is only able to deliver marginal validity. 

\begin{figure}[!b]
    \centering
    \includegraphics[scale=0.35]{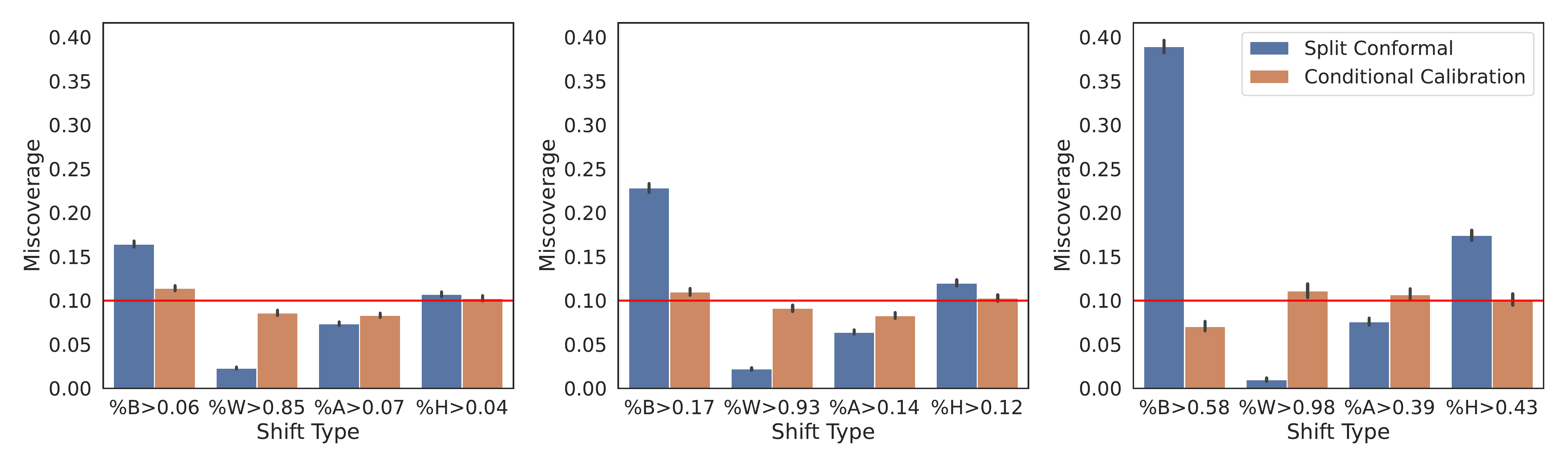}
    \caption{Empirical miscoverage achieved by split conformal prediction (blue) and the unrandomized version of our method (orange) across high racial representation subgroups. Panels from left to right show the miscoverage when the percentage of the community that falls into a racial group is in the top $p$-percentile, for $p \in \{50, 70, 90\}$. To conserve space, the short forms B, W, A, and H are used to indicate the racial categories Black, White, Asian, and Hispanic, respectively. Red lines shows the target level of $\alpha = 0.1$ and the black error bars indicate 95\% confidence intervals for the value of $\mmp(Y_{n+1} \notin \hat{C}(X_{n+1}))$ obtained by averaging over 200 train-calibration-test splits.}
    \label{fig:racial_cutoff_cov}
\end{figure}

In practice, the user is unlikely to only care about the performance under linear re-weightings. For example, they may also want to have predictions sets that are accurate on the communities with the highest racial bias, e.g.~the communities whose percentage of the population that is black is in the top $90$th percentile. Strictly speaking, our method will only provide a guarantee on these high racial representation groups if the corresponding subgroup indicators are included in $\mathcal{F}$. However, intuitively, even without explicit indicators, linear re-weightings should already push the method to accurately cover communities with high racial bias. Thus, we may expect that even without adjusting $\mathcal{F}$, the method implemented above will already perform well on these groups. To investigate this, \Cref{fig:racial_cutoff_cov} shows the miscoverage of our method across high racial representation subgroups. Formally, we say that a community has high representation of a particular racial group if the percentage of the community that falls in that group is in the top $p$-percentile. The three panels of the figure then show results for $p= 50$, $70$, and $90$, respectively. We find that even without any explicit indicators for the subgroups, our method is able to correct the errors of split conformal prediction and provide improved coverage in all settings.

\begin{figure}[ht]
  \centering
  \begin{tabular}{@{}c@{}}
    \includegraphics[width=0.45\textwidth]{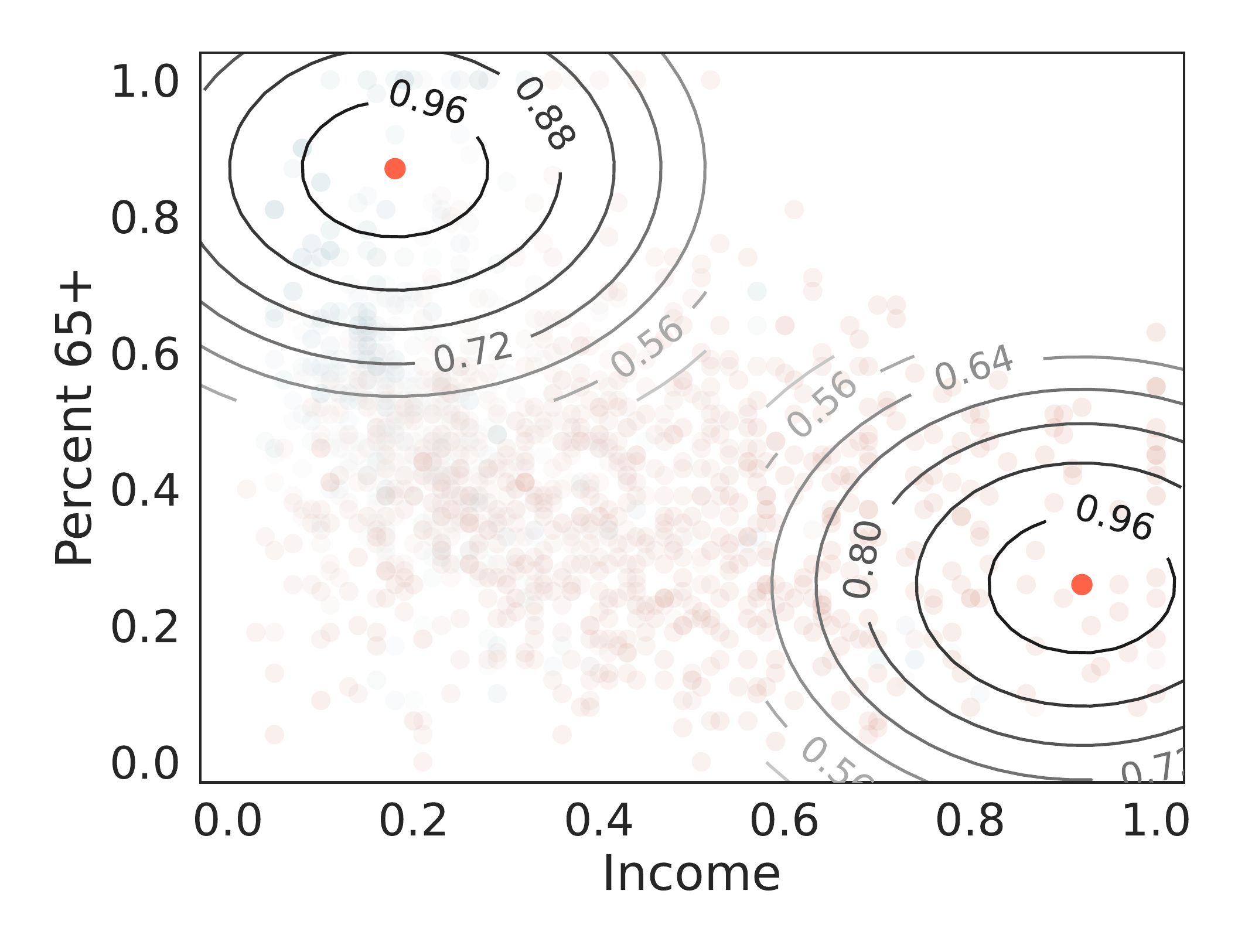}
  \end{tabular}
  \vspace{\floatsep}
  \begin{tabular}{@{}c@{}}
    \includegraphics[width=\textwidth]{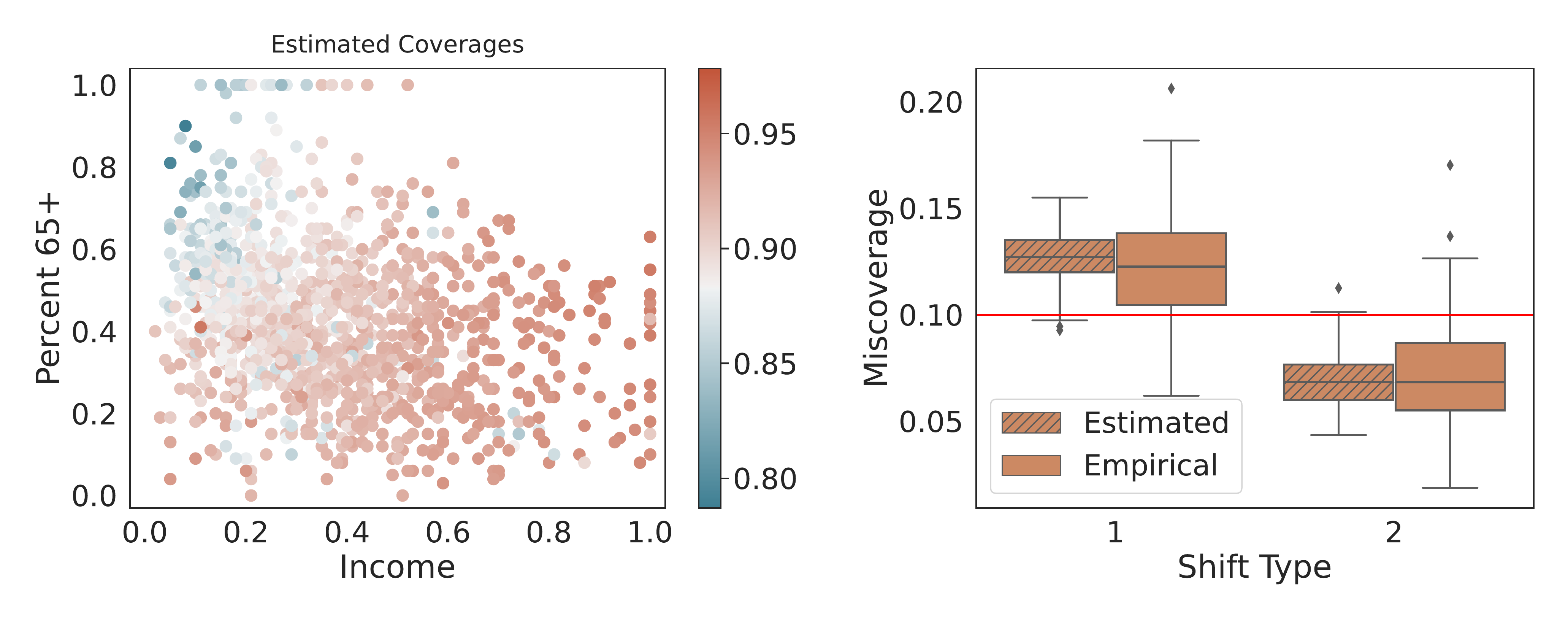}
  \end{tabular}
    \caption{Estimated coverages for re-weightings of the data according to a Gaussian kernel centered at each of the plotted data points. The bottom-left panel shows the estimated coverages output by our method, while the top panel gives level curves of the kernel at two specific points highlighted in red. Finally, the bottom-right panel compares boxplots of the empirical (solid) and estimated (striped) values of $\mmp_{f^x}(Y_{n+1} \notin \hat{C}(X_{n+1}))$ obtained over 200 train--calibration-test splits for the two highlighted values of $x$. The red line denotes the target level of $\alpha = 0.1$.}
        \label{fig:cov_heatmap}
\end{figure}

In addition to providing exact coverage over racial re-weightings, our method also provides a simple procedure for evaluating the coverage across the other, more flexibly fit, covariates. More precisely, let 
\[
(\hat{\beta}_n,\hat{g}_{n,K}) = \argmin_{\beta \in \mmr^d, g_K \in \mathcal{F}_K} \frac{1}{n} \sum_{i=1}^n \ell_{\alpha}(g_K(X_i) + \Phi(X_i)^\top\beta, S_i)  + \lambda \|g_K\|^2_K,
\]
denote the function fit on the calibration data and recall that by the results of \Cref{sec:infinite_dim} we expect that for any non-negative re-weighting $f(\cdot) =  f_K(\cdot) + \Phi(\cdot)^\top\beta $,
\begin{equation}\label{eq:cov_approximation_for_RKHS}
\mmp_f(Y_{n+1} \in \hat{C}(X_{n+1})) \cong 1-\alpha - 2\lambda \frac{\langle \hat{g}_{n,K} , f_K\rangle}{\frac{1}{n}\sum_{i=1}^n f(X_i)}.
\end{equation}
 For the Gaussian RKHS, a natural set of non-negative weight functions are the local re-weightings $f^x(y) := K(x,y) = \exp(-4\|x - y \|^2)$, which emphasize coverage in a neighbourhood around the fixed point $x$. 
 
 The bottom-left panel of \Cref{fig:cov_heatmap} plots the values of $1-\alpha - 2\lambda \frac{\langle \hat{g}_{n,K} , f^x\rangle}{\frac{1}{n}\sum_{i=1}^n f^x(X_i)}$ for all points $x$ appearing in the training and calibration sets. We see immediately that our prediction sets will undercover older communities and overcover communities with high median incomes. To aid in the interpretation of this result, the top panel of the figure indicates the level curves of $K(x,\cdot)$ for two specific choice of $x$. Finally, the bottom-right panel compares the estimates \eqref{eq:cov_approximation_for_RKHS} to the realized empirical coverages
\[
\sum_{i=1}^{694} \frac{f^x(\tilde{X}_i)}{ \frac{1}{694} \sum_{j=1}^{694} f^x(\tilde{X}_j)} \bone\{\tilde{Y}_{i} \in \hat{C}(\tilde{X}_i)\},
\]
for the same two values of $x$. Here, $\{(\tilde{X}_i,\tilde{Y}_i)\}_{i=1}^{694}$ denotes the test set. We see that as expected \eqref{eq:cov_approximation_for_RKHS} is a highly accurate estimate of the true coverage at both values of $x$. To further understand the degree of localization in these plots, it may be useful to note that this re-weighting yields an effective sample size of $\frac{(\sum_{i=1}^{694} f^x(\tilde{X}_i))^2}{  \sum_{i=1}^{694} f^x(\tilde{X}_i)^2} \cong 275$ at the two red points.

Overall, we find that our procedure provides the user with a highly accurate picture of the coverage properties of their prediction sets. In many practical settings, plots like the bottom-left panel of \Cref{fig:cov_heatmap} may prompt practitioners to adjust the quantile regression to protect against observed directions of miscoverage. While such an adjustment will not strictly follow the theory of \Cref{sec:finite_dim,sec:infinite_dim}, so long as the user is careful not to run the procedure so many times as to induce direct over-fitting to the observed miscoverage, small adjustments will likely be permissible. In practice, this type of exploratory analysis may allow practitioners to discover important patterns in their data and tune the prediction sets to reflect their coverage needs.

\subsection{Comparison against existing methods}

A variety of alternative methods for obtaining conditional coverage in conformal inference have been proposed in the literature. Many of these methods are designed to asymptotically achieve exact conditional coverage, i.e., $\mmp(Y_{n+1} \in \hat{C}(X_{n+1}) \mid X_{n+1}) \stackrel{\mmp}{\longrightarrow} 1 - \alpha$. Here, we compare against two such approaches and demonstrate why the more precise finite-sample guarantees of our method may be preferable. 

\subsubsection{Comparison against conformalized quantile regression}

Our first comparison is to the conformalized quantile regression (CQR) method of \citet{Romano2019}. They consider a version of split conformal inference in which the training set is used to fit estimates of the conditional quantiles of $Y \mid X$ and the calibration set is used to adjust these estimates to guarantee marginal coverage. Here, we implement a two-sided version of their procedures in which the upper and lower quantiles are calibrated separately. In particular, let $\hat{q}_{\tau}(X)$ denote an estimate of the $\tau$-quantile of $Y \mid X$. Let $S_{\tau}(X,Y) := Y - \hat{q}_{\tau}(X)$ be a conformity score measuring the distance between $Y$ and the quantile estimate and $c_{\tau}$ denote the $\tau$-quantile of the calibration scores $\{S_{\tau}(X,Y)\}_{i=1}^n$. Then, we define the two-sided CQR prediction set as
\begin{equation}\label{eq:cqr_set}
\hat{C}_{CQR}(X_{n+1}) := \{y : S_{\alpha/2}(X_{n+1},y) > c_{\alpha/2},\ S_{1-\alpha/2}(X_{n+1},y) \leq c_{1-\alpha/2} \}.
\end{equation}
As alluded to above, if $\hat{q}_{\tau}(X)$ is a consistent estimate of the true conditional quantile function of $Y \mid X$, then this set will asymptotically achieve exact conditional coverage \citep{Sesia2020}.

It is important to note that our approach is not in direct competition with CQR and both methods can be used in combination. Namely, here we will consider an implementation of our procedure in which the quantiles, $c_{\alpha/2}$ and $c_{1-\alpha/2}$ are replaced by adaptive estimates from our method. Implementing two-sided fitting requires requires minor modifications of the dual formulation of our prediction set and we refer the reader to \Cref{sec:app_two-sided} for details.

To compare these methods we once again employ the Communities and Crimes dataset. We consider three baseline approaches, 1) split conformal prediction with linear regression and the residual conformity score, $S(X,Y) := Y - \hat{\gamma}_0 - X^\top\hat{\gamma}_1$ for $(\hat{\gamma}_0,\hat{\gamma}_1)$ obtained by ordinary least squares, 2) split conformal prediction implemented using the CQR procedure of \citet{Romano2019} where the estimates $\hat{q}_{\tau}(X)$ are obtained using linear quantile regression, 3) the same procedure but with $\hat{q}_{\tau}(X)$ obtained using a quantile random forest. In all three cases we calibrate the upper and lower quantiles separately using either the formulation in \eqref{eq:cqr_set} or, in the case of ordinary least squares, by estimating the upper and lower quantiles of the residuals separately. We compare each of these baseline approaches against implementations of our methods where the split conformal calibration step is replaced by our method with function class, $\mathcal{F}:=\{\beta_0 + X^\top\beta_1 : \beta_0 \in \mmr,\ \beta_1 \in \mmr^d\}$ defined as the set of affine functions over the covariates from the previous section (i.e. population size, unemployment rate, median income, as well a set of race and age-based features). 

\begin{figure}
    \centering
    \includegraphics[width=\textwidth]{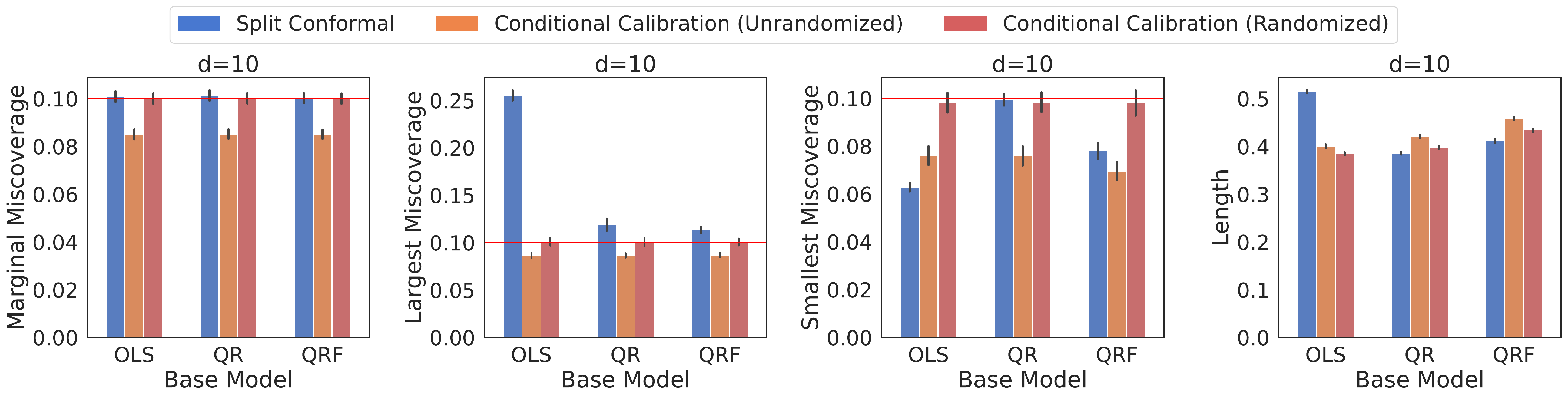}
    \caption{Comparison of split conformal inference (blue) against the unrandomized (orange) and randomized (red) implementations of our conditional calibration method. All three procedures are implemented with three different conformity scores derived from ordinary least squares (OLS), linear quantile regression (QR), and quantile random forests (QRF). Panels from left to right show the average marginal miscoverage, largest and smallest miscoverages over linear re-weightings, and prediction set lengths. Bars in the figure display averages over 200 trials, where in each trial the points in the Communities and Crimes dataset are divided uniformly at random into a training set of size 650, a calibration set of size 650, and a test set of size 694.}
    \label{fig:candc_small}
\end{figure}

\Cref{fig:candc_small} displays the results of this experiment. We find that all implementations of both split conformal and our randomized method achieve exact marginal covearge, while our unrandomized variant can slightly overcover (left panel). To evaluate the conditional coverage, the center two panels of the figure display estimates of the largest over and undercoverages obtained across linear re-weightings of the covariate space. Namely, letting $Z=(1,X)$ denote the augmented covariates, we estimate the smallest and largest values of
\[
\frac{\mme[Z_{n+1,j}\bone\{Y_{n+1} \notin \hat{C}(X_{n+1})\}]}{\mme[Z_{n+1,j}]},
\]
over all features $j$ (note that here all the features are non-negative). As expected from our theory, the unrandomized variant of our method always slightly overcovers, while the randomized version obtains exact coverage throughout. On the other hand, while CQR asymptotically achieves exact coverage in this setting, our simulations show deviations from the target level using both linear quantile regression and quantile forests. Importantly, we note that while the coverage deviations of linear quantile regression are relatively small here, its performance can worsen significantly in higher dimensions (see \Cref{fig:cc_cov_and_length}). Finally, the rightmost panel of the figure shows the lengths of the prediction sets returned by each method. We find that all of the prediction sets are of a similar size with the exception of the OLS implementation of split conformal prediction, whose predictions sets are relatively wide.

\subsubsection{Comparison against localized conformal prediction}

Our second comparison is to the localized conformal prediction method of \citet{Leying2022}. In this procedure, a localized kernel is used to re-weight the calibration data and focus on the data whose covariates are most similar to the test point. The exact method for accomplishing this is somewhat technical and we refer the reader to Theorem 3.2 of \citet{Leying2022} for details. Similar to conformalized quantile regression, localized conformal prediction guarantees exact marginal coverage in finite samples and asymptotic conditional coverage under appropriate choices of the kernel.

To compare localized conformal prediction against our approach, we once again employ the Communities and Crime dataset. We consider the Gaussian kernel $K(x_1,x_2) := \exp(-4\|x_1 - x_2\|_2^2)$ and take the conformity score $S(X,Y) := |Y - \hat{\gamma}_0 - X^\top\hat{\gamma}|$ to be the absolute values of the residuals from a linear regression. We then compare two methods, localized conformal prediction, and our method implemented with the function class given by the same kernel. More specifically, we implement our method with function class $\mathcal{F} := \{\beta_0 + f_k(X) : \beta_0 \in \mmr, f_k \in \mathcal{F}_K\}$ and penalty $\mathcal{R}(f_K) = \lambda\|f\|_K^2$, where $\lambda$ is chosen by cross-validation. 

\begin{figure}
    \centering
    \includegraphics[width=\textwidth]{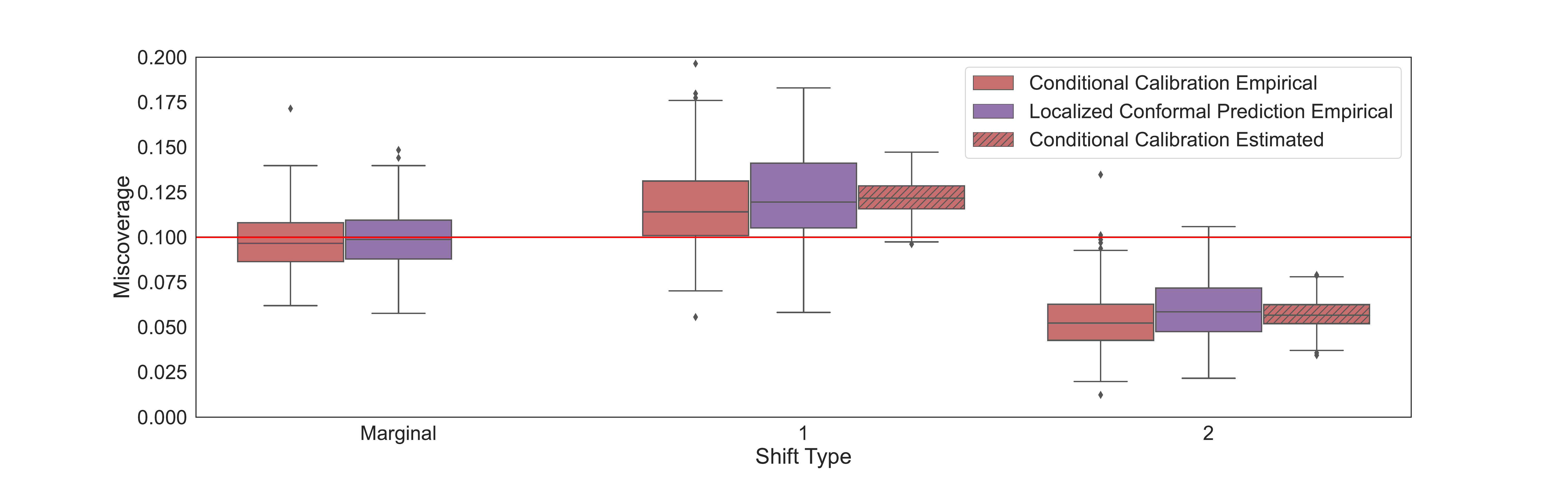}
    \caption{Comparison of the coverage properties of localized conformal prediction (purple) and our randomized conditional calibration method (red). The solid boxplots from left to right show the distributions of the calibration-conditional marginal coverage and kernel re-weighted coverage according to the shifts considered in \Cref{fig:cov_heatmap}, while hatched boxes display the corresponding coverage estimates obtained using the method of \Cref{prop:rkhs_inner_prod_est}. The feature space employed here is the same as in the previous sections and boxplots in the figure once again show results from 200 trials each containing 650 training points, 650 calibration points, and 694 test points.}
    \label{fig:lcp_comp}
\end{figure}

The results of this experiment are shown in \Cref{fig:lcp_comp}. We find that, as expected, both localized conformal prediction and the randomized version of our method achieve exact marginal. To evaluate their conditional properties, we compare the coverages obtained under the two kernel re-weightings from the previous section (namely the shifts considered in \Cref{fig:cov_heatmap}). We find that, as expected, the coverage of our method deviates from the target level, but critically this deviation is predictable and well approximated by the estimation procedure given in \Cref{prop:rkhs_inner_prod_est}. On the other hand, although localized conformal prediction obtains similar empirical results to our method, it does not offer estimates of its coverage deviation. Thus, \textit{a priori} based on the theory of \citet{Leying2022}, one may expect localized conformal prediction to give exact coverage under these shifts, while in reality, it shows notable deviations away from the target level.

\subsection{Rxrx1 data}\label{sec:cell_data}

Our next experiment examines the RxRx1 dataset (\cite{Taylor2019}) from the WILDS repository (\cite{Koh2021}). This repository contains a collection of commonly used datasets for benchmarking performance under distribution shift. In the RxRx1 task, we are given images of cells obtained using fluorescent microscopy and we must predict which one of the 1339 genetic treatments the cells received. These images come from 51 different experiments run across four different cell types. It is well known that even in theoretically identical experiments, minor variations in the execution and environmental conditions can induce observable variation in the final data. Thus, we expect to see heterogeneity in the quality of the predictions across both experiments and cell types. 

Perhaps the most obvious method for correcting for this heterogeneity would be to treat the experiments and cell types as known categories and apply the group-conditional coverage method outlined in \Cref{sec:finite_dim}. However, if we did this, we would be unable to make predictions for new unlabeled images. Thus, here we take a more data driven approach and attempt to learn a good feature representation directly from the training set. 

To predict the genetic treatments of the cells we use a ResNet50 architecture trained on 37 of the 51 experiments by the original authors of the WILDS repository. We then uniformly divide the samples from the remaining 14 experiments into a training set and a test set. The training set is further split into two halves; one for learning a feature representation, and one to be used as the calibration set. To construct the feature representation, we take the feature map from the pre-trained neural network as input and run a $\ell_2$-regularized multinomial linear regression that predicts which experiment each image comes from. We then define our features to be the predicted probabilities of experiment membership output by this model and construct prediction sets using the linear quantile regression method of \Cref{sec:finite_dim}. To define the conformity scores for this experiment, let $\{\hat{w}_i(x)\}_{i=1}^{1339}$ denote the weights output by the neural network at input $x$. Typically, we would use these weights to compute the predicted probabilities of class membership, $\hat{\pi}_i(x) := \exp(\hat{w}_i(x))/(\sum_{j} \exp(\hat{w}_j(x)))$. Here, we add an extra step in which we use multinomial logisitic regression and the calibration data to fit a parameter $T$ that re-scales the weights. This procedure is known as temperature scaling and it has been found to increase the accuracy of probabilities output by neural networks (\cite{Angelopoulos2021, Guo2017}). After running this regression, we set $\hat{\pi}_i(x) := \exp(T\hat{w}_i(x))/(\sum_{j} \exp(T\hat{w}_j(x)))$ and following \citet{Romano2020b} define the conformity score function to be
\[
S(x,y) := \sum_{i : \hat{\pi}_{i}(x) > \hat{\pi}_{y}}\hat{\pi}_{i}(x).
\]

\begin{figure}[ht]
    \centering
    \includegraphics[width=\linewidth]{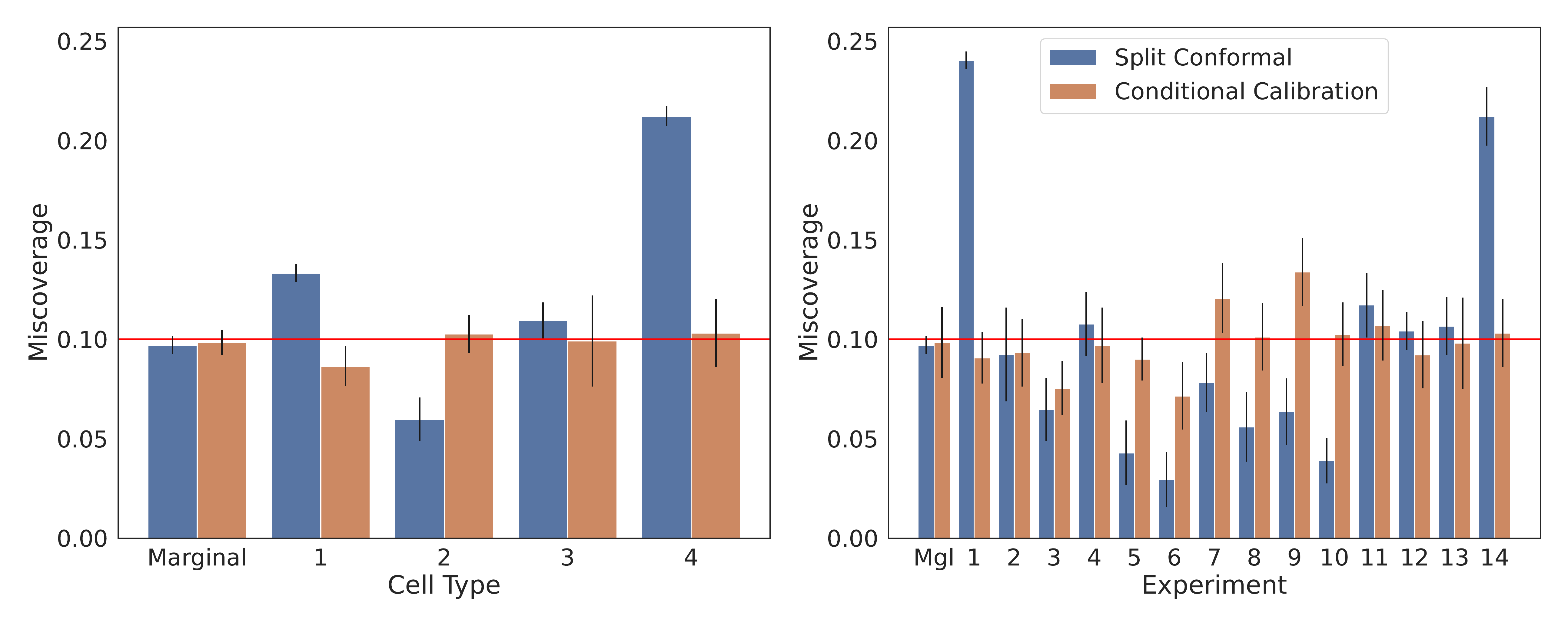}
    \caption{Empirical conditional miscoverage of the unrandomized version of our method (orange) and split conformal (blue) across cell types and experiments. Red lines indicate the target level of $\alpha = 0.1$ and black error bars show 95\% binomial confidence intervals for the calibration-conditional miscoverage $\mmp(Y_{n+1} \notin \hat{C}(X_{n+1}) | D_{\text{train}}, D_{\text{cal}})$, where $D_{\text{train}}$ and $D_{\text{cal}}$ denote the training dataset used to learn the feature representation and the calibration dataset used to implement our method, respectively.}
    \label{fig:rx_in_sample_results}
\end{figure}

The coverage properties of our unrandomized prediction sets are outlined in \Cref{fig:rx_in_sample_results}. We see that while split conformal prediction has very heterogeneous coverage across experiments and cell types, our approach performs well on all groups. Thus, the learned feature representation successfully captures the batch effects in the data and thereby enables our method to provide the desired group-conditional coverage. 

The method described above is hardly the only way to construct a feature representation for this dataset. In \Cref{sec:app_rx1_extra_exp}, we consider an alternative approach in which we implement our method using the top principal components of the feature layer of the neural network as input. The idea here is that the batch effects (i.e. the cell types and experiment memberships) should induce large variations in the images that are visible on the top principal components. Thus, correcting the coverage along these directions will provide the desired conditional validity. In agreement with this hypothesis, we find that this procedure produces nearly identical results to those seen above, i.e., good coverage across all groups.

\section{Choosing the function class and regularizer}\label{sec:fun_choice}

In order to implement our method in practice, users must choose both the function class, $\mathcal{F}$, and regularizer, $\mathcal{R}(\cdot)$. In the real data examples above we have considered some sample choices for these quantities that were designed to illustrate the guarantees of our method. A practitioner may reasonably disagree with our selections, and through their choice of $\mathcal{F}$ and $\mathcal{R}(\cdot)$, look to prioritize different conditional targets. As in many statistical estimation problems, we expect the best performance of our method to be obtained when the choices of $\mathcal{F}$ and $\mathcal{R}$ are guided by domain-specific knowledge of the most important features to the prediction task at hand. To help practitioners in making these choices, we highlight a few considerations that may arise.

First, we remark that although our theory requires fixed choices of $\mathcal{F}$ and $\mathcal{R}$, we find empirically that running cross-validation on the training set does not significantly harm the coverage. To demonstrate this, we consider two different penalized implementations of our method. In the first, we run ridge-regularized linear quantile regression with $\mathcal{F} := \{\beta_0 + X^\top \beta_1 : \beta_0 \in \mmr,\ \beta_1 \in \mmr^d\}$ and $\mathcal{R}(\beta) := \lambda \|\beta\|^2$, while in the second we run kernel quantile regression with Gaussian kernel, $K(x,y) := \exp(-0.025\, \|x-y\|_2^2)$ and corresponding kernel-norm regularizer, $\mathcal{R}(g) := \lambda \|g\|_{K}^2$. \Cref{fig:cross_val_cov_comp} shows a comparison of the marginal coverage of our method obtained when $\lambda$ is either fixed in advance, or estimated using cross-validation. We see that although our theory requires $\lambda$ to be fixed, both approaches give near exact coverage empirically. 

\begin{figure}[ht]
    \centering
    \includegraphics[width = \textwidth]{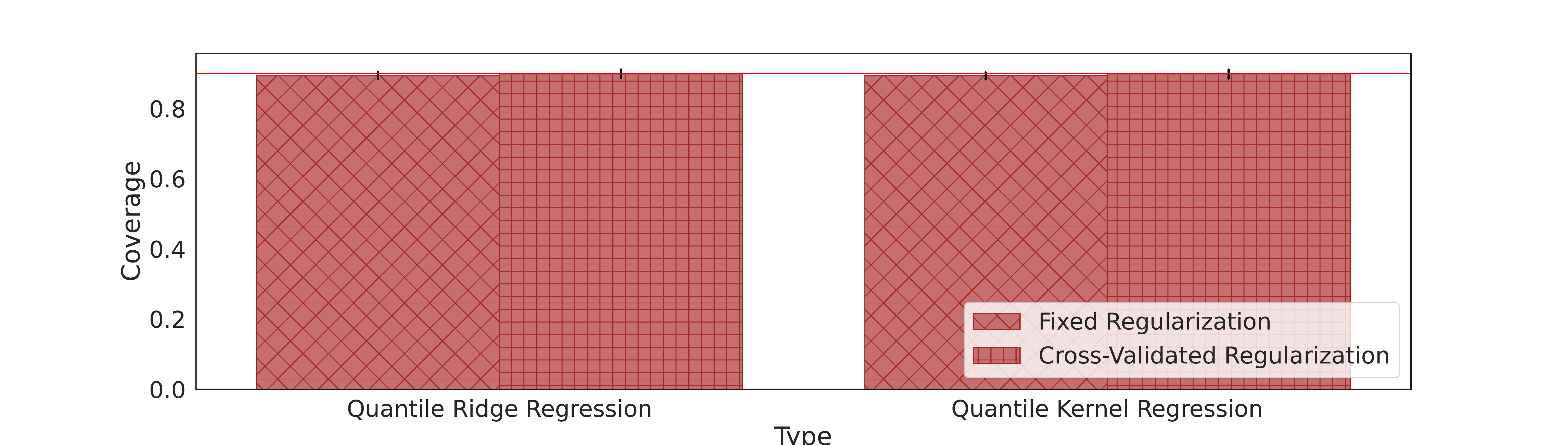}
    \caption{Marginal coverage obtained with (vertical and horizontal hatches) and without (diagonal hatches) cross-validation for the randomized version of our method. The red line shows the target level of $\alpha = 0.1$. Data for this experiment were generated from the Gaussian linear model in which $(X_i, \epsilon_i) \sim N(0,I_d)$, $Y_i = X_i^\top w + \epsilon_i$, and $w$ is a unit vector sampled uniformly at random. The conformity score is taken to be simply equal to $Y$ so that all of the data is used in the calibration set. Barplots show averages from $20$ trials where in each trial we sample $n=200$ calibration points and $100$ test points in dimension $d=50$.}
    \label{fig:cross_val_cov_comp}
\end{figure}

Perhaps even more critical than the regularization level, is the choice of function class $\mathcal{F}$. Of particular interest is the trade-off that occurs as the dimension of $\mathcal{F}$ increases. Indeed, while larger function spaces allow for a more rich set of conditional guarantees, we find empirically that they also lead to larger prediction sets. Thus, adding spurious features to $\mathcal{F}$ can harm the efficiency of our predictions.

\begin{figure}[ht]
    \centering
    \includegraphics[width = \textwidth]{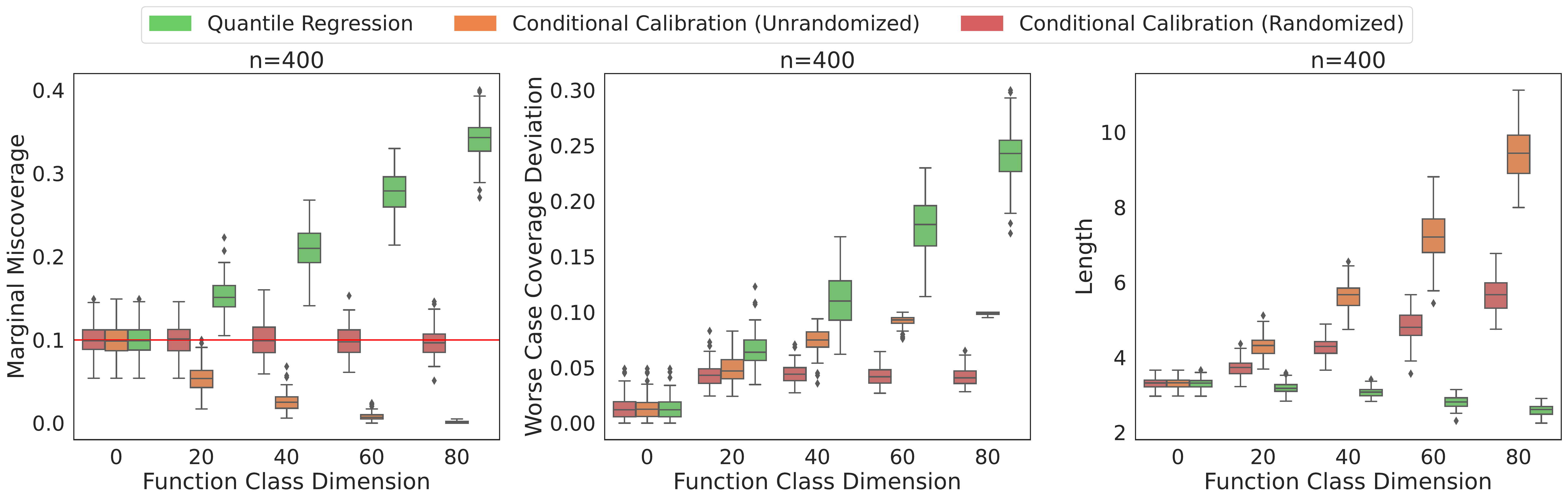}
    \caption{Estimated distributions of the calibration-conditional marginal miscoverage, worst-case conditional coverage deviation, and prediction set length of quantile regression (green), as well as the unrandomized (orange) and randomized (red) variants of our conditional calibration method. Boxplots show results from 100 trials, where in each trial the coverages and lengths are evaluated empirically on 1000 test points.}
    \label{fig:large_dim_tradeoffs}
\end{figure}

To demonstrate this phenomenon, we consider a dataset with many irrelevant features that have no relationship with the response. In particular, we work with a Gaussian model in which $X_i \sim N(0,I_d)$, $Y_i \sim N(0,1)$, and $Y_i$ is independent of $X_i$. We implement our method with conformity score $S(X,Y) = Y$ and use linear quantile regression with $\mathcal{F} = \{\beta_0 + X^\top \beta_1 : \beta_0 \in \mmr,\ \beta_1 \in \mmr^d\}$, $\mathcal{R}(g) = 0$. \Cref{fig:large_dim_tradeoffs} plots estimated distributions of the calibration-conditional marginal miscoverage, worst-case conditional coverage deviation, and length of a two-sided implementation of our method in which the lower $\alpha/2$ and upper $1-\alpha/2$ quantiles are estimated separately (see \Cref{sec:app_two-sided}), i.e. the quantities,
\begin{equation*}
  \begin{gathered}
\mmp(Y_{n+1} \notin \hat{C}(X_{n+1}) \mid  \{(X_i,Y_i)\}_{i=1}^n), \ \sup_{1 \leq j \leq d+1} \left|\frac{\mme[(1,X_{n+1})_j(\bone\{Y_{n+1}\not\in\hat{C}(X_{n+1})\} - \alpha) \mid  \{(X_i,Y_i)\}_{i=1}^n]}{\mme[|(1,X_{n+1})_j|]} \right|,\\ \text{ and }
\mme[\text{length}(\hat{C}(X_{n+1})) \mid \{(X_i,Y_i)\}_{i=1}^n] .
  \end{gathered}
\end{equation*}
As a baseline, the figure also displays results for vanilla quantile regression. We find that while the marginal coverage of our method is robust to large function classes, both the worst-case coverage deviation and the average length of the prediction sets increase as the dimension grows. Recall that all the covariates here are spurious; therefore, the increase in these quantities is purely a result of over-fitting in high dimensions and not a reflection of the underlying complexity of the relationship between $X$ and $Y$. Finally, we remark that the growth in the worst-case coverage error is substantially worse for quantile regression compared to our method. Thus, while we cannot guarantee exact calibration-conditional validity over all $f \in \mathcal{F}$ in high dimensions, our method is practically superior to sensible alternatives.

\section{Acknowledgments}

E.J.C. was supported by the Office of Naval Research grant N00014-20-1-2157, the National Science Foundation grant DMS-2032014, the Simons Foundation under award 814641, and the ARO grant 2003514594. I.G. was also supported by the Office of Naval Research grant N00014-20-1-2157 and the Simons Foundation award 814641, as well as additionally by the Overdeck Fellowship Fund. J.J.C. was supported by the John and Fannie Hertz Foundation. The authors are grateful to Kevin Guo and Tim Morrison for helpful discussion on this work.

\newpage

\bibliographystyle{rss}
\bibliography{CondCovPaper}

\newpage

\appendix

\section{Appendix}

\subsection{Computational details for Section~\ref{sec:computation}}\label{sec:app_comp_set-up}

\subsubsection{Verification of the conditions of Section~\ref{sec:computation}}

In this section we verify that all of the quantile regression problems discussed in this paper satisfy the conditions of \Cref{sec:computation}. To do this we will check that 1) each quantile regression admits an equivalent finite dimensional representation 2) all of these finite dimensional representations (and thus also their infinite dimensional counterparts) satisfy strong duality.

The fact that the linear quantile regression of \Cref{sec:finite_dim} satisfies 1) and 2) is clear. For RKHS functions, let $K \in \mmr^{n+1 \times n+1}$ denote the kernel matrix with entries $K_{ij} = K(X_i,X_j)$. Let $K_i$ denote the $i_{\text{th}}$ row (equivalently column) of $K$. Then, the primal problem \eqref{eq:general_method} can be fit by solving the equivalent convex optimization program
\begin{align*}
  &  \underset{\gamma \in \mmr^{n+1},\ \beta \in \mmr^d}{\text{minimize}}\ \frac{1}{n+1} \sum_{i=1}^n \ell_{\alpha}(K_i^\top\gamma + \Phi(X_i)^\top\beta,S_i) + \frac{1}{n+1}\ell_{\alpha}(K_{n+1}^\top\gamma + \Phi(X_{n+1})^\top\beta,S) + \lambda \gamma^\top K \gamma,
\end{align*}
and setting $\hat{g}_S(\cdot) = \sum_{i=1}^{n+1}\hat{\gamma}_i K(X_i,\cdot) + \Phi(\cdot)^{\top} \hat{\beta}$, for $(\hat{\gamma},\hat{\beta})$ any optimal solutions. With this finite-dimensional representation in hand, it now follows that the primal-dual pair must satisfy strong-duality by Slater's condition. Moreover, it is also easy to see that in this context 
\begin{align*}
\mathcal{R}^*(\eta) & = - \min_{g_K \in \mathcal{F}_K,\ \beta \in \mmr^d} \lambda\|g_K\|_K^2 - \sum_{i=1}^{n+1} \eta_i (g_K(X_i) + \Phi(X_i)^\top \beta)\\
& = - \min_{\gamma \in\mmr^{n+1},\ \beta \in \mmr^d} \lambda \gamma^\top K \gamma - \sum_{i=1}^{n+1} \eta_i (K_i^\top\gamma + \Phi(X_i)^\top \beta),
\end{align*}
which is a tractable function that we can compute.

Finally, to fit Lipschitz functions we can solve the optimization program,
\begin{align*}
& \underset{\gamma \in \mmr^{n+1},\ \beta \in \mmr^d}{\text{minimize}}\ \frac{1}{n+1} \sum_{i=1}^n \ell_{\alpha}(\gamma_i + \Phi(X_i)^\top \beta,S_i) + \frac{1}{n+1} \ell_{\alpha}(\gamma_{n+1} + \Phi(X_{n+1})^\top \beta,S) + \lambda \max_{i \neq j}\frac{|\gamma_i - \gamma_j|}{\|X_i - X_j\|_2}.
\end{align*}
The idea here is that ${\gamma}_1,\dots,{\gamma}_{n+1}$ act as proxies for the values of ${f}_L(X_1),\dots,{f}_L(X_{n+1})$. These proxies can always be extended to a complete function on all of $\mathcal{X}$ using the methods of \citet{McShane1934, Whitney1934}. Once again, with this finite-dimensional representation in hand it is now easy to check that the primal-dual pair satisfies strong-duality by Slater's condition. Finally, in this context we have 
\begin{align*}
\mathcal{R}^*(\eta) & = - \min_{g_L \in \mathcal{F}_L,\ \beta \in \mmr^d} \lambda \text{Lip}(g_L) - \sum_{i=1}^{n+1} \eta_i (g_L(X_i) + \Phi(X_i)^\top \beta)\\
& = - \min_{\gamma \in\mmr^{n+1},\ \beta \in \mmr^d} \lambda \max_{i \neq j}\frac{|\gamma_i - \gamma_j|}{\|X_i - X_j\|_2} - \sum_{i=1}^{n+1} \eta_i \gamma_i,
\end{align*}
which is a tractable function.

\subsubsection{Additional algorithmic details}

The results below elucidate various approaches to computing the prediction set.  \Cref{alg:binary_search} describes how we might use a binary search algorithm to discover the critical $S$ at which $\eta^S_{n + 1}$ equals the target cutoff. 
\begin{algorithm}
 \KwData{Observed data $\{(X_1,S_1),\dots,(X_n,S_n)\} \cup \{X_{n+1}\}$, numerical error tolerance $\epsilon$, target cutoff $C \in \{1-\alpha,U\}$, range $[a,b]$ for $S_{n+1}$ ($a, b$ = \texttt{None} indicates that no bounds are known for $S_{n+1}$).}
 \If{$b =$ \texttt{None}}{
    $b = \max\{\max_{1 \leq i \leq n} S_i,1\}$\;
    \While{$\eta^b_{n+1} < C$}{
        $b=2b$;
    }
}
 \If{$a =$ \texttt{None}}{
    $a = \min\{\min_{1 \leq i \leq n} S_i, -1\}$\;
    \While{$\eta^a_{n+1} \geq C$}{
        $a=2a$;
    }
}
 \While{$b-a < \epsilon$}{
        \If{$\eta^{(a+b)/2}_{n+1} < C$}{
            $a = (a+b)/2$\;
        }\Else{
            $b = (a+b)/2$\;
        }
 }
 \Return $(a+b)/2$.
 \caption{\textit{Binary Search Computation of $S^*_{n + 1}$}}
 \label{alg:binary_search}
\end{algorithm}

As discussed in the main text, \Cref{alg:binary_search} may be inefficient if the optimization problem possesses additional structure. In particular, a common implementation of our method is to fit an unregularized quantile regression over a finite-dimensional function class. In this case, it is possible to exactly compute $\inf \{S : \eta^S_{n + 1} < C\}$. Our approach to this problem, which leverages standard tools from LP sensitivity analysis, relies on the following observation. 
\begin{proposition}\label{prop:zero_sol}
    Assume strong duality holds for \eqref{eq:generic_convex_opt}. Let $(\hat{g}_n, \hat{\eta}^n) \in \mathcal{F} \times \mmr^n$ denote a valid primal-dual solution of \eqref{eq:generic_convex_opt},  \eqref{eq:generic_dual} and $\hat{S}_{n + 1} := \hat{g}_n(X_{n + 1})$. Then, 
    $(\hat{\eta}^n,0)$ is a valid dual solution of the dual program with data $\{(X_1,S_1)\}_{i=1}^n \cup \{(X_{n+1}, \hat{S}_{n+1})\}$.
\end{proposition}
\begin{proof}
This result follows immediately from the fact that $(\hat{\eta}^n,0)$ is feasible and the duality gap of $\hat{g}_n$ and $(\hat{\eta}^n,0)$ is zero.
\end{proof}

We now recall some basic linear programming (LP) facts. First, assuming without loss of generality that $\Phi$ is full-rank, it follows immediately from standard LP theory that there exists a solution in which all but $p$ indices in $\eta$ line at the extremes $\{\alpha, 1 - \alpha\}$ (\cite{dantzig2003linear}). Moreover, the remaining non-trivial coordinates of $\eta$ are piecewise linear in $S$, and the slope of the function given by $S \mapsto \eta^S_{n + 1}$ can be obtained by solving a linear equation (\cite{dantzig2003linear}). Thus, obtaining the critical value of $S$ at which $\eta^S_{n + 1}$ exceeds our cutoff amounts to tracing a piecewise linear function from $\hat{S}_{n + 1}$ until the cutoff is reached. \Cref{alg:lp_sensitivity} provides an explicit description of this procedure, which closely resembles the simplex method. Note that the algorithm is initialized at $\hat{S}_{n + 1}$ since the dual solution at that point is known to be zero (\Cref{prop:zero_sol}).

\begin{algorithm}
 \KwData{Observed data $\{(X_1,S_1),\dots,(X_n,S_n)\} \cup \{X_{n+1}\}$, target cutoff $C$, quantile $q$.}
 $\hat{g}_n(\cdot) = \text{LPSolver}(\{(X_i, S_i)\}_{i = 1}^n)$ \;
 $S^*_{n + 1} = \hat{g}_n(X_{n + 1})$, $\eta_{n + 1} = 0$\;
 $i_c = n + 1$ \tcc*[r]{candidate to enter basis}
 \While{$\eta_{n + 1} \neq C$}{
    $B = \text{ActiveSet}(\{1,\dots,n + 1\})$\;
    $d = -(\Phi_B^\top)^{-1} \Phi_{i^*}^\top$\;
    \If{$C > \eta_{n + 1}$} {
        $\Delta_{1:|B|} = \max \left(\frac{q - \eta_B}{d}, \frac{q - 1 - \eta_B}{d} \right)$ \tcc*[r]{element-wise division and maximum}
        $i^* = \arg \min_{i \in [|B|]} (\Delta_i)$ \tcc*[r]{candidate to exit basis}
        $\delta = \Delta_{i^*}$\;
    }
    \Else {
        $\Delta_{1:|B|} = \min \left(\frac{q - \eta_B}{d}, \frac{q - 1 - \eta_B}{d} \right)$ \tcc*[r]{element-wise division and minimum}
        $i^* = \arg \max_{i \in [|B|]} (\Delta_i)$ \tcc*[r]{candidate to exit basis}
        $\delta = \Delta_{i^*}$ \;
    }
    $\delta_c = \max(\min(\delta, q - 1 - \eta_{i_c}), q - \eta_{i_c})$\;
    $\eta_B = \eta_B + \delta_c  d$\;
    $\eta_{i_c} = \eta_{i_c} + \delta_c$\;
    \If {$\delta_c = \delta$} { 
        $B = B \setminus i^*$ \tcc*[r]{update basis} 
        $B = B \cup \{i_c\}$ \tcc*[r]{update basis} 
    }
    $A = (\Phi_B^\top)^{-1} \Phi_{B^c}^\top$\;
    $c = S_{B^c}^\top - S_B^\top A$\;
    $\Delta^S = c / A_{-1}$ \tcc*[r]{element-wise division by last row of $A$} 
    $B^* = \{i \in B \mid c_i \neq 0\}$ \;
    \If {$|B^*| = 0$} {
        $S_{n + 1} = \infty$ \;
        $\eta_{n + 1} = C$\;
    }
    \Else {
        \If {$C > 0$} {
            $i_c = \arg \min_{i \in B^*} \Delta^S_i $ \tcc*[r]{candidate to enter basis}
            $S^*_{n + 1} = S^*_{n + 1} + \Delta^S_{i_c}$ \;
        } \Else {
            $i_c = \arg \max_{i \in B^*} \Delta^S_i $ \tcc*[r]{candidate to enter basis}
            $S^*_{n + 1} = S^*_{n + 1} + \Delta^S_{i_c}$ \;
        }
    }
 }
 \Return $S^*_{n + 1}$.
 \caption{\textit{LP Sensitivity Computation of $S^*_{n + 1}$}}
 \label{alg:lp_sensitivity}
\end{algorithm}

We conclude this section with one final method for computing a conservative alternative to our prediction set that requires only a single quantile regression fit. For this, suppose we know an upper bound $M$ for $S_{n+1}$. Then, we may consider the conservative prediction set $\hat{C}_{\text{con}}(X_{n+1},M) := \{y : S(X_{n+1},y) \leq \hat{g}_M(X_{n+1})\}$. Our final proposition shows that whenever $M$ is a valid upper bound on $S_{n+1}$, $\hat{C}_{\text{con}}(X_{n+1},M)$ provides a conservative coverage guarantee.  To avoid subtle issues here related to the non-uniqueness of the optimal quantile function, we will assume here that $\hat{g}_M(X_{n + 1})$ is a min-norm solution of \eqref{eq:generic_convex_opt} (the choice of norm is not important).

\begin{proposition}\label{prop:conservative_computation}
Assume that $\mathcal{F}$ admits a strictly-convex norm $\|\cdot\|_{\mathcal{F}}$ and let $\{\hat{g}_S\}_{S \in \mmr}$ denote the corresponding min-norm solutions to \eqref{eq:generic_convex_opt}. Then, for all $M>0$, 
\[
\{y : S(X_{n+1},y) \leq \hat{g}_{S(X_{n+1},y)} (X_{n+1}),\ S \leq M\} \subseteq \{y : S(X_{n+1},y) \leq \hat{g}_{M} (X_{n+1})\}.
\]
\end{proposition}

\begin{proof}[Proof of \Cref{prop:conservative_computation}]
 Let
\[
L_{n}(g) :=  \sum_{i=1}^{n} \ell_{\alpha}(g(X_i),S_i) \ \ \text{ and } \ \ L_{S}(g) = L_n(g) + \ell_{\alpha}(g(X_{n+1}),S).
\]
Assume for the sake of contradiction that $S < M$, $S \leq \hat{g}_S(X_{n+1})$, and $\hat{g}_M(X_{n + 1}) < S$. To obtain a contradiction, we claim that it is sufficient to prove that
\begin{align}
    L_S(\hat{g}_M) - L_S(\hat{g}_S) \leq L_M(\hat{g}_M) - L_M(\hat{g}_S). \label{eq:suff_cond}
\end{align}
To see why, note that since $\hat{g}_M$ is a global optimum of $g \ \mapsto L_M(g) + (n+1)\mathcal{R}(g)$, we must have that
\begin{align*}
    L_M(\hat{g}_M) + (n+1)\mathcal{R}(\hat{g}_M) &\leq L_M(\hat{g}_S) + (n+1)\mathcal{R}(\hat{g}_S).
\end{align*}
Rearranging this and applying \eqref{eq:suff_cond} gives the inequality
\begin{align*}
    (n+1)\mathcal{R}(\hat{g}_S) - (n+1)\mathcal{R}(\hat{g}_M) &\geq L_M(\hat{g}_M) - L_M(\hat{g}_S) \\
    & \geq 
    L_S(\hat{g}_M) - L_S(\hat{g}_S),
\end{align*}
or equivalently,
\begin{align*}
    L_S(\hat{g}_M) + (n+1)\mathcal{R}(\hat{g}_M) & \leq  L_S(\hat{g}_S) + (n+1)\mathcal{R}(\hat{g}_S).
\end{align*}
Since $\hat{g}_S$ is the unique min-norm minimizer of $g \mapsto L_S(g) + (n+1)\mathcal{R}(g)$ this implies that $\|\hat{g}_S\|_{\mathcal{F}} < \hat{g}_M\|_{\mathcal{F}}$. 

Now, by a completely symmetric argument reversing the roles of $\hat{g}_M$ and $\hat{g}_S$ we also have that
\[
 L_M(\hat{g}_S) + (n+1)\mathcal{R}(\hat{g}_S)  \leq  L_M(\hat{g}_M) + (n+1)\mathcal{R}(\hat{g}_M),
\]
which by identical reasoning implies that $\|\hat{g}_M\|_{\mathcal{F}} < \|\hat{g}_S\|_{\mathcal{F}}$. Thus, we have arrived at our desired contradiction.

To prove \eqref{eq:suff_cond} we break into two cases.
\paragraph{Case 1:} $\hat{g}_M(X_{n + 1}) < S \leq \hat{g}_S(X_{n + 1}) \leq M$. 
\begin{align*}
    L_S(\hat{g}_M) - L_S(\hat{g}_S) &= (1 - \alpha) (S - \hat{g}_M(X_{n + 1})) - \alpha (\hat{g}_S(X_{n + 1}) - S) + L_n(\hat{g}_M) - L_n(\hat{g}_S)  \\
    &= (1-\alpha)(\hat{g}_S(X_{n + 1}) - \hat{g}_M(X_{n + 1})) + S - \hat{g}_S(X_{n + 1}) + L_n(\hat{g}_M) - L_n(\hat{g}_S) \\
    & \leq (1-\alpha) (\hat{g}_S(X_{n + 1}) - \hat{g}_M(X_{n + 1})) + L_n(\hat{g}_M) - L_n(\hat{g}_S) \\
    &= L_M(\hat{g}_M) - L_M(\hat{g}_S)
\end{align*}

\paragraph{Case 2:} $\hat{g}_M(X_{n + 1}) < S < M \leq \hat{g}_S(X_{n + 1})$.
\begin{align*}
    L_M(\hat{g}_M) - L_M(\hat{g}_S) &= (1-\alpha)(M - \hat{g}_M(X_{n+1})) - \alpha(\hat{g}_S(X_{n+1}) - M) + L_n(\hat{g}_M) - L_n(\hat{g}_S)\\
    & = \alpha (\hat{g}_M(X_{n + 1}) - \hat{g}_S(X_{n + 1})) + M - \hat{g}_M(X_{n + 1}) + L_n(\hat{g}_M) - L_n(\hat{g}_S) \\
    &> \alpha (\hat{g}_M(X_{n + 1}) - \hat{g}_S(X_{n + 1})) + S - \hat{g}_M(X_{n + 1}) + L_n(\hat{g}_M) - L_n(\hat{g}_S) \\
    &= (1 - \alpha)(S - \hat{g}_M(X_{n + 1})) - \alpha (\hat{g}_S(X_{n + 1}) - S) + L_n(\hat{g}_M) - L_n(\hat{g}_S) \\
    &= L_S(\hat{g}_M) - L_S(\hat{g}_S).
\end{align*}

\end{proof}

\subsection{Additional experiments on the Rxrx1 data}\label{sec:app_rx1_extra_exp}

As an alternative to estimating the probabilities of experimental membership, here we consider constructing a feature representation for the Rxrx1 data using principal component analysis. Namely, we implement the linear quantile regression method of \Cref{sec:finite_dim} using the top principal components of the feature layer of the neural network as input. We choose the number of principal components for this analysis to be 70 based off of a visual inspection of a scree plot (\Cref{fig:scree_plot}). All other steps of this experiment are kept identical to the procedure described in \Cref{sec:cell_data}. 

Similar to the results of \Cref{sec:cell_data}, we find that the empirical coverage of this method is close to the target level across all cell types and experiments (\Cref{fig:rx1_pca_cov}).

\begin{figure}[ht]
    \centering
    \includegraphics[scale=0.29]{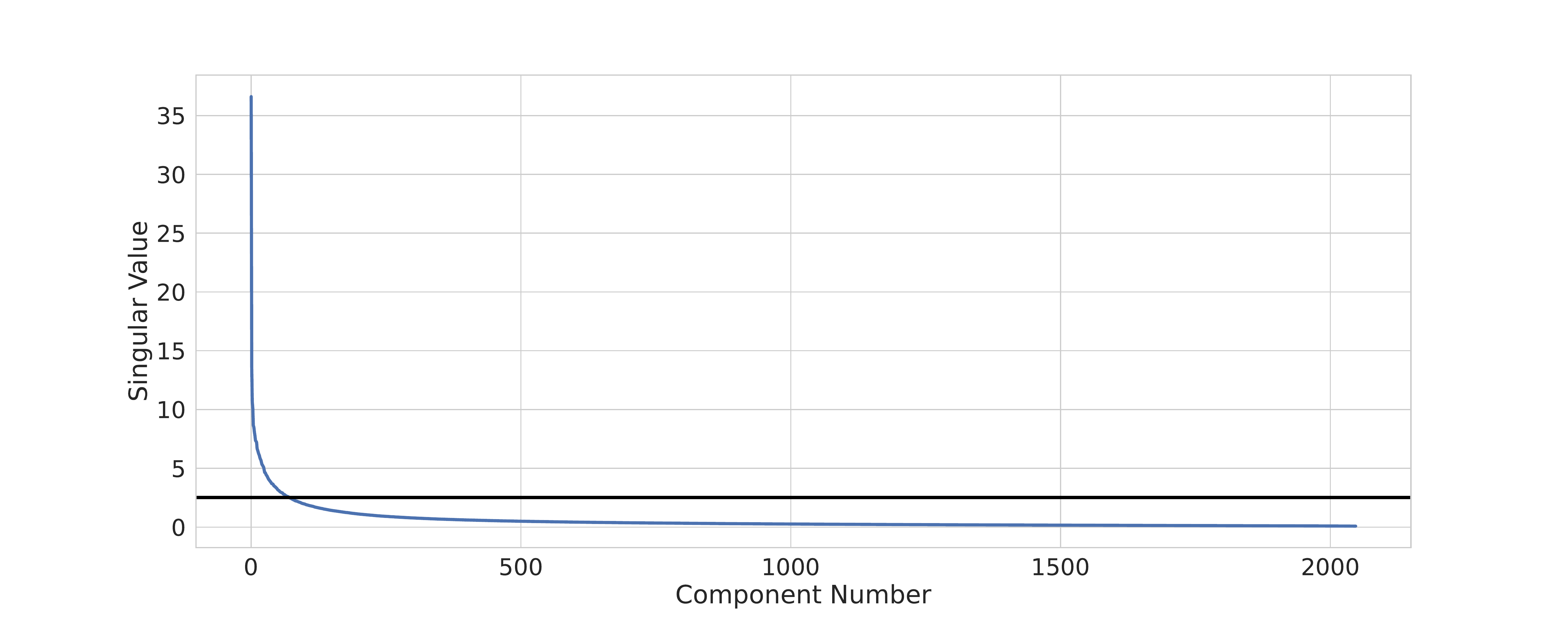}
    \caption{Scree plot for the singular value decomposition of the feature layer of the pretrained neural network. The horizontal line indicates the value of the 70th largest component.}
    \label{fig:scree_plot}
\end{figure}

\begin{figure}[ht!]
    \centering
    \includegraphics[scale=0.29]{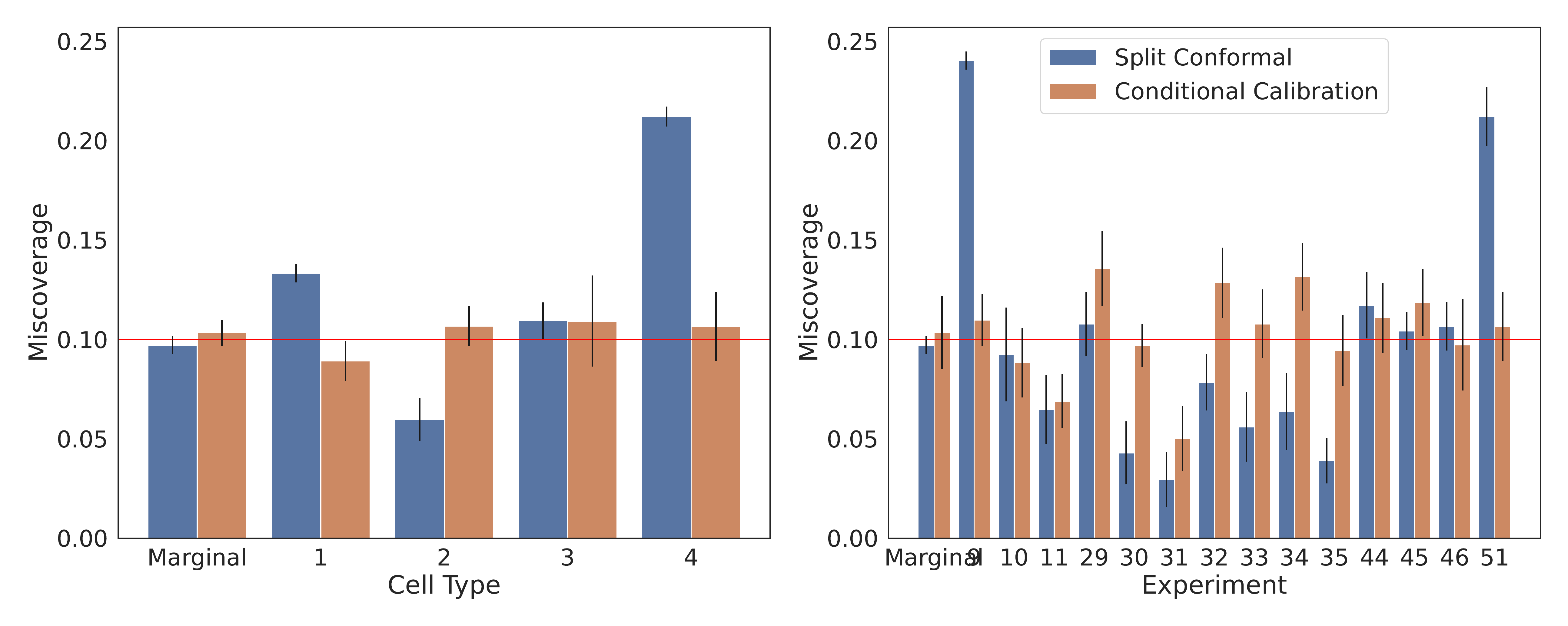}
    \caption{Empirical conditional miscoverage of the unrandomized version of our method (orange) and split conformal (blue) across cell types and experiments when principal component analysis is used to construct the features. Red lines indicate the target level of $\alpha = 0.1$ and black error bars show 95\% binomial confidence intervals for the calibration-conditional miscoverage $\mmp(Y_{n+1} \notin \hat{C}(X_{n+1}) | D_{\text{train}}, D_{\text{cal}})$, where $D_{\text{train}}$ and $D_{\text{cal}}$ denote the training dataset used to learn the feature representation and the calibration dataset used to implement our method, respectively.}
    \label{fig:rx1_pca_cov}
\end{figure}

\subsection{Proofs of the main coverage guarantees}

In this section we prove the top-level coverage guarantees of our method. We begin by proving our most general result, \Cref{thm:infinite_dim_result}, which considers a generic function class $\mathcal{F}$ and penalty $\mathcal{R}$. Then, by restricting the choices of $\mathcal{F}$ and $\mathcal{R}$, we obtain \Cref{thm:finite_dim_result} and \Cref{cor:group_coverage} as special cases.

\begin{proof}[Proof of \Cref{thm:infinite_dim_result}]
We begin by examining the first order conditions of the convex optimization problem \eqref{eq:general_method}. Namely, since $\hat{g}_{S_{n+1}}$ is a minimizer of 
\[
g \mapsto \frac{1}{n+1} \sum_{i=1}^{n+1} \ell_{\alpha}(g(X_i),S_i) + \mathcal{R}(g)
\]
we must have that for any fixed $f \in \mathcal{F}$, 
\[
0 \in \partial_{\epsilon} \left(\frac{1}{n+1}\sum_{i=1}^{n+1} \ell_{\alpha}(\hat{g}_{S_{n+1}}(X_i) + \epsilon f(X_i),S_i) + \mathcal{R}(\hat{g}_{S_{n+1}} + \epsilon f)\right)\bigg|_{\epsilon = 0}.
\]
By a straightforward computation, the subgradients of the pinball loss are given by
\begin{multline*}
\partial_{\epsilon} \left(\frac{1}{n+1}\sum_{i=1}^{n+1} \ell_{\alpha}(\hat{g}_{S_{n+1}}(X_i) + \epsilon f(X_i),S_i) \right)\\ = \left \{ \frac{1}{n + 1} \left(\sum_{S_i \neq \hat{g}_{S_{n + 1}}(X_i)} f(X_{i})\left (\alpha - \indic{S_i > \hat{g}_{S_{n + 1}}} \right)  + \sum_{S_i = \hat{g}_{S_{n + 1}}(X_i)} s_i f(x_i) \right) \,\biggr \lvert \,s_i \in 
    [\alpha - 1, \alpha]  \right \}.
\end{multline*}
Let $s_i^* \in [\alpha - 1,\alpha]$ be the values setting the subgradient to $0$. Rearranging, we obtain
\begin{align*}
 \frac{1}{n+1} \sum_{i=1}^{n+1} f(X_{i})\left (\alpha - \indic{S_i > \hat{g}_{S_{n + 1}}(X_i)}  \right)  & = \frac{1}{n + 1} \sum_{S_i = \hat{g}_{S_{n+1}}(X_i)} (\alpha - s^*_i) f(x_i) - \frac{d}{d\epsilon}R\left(\hat{g}_{S_{n+1}} + \epsilon f\right)\bigg|_{\epsilon = 0}.
\end{align*}
Our desired result now follows from the following observations. First, observe that the LHS above can be related to our desired coverage guarantee through the equation,
\begin{align*}
        \mme[f(X_{n+1})(\bone\{Y \in \hat{C}_{n+1}\} - (1- \alpha))] &= \mme[f(X_{n+1})(\alpha - \bone\{Y \notin \hat{C}_{n+1}\} )] \\
&= \mme[f(X_{n+1})(\alpha - \bone\{S_{n+1} > \hat{g}_{S_{n+1}}(X_{n+1})\} )].
\end{align*}
Moreover, since $\hat{g}_{S_{n+1}}$ is fit symmetrically, i.e., is invariant to permutations of the input data, we have that $\{(f(X_i),\hat{g}_{S_{n+1}}(X_i),S_i)\}_{i = 1}^{n + 1}$ are exchangeable. Thus, we additionally have that
\begin{equation}\label{eq:deriv_rearrangment}
\begin{split}
\mme[f(X_{n+1})(\alpha - & \bone\{S_{n+1} > \hat{g}_{S_{n+1}}(X_{n+1})\} )] = \mme\left[ \frac{1}{n+1} \sum_{i=1}^{n+1} f(X_{i})\left (\alpha - \indic{S_i > \hat{g}_{S_{n + 1}}(X_i)}  \right)  \right] \\
    & =  \mme\left[\frac{1}{n+1} \sum_{i=1}^{n+1} (\alpha - s_i^*) f(X_i) \bone\{S_i = \hat{g}_{S_{n+1}}(X_i)\} \right]  - \mme\left[ \frac{d}{d\epsilon}R\left(\hat{g}_{S_{n+1}} + \epsilon f\right)  \bigg|_{\epsilon = 0} \right].
\end{split}
\end{equation}
Finally, since $\alpha - s_i^* \in [0,1]$, we can bound the first term as 
\begin{align*}
\left| \mme\left[\frac{1}{n+1} \sum_{i=1}^{n+1} (\alpha - s_i^*) f(X_i) \bone\{S_i = \hat{g}_{S_{n+1}}(X_i)\} \right] \right|  &\leq \mme\left[\frac{1}{n+1} \sum_{i=1}^{n+1} | f(X_i)| \bone\{S_i = \hat{g}_{S_{n+1}}(X_i)\} \right] \\ 
&= \mme[| f(X_i)| \bone\{S_i = \hat{g}_{S_{n+1}}(X_i)\} ],
\end{align*}
where last step follows by exchangeability. This proves the second claim of \Cref{thm:infinite_dim_result}. 

To get the first claim of \Cref{thm:infinite_dim_result}, note that when $f$ is non-negative we can lower bound $\frac{1}{n + 1} \sum_{i = 1}^{n + 1} (\alpha - s_i^*) f(X_i) \bone\{S_i = \hat{g}_{S_{n + 1}} (X_i)\}$ by $0$. Plugging this into \eqref{eq:deriv_rearrangment} gives us the desired inequality,
\begin{align*}
  \mme[f(X_{n+1})(\bone\{Y \in \hat{C}_{n+1}\} - (1- \alpha))] \geq  - \mme\left[ \frac{d}{d\epsilon}R\left(\hat{g}_{S_{n+1}} + \epsilon f\right) \bigg|_{\epsilon = 0}  \right].
\end{align*}
\end{proof}

With the proof of \Cref{thm:infinite_dim_result} in hand, we are now ready to prove the special cases stated in \Cref{thm:finite_dim_result} and \Cref{cor:group_coverage}.

\begin{proof}[Proof of \Cref{thm:finite_dim_result}]
    The first statement of \Cref{thm:finite_dim_result} follows immediately from the first statement of \Cref{thm:infinite_dim_result} in the special case where $\mathcal{R}(\cdot) = 0$. 
    
    To get the second statement of \Cref{thm:finite_dim_result} note that the second statement of \Cref{thm:infinite_dim_result} tells us that for any $f \in \mathcal{F}$,
    \[
    \mme[f(X_{n+1})(\bone\{Y_{n+1} \in \hat{C}(X_{n+1})\} - (1-\alpha))] \leq \mme[|f(X_i)| \bone\{S_i = \hat{g}_{S_{n+1}}(X_i)\}].
    \]
    So, to complete the proof we just need to show that when the distribution of $S \mid X$ is continuous 
    \[
    \mme[|f(X_i)| \bone\{S_i = \hat{g}_{S_{n+1}}(X_i)\}] \leq \frac{d}{n+1} \mme\left [\max_{1 \leq i \leq n+1}|f(X_{n+1})| \right ].
    \]
    To do this, first note that by the exchangeability of $\{(f(X_i),\hat{g}_{S_{n+1}}(X_i),S_i)\}_{i=1}^{n+1}$ we have 
    \begin{align*}
     \mme[|f(X_i)| \bone\{S_i = \hat{g}_{S_{n+1}}(X_i)\}] & = \mme\left[\frac{1}{n+1}\sum_{i=1}^{n+1} |f(X_i)|\bone\{S_i = \hat{g}_{S_{n+1}}(X_i)\}  \right]\\
     & \leq \mme\left[ \left (\max_{1 \leq j \leq n+1} |f(X_j)| \right ) \cdot  \frac{1}{n+1}\sum_{i=1}^{n+1}\bone\{S_i = \hat{g}_{S_{n+1}}(X_i)\} \right].
    \end{align*}
    Moreover, recalling that $\hat{g}_{S_{n+1}}(X_i) = \Phi(X_i)^\top\hat{\beta}$ for some $\hat{\beta} \in \mmr^d$ we additionally have that
    \begin{align*}
        & \mmp\left(\frac{1}{n+1}\sum_{i=1}^{n+1} \bone\{S_i = \hat{g}_{S_{n+1}}(X_i)\} > d \mid  X_1,\dots,X_{n+1}\right)\\
        & = \mmp\left(\exists\,1 \leq j_1< \dots < j_{d+1} \leq n+1 \text{ such that } \forall\,1 \leq i \leq d + 1,\ S_{j_i} = \hat{g}_{S_{n+1}}(X_{j_i}) \mid X_1,\dots,X_{n+1}\right)\\
        & \leq \sum_{1 \leq j_1 <\dots <j_{d+1} \leq n+1} \mmp\left(\exists\, \beta \in \mmr^d,\ \text{such that } \forall\, 1 \leq i \leq d+1,\ S_{j_i} =  \Phi(X_{j_i})^\top\beta \mid X_1,\dots,X_{n+1} \right)\\
        & \leq \sum_{1 \leq j_1 <\dots <j_{d+1} \leq n+1} \mmp\left( (S_{j_1},\dots,S_{j_{d+1}}) \in \text{RowSpace}([\Phi(X_{j_1}) | \dots |\Phi(X_{j_{d+1}})]) \mid X_1,\dots,X_{n+1} \right)\\
        & = 0,
    \end{align*}
    where the last line follows from the fact that $(S_{j_1},\dots,S_{j_{d+1}})  \mid X_1,\dots,X_{n+1} $ are independent and continuously distributed and $\text{RowSpace}([\Phi(X_i) \mid \dots |\Phi(X_{d+1}))]^\top)$ is a $d$-dimensional subspace of $\mmr^{d+1}$. From this, we conclude that with probability 1,
    \[
    \frac{1}{n+1}\sum_{i=1}^{n+1} \bone\{S_i = \hat{g}_{S_{n+1}}(X_i)\} \leq \frac{d}{n+1},
    \]
    and plugging this into our previous calculation we arrive at the desired inequality
    \begin{align*}
    \mme[|f(X_i)| \bone\{S_i = \hat{g}_{S_{n+1}}(X_i)\}] & \leq \mme\left[ \left (\max_{1 \leq j \leq n+1} |f(X_j)| \right ) \cdot  \frac{1}{n+1}\sum_{i=1}^{n+1}\bone\{S_i = \hat{g}_{S_{n+1}}(X_i)\} \right]\\
    & \leq \frac{d}{n+1} \mme\left[ \max_{1 \leq i \leq n+1} |f(X_i)|\right].
    \end{align*}

\end{proof}

\begin{proof}[Proof of \Cref{cor:group_coverage}]
    This follows immediately by applying \Cref{thm:finite_dim_result} in the special case where $\mathcal{F} = \{ \sum_{G \in \mathcal{G}} \beta_G \bone \{X \in G\} : \beta_G \in \mmr \}$.
\end{proof}

\subsection{Proofs for RKHS functions}

In this section we prove Propositions \ref{prop:rkhs_bounds} and \ref{prop:rkhs_inner_prod_est}. Throughout both proofs we will let $\kappa^2  := \sup_x K(x,x)$ denote our upper bound on the kernel and (when applicable) $C_{S \mid X} := \sup_{x,s}p_{S_i|X_i=x}(s)$ to denote our upper bound on the density of $S_i|X_i$. Moreover, we will assume that the data satisfies the following set of moment conditions.

\begin{assumption}[Moment conditions for RKHS bounds]\label{ass:moment_conditions}
There exists constants $C_3,C_2,c_2,C_{S},C_f,\rho > 0$ such that
\begin{align*}
& \mme[\|\Phi(X_i)\|^2_2] \leq C_2 d,\ \sup_{f \in \mathcal{F}} \mme[|f(X_i)| \cdot \|\Phi(X_i)\|_2^2] \leq C_3\mme[|f(X)|]d, \ \sup_{\beta : \|\beta\|_2 = 1} \mme[|\Phi(X_i)^\top\beta|^2] \leq c_2,\\
& \sup_{f \in \mathcal{F}}\mme[|f(X_i)| S_i^2] \leq C_{S}\mme[|f(X_i)|],\ \sup_{f \in \mathcal{F}}\sqrt{\mme[|f(X_i)|^2]} \leq C_{f}\mme[|f(X_i)|], \text{ and } \inf_{\beta : \|\beta\|_2 = 1} \mme[|\Phi(X_i)^\top\beta|] \geq \rho.
\end{align*}
Furthermore, we also have that $\mme[|S_i|^2] < \infty$.
\end{assumption}

\subsubsection{Proof of Proposition \ref{prop:rkhs_bounds}}

Our main idea is to exploit the stability of RKHS regression. We will do this using two main lemmas. The first lemma is a canonical stability result first proven in \citet{Bousquet2002} that bounds the sensitivity of the RKHS fit to changes of a single data point. While this result is quite powerful, it is not sufficient for our context because it does not account for the extra linear term $\Phi(X_i)^{\top}\beta$. Thus, we will also develop a second lemma that controls the stability of the fit to changes in $\beta$. 

When formalizing these ideas it will be useful have some additional notation that explicitly separates the dependence of the fit on $\beta$ from the dependence of the fit on the data. Let
\[
\hat{g}_{\beta} := \underset{g_K \in \mathcal{F}_K}{\text{argmin}} \frac{1}{n+1} \sum_{i=1}^{n+1} \ell_{\alpha}(g_K(X_i) + \Phi(X_i)^{\top}\beta, S_i) + \lambda \|g_K\|_K^2,
\]
denote the result of fitting the RKHS part of the function class with $\beta \in \mmr^d$ held fixed. Additionally, let $\{(\tilde{X}_i,\tilde{S}_i)\}_{i=1}^{n+1}$ denote an independent copy of $\{({X}_i,{S}_i)\}_{i=1}^{n+1}$ and for any $A \subseteq \{1,\dots,n+1\}$ define
\[
\hat{g}^{-A}_{\beta} := \underset{g_K \in \mathcal{F}_K}{\text{argmin}} \frac{1}{n+1} \sum_{i \notin A} \ell_{\alpha}(g_K(X_i) + \Phi(X_i)^{\top}\beta, S_i) + \frac{1}{n+1} \sum_{i \in A} \ell_{\alpha}(g_K(\tilde{X}_i) + \Phi(\tilde{X}_i)^{\top}\beta, \tilde{S}_i) + \lambda \|g_K\|_K^2,
\]
to be the leave-$A$-out version of $\hat{g}_{\beta}$ obtained by swapping out $\{(X_i,S_i)\}_{i \in A}$ for $\{(\tilde{X}_i,\tilde{S}_i)\}_{i \in A}$. Our first lemma bounds the difference between $\hat{g}^{-A}_{\beta}$ and $\hat{g}^A_{\beta}$.

\begin{lemma}\label{lem:stab_in_data}
    Assume that $\sup_{x,x} K(x,x) = \kappa^2 < \infty$. Then, for any two datasets $\{(X_i,S_i)\}_{i=1}^{n+1}$ and $\{(\tilde{X}_i,\tilde{S}_i)\}_{i=1}^{n+1}$,
    \[
    \|\hat{g}_{\beta} - \hat{g}^{-A}_{\beta}\|_{\infty} \leq \frac{\kappa^2|A|}{2\lambda (n+1)}.
    \]
\end{lemma}
\begin{proof}
    By a straightforward calculation one can easily check that $\ell_{\alpha}(g_K(X_i) + \Phi(X_i)^{\top}\beta, S_i)$ is a 1-Lipschitz function of $S_i - g_K(X_i) -  \Phi(X_i)^{\top}\beta$ (see \Cref{lem:pinball_loss_is_lip} for details). Thus, we may apply Theorem 22 of \citet{Bousquet2002} to conclude that
    \[
    \|\hat{g}_{\beta} - \hat{g}^{-A}_{\beta}\|_{K} \leq \frac{\kappa |A|}{2\lambda (n+1)}.
    \]
    Then by applying the reproducing property of the RKHS and our bound on the kernel we arrive at the desired inequality, 
    \begin{align*}
    \|\hat{g}_{\beta} - \hat{g}^{-A}_{\beta}\|_{\infty} & = \sup_x |\langle \hat{g}_{\beta} - \hat{g}^{-A}_{\beta}, K(x,\cdot) \rangle | \leq \sup_x \|\hat{g}_{\beta} - \hat{g}^{-A}_{\beta}\|_{K} \|K(x,\cdot)\|_K\\
    & = \|\hat{g}_{\beta} - \hat{g}^{-A}_{\beta}\|_{K} \sup_x K(x,x)^{1/2} \leq \frac{\kappa^2|A|}{2\lambda (n+1)}.
    \end{align*}
\end{proof}

Our second lemma bounds the stability of the fit in $\beta$.

\begin{lemma}\label{lem:stab_in_beta}
     Assume that $\sup_{x,x} K(x,x) = \kappa^2 < \infty$. Then for any dataset $\{(X_i,S_i)\}_{i=1}^{n+1}$,
    \[
     \|\hat{g}_{\beta_1} - \hat{g}_{\beta_2}\|_{\infty} \leq \sqrt{\frac{4\kappa^2}{\lambda} \frac{1}{n+1}  \sum_{i=1}^{n+1} |\Phi(X_i)^{\top}(\beta_1 - \beta_2) |},\quad  \forall \beta_1,\beta_2 \in \mmr^d.
    \]
\end{lemma}

\begin{proof}
    The proof of this lemma is quite similar to the proof of Theorem 22 in \citet{Bousquet2002}. For ease of notation let 
    \[
    L_{\beta}(g_K) := \frac{1}{n+1} \sum_{i=1}^{n+1} \ell_{\alpha}(g_K(X_i) + \Phi(X_i)^{\top}\beta, S_i).
    \]
    By the optimality of $\hat{g}_{\beta_1}$ and $\hat{g}_{\beta_2}$ we have
    \begin{align*}
        & L_{\beta_1}(\hat{g}_{\beta_1}) + \lambda \|\hat{g}_{\beta_1}\|_K^2 \leq L_{\beta_1}\left(\frac{1}{2}\hat{g}_{\beta_1} + \frac{1}{2}\hat{g}_{\beta_2}\right) + \lambda \left\|\frac{1}{2}\hat{g}_{\beta_1} + \frac{1}{2}\hat{g}_{\beta_2}\right\|_K^2,\\
        \text{and } & L_{\beta_2}(\hat{g}_{\beta_2}) + \lambda \|\hat{g}_{\beta_2}\|_K^2 \leq L_{\beta_2}\left(\frac{1}{2}\hat{g}_{\beta_1} + \frac{1}{2}\hat{g}_{\beta_2}\right) + \lambda \left\|\frac{1}{2}\hat{g}_{\beta_1} + \frac{1}{2}\hat{g}_{\beta_2}\right\|_K^2.
    \end{align*}
    Moreover, by the convexity of $L_{\beta_1}(\cdot)$ and $L_{\beta_2}(\cdot)$ it also holds that
    \begin{align*}
    & L_{\beta_1}\left( \frac{1}{2} \hat{g}_{\beta_1} + \frac{1}{2}\hat{g}_{\beta_2} \right) \leq \frac{1}{2} L_{\beta_1}(\hat{g}_{\beta_1}) + \frac{1}{2} L_{\beta_1}(\hat{g}_{\beta_2}),\\
    \text{and } & L_{\beta_2}\left( \frac{1}{2} \hat{g}_{\beta_1} + \frac{1}{2}\hat{g}_{\beta_2} \right) \leq \frac{1}{2} L_{\beta_2}(\hat{g}_{\beta_1}) + \frac{1}{2} L_{\beta_2}(\hat{g}_{\beta_2}).
    \end{align*}
    Putting all four of these inequalities together we find that 
    \begin{align*}
        \frac{\lambda}{2} \left\|\hat{g}_{\beta_1} - \hat{g}_{\beta_2}\right\|^2_K & = \lambda \|\hat{g}_{\beta_1}\|_K^2 + \lambda \|\hat{g}_{\beta_2}\|_K^2 - 2\lambda \left\|\frac{1}{2}\hat{g}_{\beta_1} + \frac{1}{2}\hat{g}_{\beta_2}\right\|_K^2\\
        & \leq L_{\beta_1}\left(\frac{1}{2}\hat{g}_{\beta_1} + \frac{1}{2}\hat{g}_{\beta_2}\right)  + L_{\beta_2}\left(\frac{1}{2}\hat{g}_{\beta_1} + \frac{1}{2}\hat{g}_{\beta_2}\right) -  L_{\beta_1}(\hat{g}_{\beta_1}) - L_{\beta_2}(\hat{g}_{\beta_2})\\
        & \leq  \frac{1}{2} L_{\beta_1}(\hat{g}_{\beta_1}) + \frac{1}{2} L_{\beta_1}(\hat{g}_{\beta_2}) + \frac{1}{2} L_{\beta_2}(\hat{g}_{\beta_1}) + \frac{1}{2} L_{\beta_2}(\hat{g}_{\beta_2})  -  L_{\beta_1}(\hat{g}_{\beta_1}) - L_{\beta_2}(\hat{g}_{\beta_2})\\
        & =  \frac{1}{2} \left( L_{\beta_2}(\hat{g}_{\beta_1}) - L_{\beta_1}(\hat{g}_{\beta_1}) + L_{\beta_1}(\hat{g}_{\beta_2})  - L_{\beta_2}(\hat{g}_{\beta_2}) \right)\\
        & \leq \frac{1}{n+1} \sum_{i=1}^{n+1} |\Phi(X_i)^{\top}(\beta_1 - \beta_2) |,
    \end{align*}
    where the last inequality follows from the Lipschitz property of $\ell_{\alpha}(\cdot,\cdot)$ (see \Cref{lem:pinball_loss_is_lip}). To conclude the proof one simply notes that by the reproducing property of the RKHS we have that 
    \[
    \|\hat{g}_{\beta_1} - \hat{g}_{\beta_2}\|_{\infty} \leq \kappa \|\hat{g}_{\beta_1} - \hat{g}_{\beta_2}\|_K \leq \sqrt{\frac{4\kappa^2}{\lambda} \frac{1}{n+1}  \sum_{i=1}^{n+1} |\Phi(X_i)^{\top}(\beta_1 - \beta_2) |},
    \]
    as desired.    
\end{proof}

In order to apply this lemma to bound we will need to control the size of $|\Phi(X_i)^{\top}(\beta_1 - \beta_2) |$. This is done in our next result. The statement of this lemma may look somewhat peculiar due to the presence of a re-weighting function $f \in \mathcal{F}$. To help aid intuition it may be useful to keep in mind the special case $f = 1$, which turns the expectation below into a simple tail probability. While somewhat strange, our reason for stating the lemma in this form is that it will fit seamlessly into our later calculations without the need for additional exposition. 

\begin{lemma}\label{lem:bound_on_top_phi_mean}
    Assume that $X_1,\dots,X_{n+1} \stackrel{i.i.d}{\sim} P_X$. Let $f \in \mathcal{X} \to \mmr$ and assume that there exists constants $C_2, C_3\geq 1$ such that, $\mme[\|\Phi(X_i)\|_2^2] \leq C_2 d$ and $\mme[|f(X_i)| \cdot \|\Phi(X_i)\|_2^2] \leq C_3\mme[|f(X)|]d$. Then for any $\epsilon > 0$ and $1 \leq j \leq n+1$,
    \[
    \mme\left[|f(X_j)| \bone\left\{ \sup_{\beta : \|\beta\| \leq \epsilon} \frac{1}{n+1}  \sum_{i=1}^{n+1} |\Phi(X_i)^{\top}\beta | > 2\epsilon\sqrt{C_2 d}  \right\} \right] \leq  O\left(\frac{ \mme[|f(X)|]}{n} \right).
    \]
\end{lemma}
\begin{proof}
    By the Cauchy-Schwartz inequality we have 
    \[
    \sup_{\beta : \|\beta\| \leq \epsilon} \frac{1}{n+1}  \sum_{i=1}^{n+1} |\Phi(X_i)^{\top}\beta |  \leq \sup_{\beta : \|\beta\| \leq \epsilon} \frac{1}{n+1}  \sum_{i=1}^{n+1} \|\Phi(X_i)\|_2  \|\beta\|_2 =  \frac{\epsilon}{n+1}  \sum_{i=1}^{n+1} \|\Phi(X_i)\|_2.
    \]
    Additionally, by Jensen's inequality it also holds that $\mme[ \|\Phi(X_i)\|_2] \leq \sqrt{ \mme[ \|\Phi(X_i)\|^2_2]} = \sqrt{C_2 d}$. So, putting these two inequalities together we arrive at 
    \begin{align*}
        & \mme\left[|f(X_j)| \bone\left\{ \sup_{\beta : \|\beta\| \leq \epsilon} \frac{1}{n+1}  \sum_{i=1}^{n+1} |\Phi(X_i)^{\top}\beta | > 2\epsilon\sqrt{C_2 d}  \right\} \right]\\
        & \leq \mme\left[|f(X_j)| \bone\left\{ \frac{1}{n+1}  \sum_{i=1}^{n+1} \|\Phi(X_i) \|_2 - \mme[\|\Phi(X_i) \|_2 ] > \sqrt{C_2 d}  \right\} \right] \\
        & \leq \frac{1}{C_2 d} \mme\left[|f(X_j)| \left(\frac{1}{n+1}  \sum_{i=1}^{n+1} \|\Phi(X_i) \|_2 - \mme[\|\Phi(X_i) \|_2] \right)^2 \right] = O\left( \frac{\mme[|f(X)|]}{n}\right).
    \end{align*}
\end{proof}

The final preliminary lemmas that we will require are controls on the maximum possible sizes of $\hat{g}_{S_{n+1},K}$ and $\hat{\beta}_{S_{n+1}}$. Once again these lemmas will involve re-weighting functions the purpose of which is to ease our calculations further on.

\begin{lemma}\label{lem:f_is_bounded}
 It holds deterministically that
 \[
\|\hat{g}_{S_{n+1},K}\|_K \leq \frac{1}{\sqrt{\lambda}} \sqrt{\frac{1}{n+1}\sum_{i=1}^{n+1} |S_i|}.
 \]
If in addition, $(X_1,S_1),\dots,(X_{n+1},S_{n+1}) \stackrel{i.i.d.}{\sim} P$, $\mme[S_i^2] < \infty$, and $f : \mathcal{X} \to \mmr$ is a function satisfying $\mme[|f(X_i)| S_i^2] \leq C_{f,S}\mme[|f(X_i)|]$ for some $C_{f,S} > 0$, then we also have that for all $1 \leq j \leq n+1$,
 \[
\mme\left[ |f(X_i)| \bone\left\{\|\hat{g}_{S_{n+1},K}\|_K \geq \frac{\sqrt{2\mme[|S_i|]}}{\sqrt{\lambda}} \right\} \right] \leq O\left( \frac{\mme[|f(X_i)|]}{n}\right)
 \]
\end{lemma}
\begin{proof}
    Taking $\beta, f_K = 0$ gives loss
    \[
    \frac{1}{n+1} \sum_{i=1}^{n+1} \ell_{\alpha}(0,S_i) + \lambda \|0\|_K^2 \leq \frac{1}{n+1} \sum_{i=1}^{n+1} |S_i|.
    \]
    So, since $(\hat{g}_{S_{n+1},K},\hat{\beta}_{S_{n+1}})$ is a minimizer of the quantile regression objective we must have that
    \[
    \lambda \|\hat{g}_{S_{n+1},K}\|^2_K \leq \frac{1}{n+1} \sum_{i=1}^{n+1} \ell_{\alpha}(\hat{g}_{S_{n+1},K}(X_i) + \Phi(X_i)^\top \hat{\beta}_{S_{n+1}},S_i) + \lambda \|\hat{g}_{S_{n+1},K}\|^2_K \leq \frac{1}{n+1} \sum_{i=1}^{n+1} |S_i|.
    \]
    This proves the first part of the lemma. To get the second part we simply note that
    \begin{align*}
        \mme\left[ |f(X_j)| \bone\left\{\|\hat{g}_{S_{n+1},K}\|_K \geq \frac{\sqrt{2\mme[|S_i|]}}{\sqrt{\lambda}} \right\} \right] & \leq \mme\left[ |f(X_j)| \bone\left\{\frac{1}{n+1}\sum_{i=1}^{n+1} |S_i| - \mme[|S_i|] \geq \mme[|S_i|]\right\}\right]\\
    & \leq \frac{1}{\mme[|S_i|]^2}\mme\left[|f(X_j)|\left( \frac{1}{n+1}\sum_{j=1}^{n+1} |S_i| - \mme[|S_i|] \right)^2\right]\\
    & = O\left(\frac{\mme[|f(X)|]}{n} \right).
    \end{align*}
\end{proof}

\begin{lemma}\label{lem:beta_is_bounded}
    Let $(X_1,S_1),\dots,(X_{n+1},S_{n+1}) \stackrel{i.i.d.}{\sim} P$ and $f : \mathcal{X} \to \mmr$. Assume that $\mme[S_i^2] < \infty$, $\sup_{x}K(x,x) = \kappa^2 < \infty$, and there exists constants $C_2, c_2, C_f,C_{f,S}, \rho > 0$ such that $\mme[|f(X_i)| S_i^2] \leq C_{f,S}\mme[|f(X_i)|]$, $\sqrt{\mme[|f(X_i)|^2]} \leq C_{f}\mme[|f(X_i)|]$, $\inf_{\beta} \mme[|\Phi(X_i)^\top\beta|] \geq \rho$, $\sup_{\beta : \|\beta\|_2 = 1} \mme[|\Phi(X_i)^\top \beta|^2]^{1/2} \leq c_2$, and $\mme[\|\Phi(X_i)\|^2] \leq C_2 d$. Then there exists a constant $c_{\beta} > 0$ such that for all $1 \leq j \leq n+1$,
    \[
    \mme\left[|f(X_j)|\bone\left\{\|\hat{\beta}_{S_{n+1}} \|_2 >\frac{1}{\sqrt{\lambda}}c_{\beta} \right\} \right] \leq  O\left(\frac{d\mme[|f(X_i)|]}{n}\right).
    \]
\end{lemma}
\begin{proof}   
     Observe that
    \begin{align*}
    & \frac{1}{n+1} \sum_{i=1}^{n+1} \ell_{\alpha}(\hat{g}_{S_{n+1},K}(X_i) + \Phi(X_i)^\top \hat{\beta}_{S_{n+1}},S_i) +  \lambda \|\hat{g}_{S_{n+1},K}\|^2_K\\
    & \geq \frac{1}{n+1}\sum_{i=1}^{n+1}  \min\{\alpha, 1-\alpha\} |S_i - \hat{g}_{S_{n+1},K}(X_i) - \Phi(X_i)^\top \hat{\beta}_{S_{n+1}}| \\
    & \geq  \min\{\alpha, 1-\alpha\} \left(\frac{1}{n+1}\sum_{i=1}^{n+1} |\Phi(X_i)^\top \hat{\beta}_{S_{n+1}}| -  \frac{1}{n+1}\sum_{i=1}^{n+1}|S_i| - \frac{1}{n+1}\sum_{i=1}^{n+1} |\hat{g}_{S_{n+1},K}(X_i)|\right).
    \end{align*}
    Moreover, by the reproducing property of the RKHS
    \[
     \frac{1}{n+1}\sum_{i=1}^{n+1} |\hat{g}_{S_{n+1},K}(X_i)| \leq \|\hat{g}_{S_{n+1},K}\|_{\infty}  \leq \kappa \|\hat{g}_{S_{n+1},K}\|_{K},
    \]
    and by \Cref{lem:f_is_bounded}
    \[
    \|\hat{g}_{S_{n+1},K}\|_{K} \leq \frac{1}{\sqrt{\lambda}} \sqrt{\frac{1}{n+1} \sum_{i=1}^{n+1} |S_i|}.
    \]
    So, combining these two facts we find that
    \[
    \frac{1}{n+1}\sum_{i=1}^{n+1} |\Phi(X_i)^\top \hat{\beta}_{S_{n+1}}| \leq \frac{1}{\min\{\alpha, 1-\alpha\}} \left(\frac{\kappa}{\sqrt{\lambda}} \sqrt{\frac{1}{n+1} \sum_{i=1}^{n+1} |S_i|}  + \frac{1}{n+1} \sum_{i=1}^{n+1} |S_i|\right) .
    \]
    To use this inequality to bound $\|\hat{\beta}_{S_{n+1}}\|_2$ we will need to lower bound the mean of $|\Phi(X_i)^\top \hat{\beta}_{S_{n+1}}|$. This is done in \Cref{lem:lower_concentration_of_phi} below where we show that there exists constants $c, c' > 0$, such that
    \[
    \mmp\left( \inf_{\beta : \|\beta\| = 1} \frac{1}{n+1}\sum_{i=1}^{n+1} |\Phi(X_i)^\top \beta| \geq c\right) \geq 1-c'\frac{d^2}{(n+1)^2}.
    \]
    Thus,
    \begin{align*}
         & \mme\left[|f(X_j)| \bone\left\{\|\hat{\beta}_{S_{n+1}}\|_2 >\frac{1}{\sqrt{\lambda}}c_{\beta} \right\} \right] \leq \mme\left[|f(X_j)| \bone\left\{ \frac{1}{n+1} \sum_{i=1}^{n+1} |S_i| - \mme[|S_i|] >\Omega(1) \right\} \right]\\
        & \ \ \ \ + \mme\left[|f(X_j)| \bone\left\{\inf_{\beta : \|\beta\| = 1} \frac{1}{n+1}\sum_{i=1}^{n+1} |\Phi(X_i)^\top \beta| \geq c \right\} \right]\\
        & \leq  O(1)\mme\left[|f(X_j)| \left( \frac{1}{n+1} \sum_{i=1}^{n+1} |S_i| - \mme|S_i| \right)^2 \right] + \mme[|f(X_i)|^2]^{1/2} \sqrt{c'}\frac{d}{n+1} \leq O\left( \frac{d\mme[|f(X_i)|]}{n+1}\right).
    \end{align*}
\end{proof}

With all of these results in hand we are ready to prove \Cref{prop:rkhs_bounds}.

\begin{proof}[Proof of \Cref{prop:rkhs_bounds}.]
    We will exploit the stability of the RKHS fit. Let $\epsilon  = O\left(\frac{\lambda}{n^2\sqrt{d}}\right)$ and define the event 
    \[
    E := \left\{ \sup_{\beta : \|\beta\|_2 \leq \epsilon} \frac{1}{n+1} \sum_{i=1}^{n+1} |\Phi(X_i)^\top \beta| \leq 2\epsilon\sqrt{C_2 d},\ \|\hat{\beta}_{S_{n+1}}\|_2 \leq \frac{c_{\beta}}{\sqrt{\lambda}} \right\}.
    \]
    By Lemmas \ref{lem:bound_on_top_phi_mean} and \ref{lem:beta_is_bounded} we know that 
    \begin{align*}
    & \mme[|f(X_i)| \bone\{S_i =  \hat{g}_{S_{n+1},K}(X_i) + \Phi(X_i)^\top\hat{\beta}_{S_{n+1}}\} ]\\
    & \leq \mme[|f(X_i)| \bone\{S_i =  \hat{g}_{S_{n+1},K}(X_i) + \Phi(X_i)^\top\hat{\beta}_{S_{n+1}}\}\bone\{E\} ] + O\left(\frac{d\mme[|f(X_i)|]}{n} \right).
    \end{align*}
    Thus, we just need to focus on what happens on the event $E$. By applying the exchangeability of the quadruples $(\hat{g}_{S_{n+1},K}(X_i),\hat{\beta}_{S_{n+1}}, X_i, S_i) $ we have that
    \begin{align*}
    & \mme[|f(X_i)| \bone\{S_i =  \hat{g}_{S_{n+1},K}(X_i) + \Phi(X_i)^\top\hat{\beta}_{S_{n+1}}\}\bone\{E\} ]\\
    & = \mme\left[ \left(\frac{1}{n+1} \sum_{i=1}^{n+1} |f(X_i)| \bone\{S_i =  \hat{g}_{S_{n+1},K}(X_i) + \Phi(X_i)^\top\hat{\beta}_{S_{n+1}}\}\right) \bone\{E\}  \right]\\
    & \leq \mme\left[ \max_{1\leq i \leq n+1}|f(X_i)| \mme\left[ \left(\frac{1}{n+1} \sum_{i=1}^{n+1}  \bone\{S_i =  \hat{g}_{\hat{\beta}_{S_{n+1}}}(X_i) + \Phi(X_i)^\top\hat{\beta}_{S_{n+1}}\}\right)\bone\{E\} \mid (X_i)_{i=1}^{n+1} \right] \right].
    \end{align*}
    To bound this quantity we just need to control the inner expectation. We will begin by fixing a large integer $m > 1$ and applying the inequality
    \begin{align*}
    & \mme\left[ \left(\frac{1}{n+1} \sum_{i=1}^{n+1}  \bone\{S_i =  \hat{g}_{\hat{\beta}_{S_{n+1}}}(X_i) + \Phi(X_i)^\top\hat{\beta}_{S_{n+1}}\}\right)\bone\{E\} \mid (X_i)_{i=1}^{n+1} \right]\\
    & \leq \mme\left[ \left(\frac{1}{n+1} \sum_{i=1}^{n+1}  \bone\{S_i =  \hat{g}_{\hat{\beta}_{S_{n+1}}}(X_i) + \Phi(X_i)^\top\hat{\beta}_{S_{n+1}}\}\right)^m\bone\{E\} \mid (X_i)_{i=1}^{n+1} \right]^{1/m}.
    \end{align*}
    Our motivation for applying this bound is that by choosing $m$ sufficiently large we will be able to swap a sum and a maximum without losing too much slack. More precisely, let $\mathcal{N} \subseteq \mmr^d$ be a minimal size $\epsilon$-net of $\{\beta \in \mmr^d : \|\beta\|_2 \leq c_{\beta}/\sqrt{\lambda}\}$. It is well known that there exists an absolute constant $C_N > 0$ such that $|\mathcal{N}| \leq \exp(C_Nd\log(\frac{c_{\beta}}{\sqrt{\lambda}\epsilon}))$. Then, using this $\epsilon$-net we compute that
    \begin{align*}
   & \mme\left[ \left(  \frac{1}{n+1} \sum_{i=1}^{n+1}  \bone\{S_i =  \hat{g}_{\hat{\beta}_{S_{n+1}}}(X_i) + \Phi(X_i)^\top\hat{\beta}_{S_{n+1}}\}  \right)^m \bone\{E\} \mid (X_i)_{i=1}^{n+1} \right]\\
   & \leq  \mme\left[ \sup_{\beta : \|\beta\|_2 \leq c_{\beta}/\sqrt{\lambda}} \left(  \frac{1}{n+1} \sum_{i=1}^{n+1}  \bone\{S_i =  \hat{g}_{\beta}(X_i) + \Phi(X_i)^\top\beta\}  \right)^m \bone\{E\} | (X_i)_{i=1}^{n+1} \right]\\
    & \leq  \mme\left[  \sup_{\beta \in \mathcal{N}} \left(  \frac{1}{n+1} \sum_{i=1}^{n+1}  \bone\{|S_i -  \hat{g}_{\beta}(X_i) - \Phi(X_i)^\top\beta| \leq O\left(1/n\right) \}  \right)^m \bone\{E\} \mid (X_i)_{i=1}^{n+1} \right]\\
    & \leq \sum_{\beta \in \mathcal{N}} \mme\left[\left(  \frac{1}{n+1} \sum_{i=1}^{n+1}  \bone\{|S_i -  \hat{g}_{\beta}(X_i) - \Phi(X_i)^\top\beta| \leq O\left(1/n\right) \}  \right)^m \mid (X_i)_{i=1}^{n+1} \right],
    \end{align*}
    where the first inequality follows from the definition of $E$ and the second inequality uses both \Cref{lem:stab_in_beta} and the fact that on the event $E$
    \begin{align*}
    &\sup_{\beta : \|\beta\| \leq \epsilon} \frac{1}{n+1}  \sum_{i=1}^{n+1} |\Phi(X_i)^{\top}\beta | \leq 2\epsilon\sqrt{C_2 d} \\
    \implies & \sup_{\beta : \|\beta\| \leq \epsilon} \max_{1 \leq i \leq n+1} |\Phi(X_i)^{\top}\beta| \leq \sup_{\beta : \|\beta\| \leq \epsilon} \frac{n+1}{n+1} \sum_{i=1}^{n+1} |\Phi(X_i)^{\top}\beta |  \leq (n+1)2\epsilon \sqrt{C_2 d} \leq O\left(\frac{1}{n} \right).
   \end{align*}
   Continuing this calculation directly we see that,
   \begin{align*}
        & \mme\left[\left(  \frac{1}{n+1} \sum_{i=1}^{n+1}  \bone\{|S_i -  \hat{g}_{\beta}(X_i) - \Phi(X_i)^\top\beta| \leq O\left(1/n\right) \}  \right)^m \mid (X_i)_{i=1}^{n+1} \right]\\
        & = \sum_{k=1}^m \binom{n+1}{k} \binom{m}{k} k! k^{m-k} \frac{1}{(n+1)^m} \mme\left[\prod_{i=1}^k \bone\{|S_i -  \hat{g}_{\beta}(X_i) - \Phi(X_i)^\top\beta| \leq O\left(1/n\right) \} \mid (X_i)_{i=1}^{n+1}\right]\\
        & \leq \sum_{k=1}^m (n+1)^k \binom{m}{k} \frac{m^{m-k}}{(n+1)^m}\mme\left[\prod_{i=1}^k \bone\left\{|S_i -  \hat{g}^{-\{1,\dots,k\}}_{\beta}(X_i) - \Phi(X_i)^\top\beta| \leq O\left(\frac{k}{\lambda n}\right) \right\} \mid (X_i)_{i=1}^{n+1} \right],
    \end{align*}
    where the last line applies \Cref{lem:stab_in_data}. Finally, using the fact that $S_i|X_i$ has a bounded density we may upper bound the above display by
    \begin{align*}
        &  \sum_{k=1}^m (n+1)^k  \binom{m}{k}\frac{m^{m-k}}{(n+1)^m} O\left( \left(\frac{k}{\lambda n} \right)^k \right) \leq O\left( \left(\frac{m}{\lambda n}\right)^{m} \right) \sum_{k=1}^m \binom{m}{k} \leq O\left( 2^m \left(\frac{m}{\lambda n}\right)^{m} \right).
   \end{align*}
    Putting this all together we conclude that 
   \[
   \mme\left[ \left(  \frac{1}{n+1} \sum_{i=1}^{n+1}  \bone\{S_i =  \hat{g}_{\hat{\beta}_{S_{n+1}}}(X_i) + \Phi(X_i)^\top\hat{\beta}_{S_{n+1}}\}  \right)^m \mid (X_i)_{i=1}^{n+1} \right]^{\frac{1}{m}} \leq O\left(\exp\left(\frac{C_Nd\log\left(\frac{1}{\sqrt{\lambda}\epsilon}\right)}{m}\right) \frac{m}{\lambda n}\right).
   \]
   The desired result then follows by taking $m = d\log(\frac{1}{\sqrt{\lambda}\epsilon})$ and plugging in our definition for $\epsilon$.

\end{proof}

\subsubsection{Proof of Proposition \ref{prop:rkhs_inner_prod_est}}

To simplify the notation let
\begin{align*}
& L_{n}(\beta,g_K) :=  \frac{1}{n+1} \sum_{i = 1}^{n} \ell_{\alpha}(\Phi(X_i)^\top\beta + g_K(X_i),S_i)\\
\text{and } & L_{\infty}(\beta,g_K) :=  \mme[\ell_{\alpha}(\Phi(X_i)^\top\beta + g_K(X_i),S_i)],
\end{align*}
denote the empirical and population losses and let
\begin{align*}
& M_{n}(\beta,g_K) :=  L_{n}(\beta,g_K) + \lambda \|g_K\|_K^2\\
\text{and } & M_{\infty}(\beta,g_K) :=  L_{\infty}(\beta,g_K) + \lambda \|g_K\|_K^2,
\end{align*}
denote the corresponding empirical and population objectives. Note that $M_{n}$ and $M_{\infty}$ are strictly convex in $f$ and convex in $\beta$. Thus, we may let $(\hat{B}_{n},\hat{g}_{n,K}),(B^*, g_{K}^*) \in 2^{\mmr^d} \times \mathcal{F}_K$ denote the minimizers of $ M_{n}$ and $M_{\infty}$ respectively. To further ease notation in the sections that follows we will sometimes use $\hat{\beta}_{n}$ and $\beta^*$ to denote arbitrarily elements of $\hat{B}_{n}$ and $B^*$. Finally, we will let $\Pi_{\hat{B}_n}, \Pi_{B^*} : \mmr^{d} \to \mmr^d$ denote the projections operators onto $\hat{B}_n$ and $B^*$, respectively.

With these preliminary definitions in hand we now formally state the assumptions of \Cref{prop:rkhs_inner_prod_est}. Our first assumption is that $M_{\infty}$ is locally strongly convex around its minimum.
\begin{assumption}[Population Strong Convexity]\label{ass:pop_strong_convex}
   Let $d(\beta,g_K) := \inf_{\beta' \in B^*} \|\beta - \beta'\|_2 + \|g_K - g^*_{K}\|_K$ denote the distance from $(g_K,\beta)$ to the nearest population minimizer. Then, there exists constants $C_M, \delta_M > 0$ such that
   \[
   d(\beta,g_K)   \leq \delta_M \implies M_{\infty}(\beta,g_K) - M_{\infty}(\beta^*_,g^*) \geq C_Md(\beta,g_K)^2.
   \]
\end{assumption}
 Overall, we believe that this assumption is mild and should hold for all distributions of interest. For instance, for continuous data it is easy to check that this condition holds whenever $S \mid X$ has a positive density on $\mmr$. On the other hand, for discrete data we expect to have the even stronger inequality $M_{\infty}(\beta,g_K) - M_{\infty}(\beta^*_,g^*) \geq C_Md(\beta,g_K)$. This is due to the fact that for discrete data $L_{\infty}(\cdot,\cdot)$ has sharp jump discontinuities that give rise to large increases in the loss when $(\hat{\beta},\hat{g}_{n,K})$ moves away from $(B^*,g^*)$.
 
The second assumption we will need is a set of moment conditions on $S$ and $X$. 
\begin{assumption}[Moment Conditions]\label{ass:moment_cond_for_quant_conv}
    There exists constants $C_2,\rho > 0$ such that
    \begin{align*}
    & \mme[\|\Phi(X_i)\|^2_2] \leq C_2 d \ \ \text{ and } \ \ \inf_{\beta : \|\beta\|_2 = 1} \mme[|\Phi(X_i)^\top\beta|] \geq \rho.
    \end{align*}
    Furthermore, we also have that $\mme[|S_i|^2] < \infty$.
\end{assumption}

With these assumptions in hand we are ready to prove \Cref{prop:rkhs_inner_prod_est}. We begin by giving a technical lemma that controls the concentration of $L_n$ around $L_{\infty}$.

\begin{lemma}\label{lem:localized_rad_of_loss}
    Assume that $(X_1,S_1),\dots,(X_n,S_n) \stackrel{i.i.d.}{\sim} P$ and that there exists constants $C_2, \kappa > 0$ such that $\mme[\|\Phi(X_i)\|^2_2] \leq C_2 d$ and $\sup_x K(x,x) = \kappa^2 < \infty$. Then for any $\delta_1, \delta_2 > 0$,
    \begin{align*}
    & \mme\left[\sup_{\|\beta - \Pi_{B^*}\beta\|_2 \leq \delta_1,\ \|g_K - g^*_{K}\|_K \leq \delta_2} \left| L_n(\beta,g_K) - L_n(\Pi_{B^*}\beta,g^*_K) - (L_{\infty}(\beta,g_K) - L_{\infty}(\Pi_{B^*}\beta,g^*_K)) \right|\right]\\
    & \hspace{1cm} \leq O\left(\delta_1\sqrt{\frac{d}{n}} + \delta_2\sqrt{\frac{1}{n}}\right)
    \end{align*}
\end{lemma}
\begin{proof}
    Let $E := \{(\beta,g_K) \in \mmr^d \times \mathcal{F}_K : \|\beta - \Pi_{B^*}\beta\|_2 \leq \delta_1,\ \|g_K - g^*_{K}\|_K \leq \delta_2 \}$ and $\sigma_1,\dots,\sigma_{n+1} \stackrel{i.i.d}{\sim} \text{Unif}(\{\pm1\})$ be Rademacher random variables. Since the pinball loss is 1-Lipschitz (see \Cref{lem:pinball_loss_is_lip}) we may apply the symmetrization and contraction properties of Rademacher complexity to conclude that
    \begin{align*}
    & \mme\left[\sup_{(\beta,g_K) \in E} \left| L_n(\beta,g_K) - L_n(\Pi_{B^*}\beta,g^*_K) - (L_{\infty}(\beta,g_K) - L_{\infty}(\Pi_{B^*}\beta,g^*_K))  \right|\right]\\
    & \leq 2\mme\left[\sup_{(\beta,g_K) \in E} \left| \frac{1}{n} \sum_{i=1}^{n} \sigma_i(\ell(\Phi(X_i)^\top\beta + g_K(X_i),S_i) - \ell(\Phi(X_i)^\top\Pi_{B^*}\beta + g^*_K(X_i),S_i))  \right|\right]\\
    & \leq 2\mme\left[\sup_{(\beta,g_K) \in E} \frac{1}{n} \sum_{i=1}^{n} \sigma_i (\Phi(X_i)^\top (\beta - \Pi_{\beta^*}\beta) + g_K(X_i) - g^*_K(X_i))  \right]\\
    & \leq 2\mme\left[\sup_{\|\beta - \Pi_{B^*}\beta\|_2 \leq \delta_1}  \frac{1}{n} \sum_{i=1}^{n} \sigma_i \Phi(X_i)^\top (\beta - \Pi_{\beta^*}\beta) \right]\\
    & \ \ \ \ + 2\mme\left[ \sup_{\|g_K - g^*_K\|_K \leq \delta_2}\frac{1}{n}\sum_{i=1}^{n}\sigma_i(g_K(X_i) - g^*_K(X_i)) \right]\\
    & \leq O\left(\delta_1 \sqrt{\frac{d}{n}} + \delta_2 \sqrt{\frac{1}{n}} \right),
    \end{align*}
    where the last inequality follows from standard bounds on the Rademacher complexities of linear and kernel function classes (see e.g. Secion 4.1.2 of (\cite{Boucheron2005}))
\end{proof}

We now prove the main proposition.

\begin{proof}[Proof of \Cref{prop:rkhs_inner_prod_est}]
We will show that 
\begin{enumerate}
    \item $\sup_{f \in \mathcal{F}_{\delta}}|\frac{1}{n}\sum_{i=1}^n |f(X_i)| - \mme_P[|f(X)|]| = O_{\mmp}(\sqrt{d/n})$,
    \item $\sup_{f_K \in \mathcal{F}_K:\|f_K\|_K \leq 1} \lambda|\mme[\langle \hat{g}_{S_{n+1},K}, f_K \rangle_K]| \leq O(1)$,
    \item   $\sup_{f_K \in \mathcal{F}_K:\|f_K\|_K \leq 1} \lambda|\langle \hat{g}_{n,K}, f_K \rangle_K - \mme[\langle \hat{g}_{S_{n+1},K}, f_K \rangle_K]| \leq O_{\mmp}(\sqrt{d\log(n)/n}) $.
\end{enumerate}
Our desired result will then follow by writing
\begin{align*}
    & \sup_{f(\cdot) = \Phi(\cdot)^\top\beta + f_K(\cdot) \in \mathcal{F}_{\delta}} \left|2\lambda\frac{\langle \hat{g}_{n,K}, f_K \rangle}{\frac{1}{n}\sum_{i=1}^nf(X_i)} - 2\lambda\frac{\mme[\langle \hat{g}_{S_{n+1},K}, f_K \rangle] }{\mme_P[|f(X)|]}\right|\\
    & \leq \sup_{f_K \in \mathcal{F}_K : \|f_K\|_K \leq 1}2\lambda \frac{|\langle \hat{g}_{n,K}, f_K \rangle - \mme[\langle \hat{g}_{S_{n+1},K}, f_K \rangle]| }{\mme_P[|f(X)|]}\\
    & \ \ \ \ \ + 2\lambda \sup_{f(\cdot) = \Phi(\cdot)^\top\beta + f_K(\cdot) \in \mathcal{F}_{\delta}} \left| \frac{\langle \hat{g}_{n,K}, f_K \rangle }{\frac{1}{n}\sum_{i=1}^n |f(X_i)|} - \frac{\langle \hat{g}_{n,K}, f_K \rangle }{\mme_P[|f(X_i)|]} \right|\\
    & \leq O_{\mmp}\left(\sqrt{\frac{d\log(n)}{n}}\right) + \frac{\sup_{f \in \mathcal{F}_K}\mme[\langle \hat{g}_{S_{n+1},K}, f_K ]\rangle| + O_{\mmp}(\sqrt{\frac{d\log(n)}{n}})}{\delta^2 - O_{\mmp}(\sqrt{d/n})} \sup_{f \in \mathcal{F}_{\delta}}\left|\frac{1}{n}\sum_{i=1}^n |f(X_i)| - \mme_P[|f(X)|]\right|\\
    & =  O_{\mmp}\left(\sqrt{\frac{d\log(n)}{n}}\right).
\end{align*}
We establish each of these three facts in order.\\

\noindent \textbf{Step 1:} By the results of Section 4.1.2 in \citet{Boucheron2005} we know that $\{f_K(\cdot) + \Phi(\cdot)^\top \beta  : \|f\|_K + \|\beta\|_2 \leq 1\}$ has Rademacher complexity at most $O(\sqrt{d/n})$. By the contraction property this also implies that $\{|f_K(\cdot) + \Phi(\cdot)^\top \beta|  : \|f\|_K + \|\beta\|_2 \leq 1\}$ has Rademacher complexity at most $O(\sqrt{d/n})$. So, by the symmetrization inequality we have that for any $C>0$,
\begin{align*}
& \mmp\left(\sup_{f \in \mathcal{F}_{\delta}} \left|\frac{1}{n}\sum_{i=1}^n |f(X_i)| - \mme_P[|f(X)|]\right| > C \right) \leq \frac{1}{C} \mme\left[ \sup_{f \in \mathcal{F}_{\delta}} \left|\frac{1}{n}\sum_{i=1}^n |f(X_i)| - \mme_P[|f(X)|]\right| \right]\\
& \hspace{3cm}  \leq \frac{2}{C} \text{RadComplex}_n(\{|f_K(\cdot) + \Phi(\cdot)^\top \beta|  : \|f\|_K + \|\beta\|_2 \leq 1\}) \leq \frac{1}{C}O\left(\sqrt{\frac{d}{n}}\right).
\end{align*}
This proves that $\sup_{f \in \mathcal{F}_{\delta}} \left|\frac{1}{n}\sum_{i=1}^n |f(X_i)| - \mme_P[|f(X)|]\right|  = O_{\mmp}(\sqrt{d/n})$, as desired.
\\

\noindent \textbf{Step 2:} By \Cref{lem:f_is_bounded} we know that
\begin{align*}
\sup_{f \in \mathcal{F}_K : \|f_K\|_K \leq 1} |\mme[\langle \hat{g}_{S_{n+1},K}, f_K \rangle_K]| \leq \mme[\|\hat{g}_{S_{n+1},K}\|_K ] \leq \mme\left[\frac{1}{\sqrt{\lambda}}  \sqrt{\frac{1}{n+1}\sum_{i=1}^{n+1}|S_i|}\right]  \leq \sqrt{\frac{\mme[|S_i|]}{\lambda}}.
\end{align*}
Multiplying both sides by $\lambda$ gives the desired result.\\

\noindent \textbf{Step 3:} This step is considerably more involved than the previous two. To begin write
\begin{align*}
\sup_{f_K:\|f_K\|_K \leq 1} \langle \hat{g}_{n,K}, f_K \rangle - \mme[\langle \hat{g}_{S_{n+1},K}, f_K \rangle] & = \sup_{f_K:\|f_K\|_K \leq 1} \langle \hat{g}_{n,K} - g^*_K, f_K \rangle  +  \mme[\langle g^*_K-\hat{g}_{S_{n+1},K}, f_K \rangle] \\
& \leq \|\hat{g}_{n,K} - g^*_K\|_K + \mme[\|g^*_K-\hat{g}_{S_{n+1},K}\|_K]. 
\end{align*}
We will bound each of the two terms on the right hand side separately. To do this we will use a two-step peeling argument where each step gives a tighter bound on $\|\hat{g}_{n,K} - g^*_K\|_K$ than the previous one. 

Our first step will show that with high probability $(\hat{\beta}_n,\hat{g}_{n,K})$ must be within $\delta_M$ of $(B^*,g^*_{K})$. Let $c_{\beta}$ be the constant appearing in \Cref{lem:beta_is_bounded}. Then, by a direct computation we have that
\begin{align*}
    & \mmp(d(\hat{\beta}_n,\hat{g}_{n,K}) > \delta_M) = \mmp(M_n(\hat{\beta}_n,\hat{g}_{n,K}) - M_n(\Pi_{B^*}\hat{\beta}_n,g^*_K) \leq 0,\ d(\hat{\beta}_n,\hat{g}_{n,K}) > \delta_M )\\
    & \leq \mmp(M_n(\hat{\beta}_n,\hat{g}_{n,K}) - M_n(\Pi_{B^*}\hat{\beta}_n,g^*_K) - (M_{\infty}(\hat{\beta}_n,\hat{g}_{n,K}) - M_{\infty}(\Pi_{B^*}\hat{\beta}_n,g^*_K)) \leq -\delta_M^2)\\
    & \leq \mmp\left(\sup_{\|\beta\|_2 \leq \frac{c_{\beta}}{\sqrt{\lambda}}, \|g_K\| \leq \sqrt{\frac{2\mme[|S_i|]}{\lambda}}} |M_n(\beta,g_{K}) - M_n(\Pi_{B^*}\beta,g^*_K) - (M_{\infty}(\beta,g_{K}) - M_{\infty}(\Pi_{B^*}\beta,g^*_K))| \geq \delta_M^2\right) \\
    & \quad +\mmp\left(\|\hat{\beta}_n\|_2 \geq \frac{c_{\beta}}{\sqrt{\lambda}}\right) +\mmp\left(\|\hat{g}_{n,K}\|_K \geq  \sqrt{\frac{2\mme[|S_i|]}{\lambda}}\right)\\
    & \leq \frac{1}{\delta_M^2} \mme\left[\sup_{\|\beta\|_2 \leq \frac{c_{\beta}}{\sqrt{\lambda}}, \|g_K\| \leq \sqrt{\frac{2\mme[|S_i|]}{\lambda}}} |L_n(\beta,g_{K}) - L_n(\Pi_{B^*}]\beta,g^*_K) - (L_{\infty}(\beta,g_{K}) - L_{\infty}(\Pi_{B^*}\beta,g^*_K))| \right] +  O\left(\frac{d}{n} \right),
\end{align*}
where the last line follows by applying Lemmas \ref{lem:f_is_bounded} and \ref{lem:beta_is_bounded} with $f(\cdot) = 1$. Finally, by \Cref{lem:localized_rad_of_loss} we can additionally bound the first term above as 
\[
 \mme\left[\sup_{\|\beta\|_2 \leq \frac{c_{\beta}}{\sqrt{\lambda}}, \|g_K\| \leq \sqrt{\frac{2\mme[|S_i|]}{\lambda}}} |L_n(\beta,g_{K}) - L_n(\Pi_{B^*}\hat{\beta}_n,g^*_K) - (L_{\infty}(\beta,g_{K}) - L_{\infty}(\Pi_{B^*}\hat{\beta}_n,g^*_K))| \right] \leq  O\left(\sqrt{\frac{d}{\lambda n}} \right),
\]
So, in total we find that
\[
\mmp(d(\hat{\beta}_n,\hat{g}_{n,K}) > \delta_M) \leq O\left(\sqrt{\frac{d}{\lambda n}} \right).
\]
This concludes the proof of our first concentration inequality for $(\hat{\beta}_n,\hat{g}_{n,K})$.

In our second step we will use this preliminary bound to get an even tighter control on $d(\hat{\beta}_n,\hat{g}_{n,K})$. Fix any $C>0$ with $C \sqrt{d/(\lambda n)} < \delta_M$. For any $j \in \mmr$ let $A_j := \{(\beta, g_K) : 2^{j-1} < \sqrt{\frac{\lambda n}{d}}d(\beta,g_{n,K}) \leq 2^j\}$. Then,

\begin{align*}
    & \mmp\left(d(\hat{\beta}_n,\hat{g}_{n,K}) > C \sqrt{d/(\lambda n)}\right)\\
    & \leq \sum_{j : \frac{1}{2}C \leq 2^j \leq 2\sqrt{\frac{d}{n\lambda}}\delta_M} \mmp\left(\sup_{(\beta,g_K) \in A_j} M_n( \beta,g_K) - M_{\infty}(\beta,g_K) \leq 0\right) + \mmp(d(\hat{\beta}_n,\hat{g}_{n,K}) 
 > \delta_M)\\
    & \leq \sum_{j : \frac{1}{2}C \leq 2^j \leq 2\sqrt{\frac{d}{n\lambda}}\delta_M} \mmp\left(\sup_{(\beta,g_K) \in A_j} |M_n( \beta,g_K) - M_{n}(\Pi_{B^*}\beta,g^*_K) - (M_{\infty}(\beta,g_K) - M_{\infty}(\Pi_{B^*}\beta,g^*_K))| \geq  \frac{2^{2j-2}d}{n\lambda}\right)\\
    & \ \ \ \ + O\left( \sqrt{\frac{d}{\lambda n}} \right)\\
    & \leq \sum_{j : \frac{1}{2}C \leq 2^j \leq 2\sqrt{\frac{d}{n\lambda}}\delta_M}\frac{n\lambda}{2^{2j-2}d}  \mme\left[\sup_{(\beta,g_K) \in A_j} \left| (L_n(\beta,g_K) - L_n(\Pi_{B^*}\beta,g^*_K) - (L_{\infty}(\beta,g_K) - L_{\infty}(\Pi_{B^*}\beta,g^*_K)) \right| \right]\\
    & \ \ \ \ + O\left( \sqrt{\frac{d}{\lambda n}} \right)\\
    & \leq \sum_{j : \frac{1}{2}C \leq 2^j \leq 2\sqrt{\frac{d}{n\lambda}}\delta_M} O(\sqrt{\lambda} 2^{-j}) + O\left( \sqrt{\frac{d}{\lambda n}} \right)  \leq O\left(\frac{\sqrt{\lambda}}{C}\right) + O\left( \sqrt{\frac{d}{\lambda n}} \right) .
\end{align*}
This proves that $\|\hat{g}_{n,K} - g^*_K\|_K = O_{\mmp}(\sqrt{\frac{d}{\lambda n}})$. To get a similar bound on the expectation write
\begin{align*}
& \mme[\|\hat{g}_{S_{n+1},K} - g^*_K\|_K]  \leq \int_{0}^{\delta_M}\mmp(\|\hat{g}_{S_{n+1},K} - g^*_K\| > t) dt\\
& \hspace{5cm} + \mme[(\|\hat{g}_{n,K} \|_K + \|g^*_K\|_K)\bone\{\|\hat{g}_{S_{n+1},K} - g^*_K\|_K > \delta_M\}]\\
& \leq \int_{0}^{\delta_M}\min\left\{1,O\left(\frac{1}{t}\sqrt{\frac{d}{\lambda n}}\right) + O\left(\sqrt{\frac{d}{\lambda n}} \right) \right\} dt\\
& \hspace{5cm} + \mme\left[\left(\sqrt{\frac{1}{\lambda(n+1)}\sum_{i=1}^{n+1} |S_i|}+\sqrt{ \frac{\mme[|S_i|]}{\lambda}}\right)\bone\{\|\hat{g}_{S_{n+1},K} - g^*_K\|_K > \delta_M\}\right]\\
& \leq O\left( \sqrt{\frac{d\log(n)}{\lambda n}} \right)  + \mme\left[\left(\sqrt{\frac{1}{\lambda(n+1)}\sum_{i=1}^{n+1} (|S_i| - \mme[|S_i|])}+2\sqrt{ \frac{\mme[|S_i|]}{\lambda}}\right)\bone\{\|\hat{g}_{S_{n+1},K} - g^*_K\|_K > \delta_M\}\right]\\
& \leq O\left( \sqrt{\frac{d\log(n)}{\lambda n}} \right)+ \mme\left[\left|\frac{1}{\lambda(n+1)}\sum_{i=1}^{n+1} (S_i - \mme[S_i])\right|\right]^{1/2}\mmp(\|\hat{g}_{S_{n+1},K} - g^*_K\|_K > \delta_M)^{1/2}\\
& \hspace{5cm} + 2\sqrt{\frac{\mme[|S_i|]}{\lambda}}\mmp(\|\hat{g}_{S_{n+1},K} - g^*_K\|_K > \delta_M)\\
& = O\left( \frac{1}{\lambda}\sqrt{\frac{d\log(n)}{ n}} \right),
\end{align*}
as desired.

\end{proof}
\subsection{Proofs for Lipschitz Functions}

In this section we prove \Cref{prop:lip_bounds}. Throughout we make the following set of technical assumptions
\begin{assumption}\label{ass:lip_tech_conditions}
    There exists constants $C_X,C_{\Phi},C_S,C_{f}, \rho > 0$ such that $\sup_{f \in \mathcal{F}}\sqrt{\mme[|f(X_i)|^2]} \leq C_{f}\mme[|f(X_i)|]$, $\inf_{\beta : \|\beta\|_2 = 1}\mme[|\Phi(X_i)^\top\beta|] \geq \rho$, and with probability 1, $\|X_i\|^2_2 \leq C_X p$, $\|\Phi(X_i)\|^2_2 \leq C_{\Phi} d$, and $|S_i| \leq C_S$ for all $i$.
\end{assumption}

The primary technical tool that we will need for the proof is a covering number bound for Lipschitz functions. This result is well-known from prior literature and we re-state it here for clarity. In the work that follow we use $B_{p}(0,C) := \{x \in \mmr^p : \|x\|_2 \leq C\}$ to denote the ball of radius $C$ in $\mmr^p$.

\begin{definition}
The covering number $\mathcal{N}(\mathcal{F},\epsilon,\|\cdot\|)$ of a set $\mathcal{F}$ under norm $\|\cdot\|$ is the minimum number of balls $B(f,\epsilon) = \{g \in \mathcal{F} : \|f - g\| \leq \epsilon\}$ of radius $\epsilon$ needed to cover $\mathcal{F}$.
\end{definition}

\begin{lemma}[Covering number of Lipschitz functions, Theorem 2.7.1 in \citet{VDVWellner1996}]\label{lem:lip_covering}
    Let $\mathcal{F}^{\textup{lip}}_{L,B_1,B_2,p} := \{f : B_p(0,B_1) \to \mmr \mid \textup{Lip}(f) \leq L,\ \|f\|_{\infty} \leq B_2\}$ denote the space of bounded Lipschitz functions on $B_p(0,B_1)$. Then there exists a constant $C>0$ such that for any $\epsilon > 0$,
    \[
    \log(\mathcal{N}(\mathcal{F}^{\textup{lip}}_{L,B_1,B_2,p},\epsilon,\|\cdot\|_{\infty})) \leq \left(\frac{CB_1\max\{L,B_2\}}{\epsilon}\right)^p.
    \]
\end{lemma}

In the present context we need to control the behaviour of Lipschitz fitting under a re-weighting $f$. To account for this we will require a more general covering number bound on a weighted class of Lipschitz functions. This is given in the following lemma. In the work that follows, recall that for a probability measure $Q$, the $L_2(Q)$ norm of a function $f$ is defined as $\|f\|_{L_2(Q)} := \mme_{X \sim Q}[f(X)^2]^{1/2}$.

\begin{lemma}\label{lem:re-weighted_lip_covering}
    Let $f  : B_p(0,B_1) \to \mmr$ be a fixed function and $\mathcal{F}^{\textup{wlip}}_{L,B_1,B_2,p,f} := \{fg \mid g : B_p(0,B_1) \to \mmr,\  \textup{Lip}(g) \leq L,\ \|g\|_{\infty} \leq B_2\}$ denote the space of bounded Lipschitz functions on $B_p(0,B_1)$ multiplied with $f$. Then for any probablity measure $Q$, there exists a constant $C>0$ such that for any $\epsilon > 0$,
    \[
    \log(\mathcal{N}(\mathcal{F}^{\textup{wlip}}_{L,B_1,B_2,p,f},\epsilon,\|\cdot\|_{L_2(Q)})) \leq \left(\frac{CB_1\max\{L,B_2\}\|f\|_{L_2(Q)}}{\epsilon}\right)^p.
    \]
\end{lemma}
\begin{proof}
    Recall that we defined $\mathcal{F}^{\textup{lip}}_{L,B_1,B_2,p} := \{g : B_p(0,B_1) \to \mmr \mid \textup{Lip}(g) \leq L,\ \|g\|_{\infty} \leq B_2\}$. Fix any $\epsilon > 0$ and let $A \subseteq \mathcal{F}^{\textup{lip}}_{L,B_1,B_2,p}$ be a minimal $\|\cdot\|_{\infty}$-norm, $\epsilon/\|f\|_{L_2(Q)}$-covering of $\mathcal{F}^{\textup{lip}}_{L,B_1,B_2,p}$. Fix any $h \in \mathcal{F}^{\textup{lip}}_{L,B_1,B_2,p}$ and let $ h' \in A$ be such that $\|h - h'\|_{\infty} \leq \epsilon/\|f\|_{L_2(Q)}$. Then,
    \[
    \|fh - fh'\|_{L_2(Q)} = \mme_{X \sim Q}[|f(X)|^2\cdot|h(X) - h(X')|^2]^{1/2}  \leq \mme_{X \sim Q}\left[|f(X)|^2 \left(\frac{\epsilon}{\|f\|_{L_2(Q)}}\right)^2\right]^{1/2} = \epsilon.
    \]
    In particular, we find that $\{fh : h \in A\}$ is an $\|\cdot\|_{L_2(Q)}$-norm, $\epsilon$-covering of $\mathcal{F}^{\textup{wlip}}_{L,B_1,B_2,p,f}$. The desired result then immediately follows by applying \Cref{lem:lip_covering} to get a bound on $|A|$.
\end{proof}

The previous two lemmas only apply to bounded Lipschitz functions with maximum Lipshitz norm $L$. To apply these results to our current setting we will need to bound $\text{Lip}(\hat{g}_{S_{n+1},L})$ and $\|\hat{g}_{S_{n+1},L}\|_{\infty}$. Our next lemma does exactly this.

\begin{lemma}\label{lem:lip_part_is_bounded}
    Assume that there exist constants $C_X, C_S > 0$ such that with probability one $\|X_i\|^2_2 \leq pC_X $ and $|S_i| \leq C_S$ for all $i$. Then with probability one, $ \text{Lip}(\hat{g}_{S_{n+1},L}) \leq \frac{C_S}{\lambda}$ and $\|\hat{g}_{S_{n+1},L}\|_{\infty} \leq \frac{\sqrt{C_X p}C_S}{\lambda}$ .
\end{lemma}
\begin{proof}
    Since $(\hat{g}_{S_{n+1},L},\hat{\beta}_{S_{n+1}})$ is a minimizer of the quantile regression objective we must have
    \begin{align*}
    \lambda\text{Lip}(\hat{g}_{S_{n+1},L}) & \leq \frac{1}{n+1}\sum_{i=1}^{n+1} \ell_{\alpha}\left (\hat{g}_{S_{n+1},L}(X_i) + \Phi(X_i)^\top\hat{\beta}_{S_{n+1}},S_i \right ) + \lambda\text{Lip}(\hat{g}_{S_{n+1},L})\\
    & \leq \frac{1}{n+1}\sum_{i=1}^{n+1} \ell_{\alpha}(0,S_i) + \lambda\text{Lip}(0) \leq \frac{1}{n+1} \sum_{i=1}^{n+1} |S_i| \leq C_S.
    \end{align*}
    This proves the first part of the proposition. To get the second part, note that since $\Phi(\cdot)$ has an intercept term we may assume without loss of generality that $\hat{g}_{S_{n+1},L}(0) = 0$. Thus,
    \[
    \|\hat{g}_{S_{n+1},L}\|_{\infty} \leq \sqrt{C_X p}\sup_{x \in B_p(0,\sqrt{C_Xp})} \frac{\|\hat{g}_{S_{n+1},L}(x) - \hat{g}_{S_{n+1},L}(0)\|}{\|x\|_2} \leq  \sqrt{C_X p} \text{Lip}(\hat{g}_{S_{n+1},L}) \leq \sqrt{C_X p} \frac{C_S}{\lambda},
    \]
    as desired.
\end{proof}

Our final preliminary lemma gives a control on the norm of the linear part of the fit. Similar to what we had for RKHS functions above, here we state this result under an arbitrary re-weighting $f$. 

\begin{lemma}\label{lem:beta_is_bounded_for_lip}
    Let $f : \mathcal{X} \to \mmr$ and $(X_1,S_1),\dots,(X_{n+1},S_{n+1}) \stackrel{i.i.d.}{\sim} P$. Assume that there exists constants $C_X,C_{\Phi},C_S,C_{f}, \rho > 0$ such that $\sqrt{\mme[|f(X_i)|^2]} \leq C_{f}\mme[|f(X_i)|]$, $\inf_{\beta : \|\beta\|_2 = 1}\mme[|\Phi(X_i)^\top\beta|] \geq \rho$, and with probability 1, $\|X\|_2^2 \leq C_Xp$, $\|\Phi(X_i)\|^2_2 \leq C_{\Phi} d$, and $|S_i| \leq C_S$ for all $i$. Then there exists a constant $c_{\beta} > 0$ such that
    \[
    \mme\left[|f(X_i)|\bone\left\{\|\hat{\beta}_{S_{n+1}} \|_2>c_{\beta}\frac{\sqrt{p}}{\lambda} \right\} \right] \leq  O\left(\frac{d\mme[|f(X_i)|]}{n}\right).
    \]
\end{lemma}

\begin{proof}
    By \Cref{lem:lip_part_is_bounded} we know that without loss of generality $\|\hat{g}_{S_{n+1},L}\|_{\infty} \leq \frac{\sqrt{C_X p}C_S}{\lambda}$. Moreover, by assumption we have that deterministically $\frac{1}{n+1} \sum_{i=1}^{n+1} |S_i| \leq C_S$. With these preliminary facts in hand the desired result follows by repeating the proof of \Cref{lem:beta_is_bounded}.
\end{proof}

With these preliminaries out of the way we are now ready to prove \Cref{prop:lip_bounds}.

\begin{proof}[Proof of \Cref{prop:lip_bounds}]
The main idea of this proof is to show that $\frac{1}{n+1} \sum_{i=1}^{n+1} \bone\{S_i = g_L(X_i) + \Phi(X_i)^\top\beta\}$ concentrates uniformly around its expectation. Since for any fixed $(g_L,\beta)$, $\mme[\bone\{S = g_L(X) + \Phi(X)^\top\beta\}] = 0$ this will imply that
\[
\frac{1}{n+1} \sum_{i=1}^{n+1} \bone\{S_i = \hat{g}_{S_{n+1},L}(X_i) + \Phi(X_i)^\top\hat{\beta}_{S_{n+1}}\} \cong \mme[\bone\{S = \hat{g}_{S_{n+1}}(X) + \Phi(X)^\top\hat{\beta}_{S_{n+1}}\}  ] = 0.
\]
We now formalize this idea. Define the event
\[
    E := \left\{\|\hat{g}_{S_{n+1},L}\|_{\infty} \leq \frac{\sqrt{C_Xp}C_S}{\lambda},\ \text{Lip}(f_L) \leq \frac{C_S}{\lambda},\ \|\hat{\beta}_{S_{n+1}}\|_2 \leq c_{\beta}\frac{\sqrt{p}}{\lambda} \right\}.
    \]
    By Lemmas \ref{lem:lip_part_is_bounded} and  \ref{lem:beta_is_bounded_for_lip} we know that 
    \begin{align*}
    & \mme[|f(X_i)| \bone\{S_i =  \hat{g}_{S_{n+1},L}(X_i) + \Phi(X_i)^\top\hat{\beta}_{S_{n+1}}\} ]\\
    & \leq \mme[|f(X_i)| \bone\{S_i =  \hat{g}_{S_{n+1},L}(X_i) + \Phi(X_i)^\top\hat{\beta}_{S_{n+1}}\}\bone\{E\} ] + O\left(\frac{d\mme[|f(X_i)|]}{n} \right).
    \end{align*}
    Thus, we just need to focus on what happens on the event $E$. By the exchangeability of the quadruples $(\hat{g}_{S_{n+1},L}(X_i),\hat{\beta}_{S_{n+1}}, X_i, S_i) $ we have
    \begin{align*}
    & \mme[|f(X_i)| \bone\{S_i =  \hat{g}_{S_{n+1},L}(X_i) + \Phi(X_i)^\top\hat{\beta}_{S_{n+1}}\}\bone\{E\} ]\\
    & = \mme\left[ \left(\frac{1}{n+1} \sum_{i=1}^{n+1} |f(X_i)| \bone\{S_i =  \hat{g}_{S_{n+1},L}(X_i) + \Phi(X_i)^\top\hat{\beta}_{S_{n+1}}\}\right) \bone\{E\}  \right].
    \end{align*}
    Let $\delta > 0$ be a small constant that we will specify later and $h$ denote the tent function
\[
h(x) = \begin{cases}
0,\  |x| > \delta\\
1-\delta^{-1}|x|,\  |x| \leq \delta.
\end{cases}
\]
Let $\mathcal{G} := \{g : g(\cdot) = g_L(\cdot) + \Phi(\cdot)^\top\beta,\ \|\beta\|_2 \leq \frac{c_{\beta}\sqrt{p}}{\lambda},\ \|g_L\|_{\infty} \leq \frac{\sqrt{C_Xp}C_S}{\lambda},\ \text{Lip}(g_L) \leq \frac{C_S}{\lambda}\}$ and $\sigma_1,\dots,\sigma_n \stackrel{i.i.d.}{\sim} \text{Unif}(\{\pm1\})$. Then, 
\begin{align*}
    & \mme\left[\frac{1}{n+1} \sum_{i=1}^{n+1} |f(X_i)| \bone\{S_i = \hat{g}_{S_{n+1}}(X_i) + \Phi(X_i)^\top\hat{\beta}_{S_{n+1}}\} \bone\{E\}  \right] \\
    & \leq \mme\left[ \frac{1}{n+1} \sum_{i=1}^{n+1} |f(X_i)| h(S_i - \hat{g}_{S_{n+1}}(X_i) - \Phi(X_i)^\top\hat{\beta}_{S_{n+1}}) \bone\{E\}  \right]\\
    & \leq \mme\left[\sup_{g \in \mathcal{G}} \frac{1}{n+1}\sum_{i=1}^{n+1} |f(X_i)| h(S_i - g(X_i)) - \mme[|f(X_1)| h(S_1 - g(X_1)) ] \right]\\
    & \quad + \sup_{g \in \mathcal{G}} \mme[|f(X_1)| h(S_1 - g(X_1)) ]\\
    & \leq 2\mme\left[\sup_{g \in \mathcal{G}} \frac{1}{n+1}\sum_{i=1}^{n+1} \sigma_i |f(X_i)| h(S_i - g(X_i)) \right] + \sup_{g \in \mathcal{G}} \mme[|f(X_1)|\mmp(|S_1 - g(X_1)| \mid X_1\leq \delta) ]\\
    & \leq 2\delta^{-1} \mme\left[\sup_{g \in \mathcal{G}} \frac{1}{n+1} \sum_{i=1}^{n+1} \sigma_i |f(X_i)| (S_i - g(X_i))  \right] + O(\delta\mme[|f(X_i)| ])\\
    & \leq  2\delta^{-1} \mme\left[\left|\frac{1}{n+1} \sum_{i=1}^{n+1} \sigma_i |f(X_i)| S_i \right|  \right]+ 2\delta^{-1} \mme\left[\sup_{\beta : \|\beta\|_2 \leq \frac{c_{\beta}\sqrt{p}}{\lambda}} \left| \frac{1}{n+1} \sum_{i=1}^{n+1} \sigma_i |f(X_i)| \Phi(X_i)^\top\beta \right|   \right]\\
    & \quad +  2\delta^{-1} \mme\left[\sup_{g_L : \|g_L\|_{\infty} \leq \frac{\sqrt{C_Xp}C_S}{\lambda},\ \text{Lip}(g_L) \leq \frac{C_S}{\lambda}} \left| \frac{1}{n+1} \sum_{i=1}^{n+1} \sigma_i |f(X_i)| g_L(X_i) \right|   \right] + O(\delta\mme[|f(X_i)|]),
\end{align*}
where the third inequality follows by symmetrization, and the fourth inequality uses the fact that $S_i|X_i$ has a bounded density, the contraction inequality, and the fact that $h(\cdot)$ is $\delta^{-1}$-Lipschitz. 

To conclude the proof we  bound each of the first three terms appearing on the last line above. We have that
\begin{align*}
\mme\left[\left|\frac{1}{n+1} \sum_{i=1}^{n+1} \sigma_i |f(X_i)| S_i \right| \right] & \leq \sqrt{\text{Var}\left(\frac{1}{n+1} \sum_{i=1}^{n+1} \sigma_i |f(X_i)| S_i  \right)}\\
& \leq \sqrt{\frac{C_S^2 \mme[f(X_i)^2]^{1/2}}{n+1}} = O\left( \sqrt{\frac{\mme[|f(X_i)|]}{n}} \right), \text{ (by Assumption \ref{ass:lip_tech_conditions})},
\end{align*}
while
\begin{align*}
  \mme\left[  \sup_{\beta : \|\beta\|_2 \leq \frac{c_{\beta}\sqrt{p}}{\lambda}} \left| \frac{1}{n+1} \sum_{i=1}^{n+1} \sigma_i |f(X_i)| \Phi(X_i)^\top\beta \right|  \right] & \leq \frac{c_{\beta}\sqrt{p}}{\lambda} \mme\left[   \left\|\frac{1}{n+1} \sum_{i=1}^{n+1} \sigma_i |f(X_i)| \Phi(X_i) \right\|_2 \right] \\ 
  & \leq \frac{c_{\beta}\sqrt{p}}{\lambda}   \mme\left[\left\|\frac{1}{n+1} \sum_{i=1}^{n+1} \sigma_i |f(X_i)| \Phi(X_i) \right\|_2^2  \right]^{1/2}\\
  & = O\left(\sqrt{\frac{dp\mme[|f(X_i)|^2]}{\lambda^2 n}} \right)\\
  & = O\left(\sqrt{\frac{dp}{\lambda^2 n}} \mme[|f(X_i)|] \right), \text{ (by Assumption \ref{ass:lip_tech_conditions})}.
\end{align*}
Finally, by \Cref{lem:re-weighted_lip_covering} we have the covering number bound
\begin{align*}
& \log\left(\mathcal{N}\left(\left\{|f|g_L : B_p(0,\sqrt{C_Xp}) \to \mmr \mid \|g_L\|_{\infty} \leq \frac{\sqrt{C_Xp}C_S}{\lambda},\ \text{Lip}(g_L) \leq \frac{C_S}{\lambda}\right\},\epsilon, \|\cdot\|_{L_2(P_X)} \right)\right)\\
& \leq O\left(\frac{p \mme[f(X_i)^2]^{1/2}}{\lambda \epsilon} \right)^{p} = O\left(\frac{p \mme[|f(X_i)|]}{\lambda \epsilon} \right)^{p},
\end{align*}
and so by Dudley's entropy integral
\begin{align*}
    \mme\left[\sup_{g_L : \|g_L\|_{\infty} \leq \frac{\sqrt{C_Xp}C_S}{\lambda},\ \text{Lip}(g_L) \leq \frac{C_S}{\lambda}} \left| \frac{1}{n+1} \sum_{i=1}^{n+1} \sigma_i |f(X_i)|g_L(X_i) \right|  \right] & \leq O\left(\frac{p\mme[|f(X_i)|]\log(n)}{\lambda n^{\min\{1/2,1/p\}}}  \right).
\end{align*}
Putting all of these results together gives the final bound
\begin{align*}
 & \mme\left[\bone\{S = \hat{g}_{S_{n+1}}(X) + \Phi(X)^\top\hat{\beta}_{S_{n+1}}\}\bone\{E\} \right]\\
 & \leq \delta^{-1}  \left(O\left(\frac{p\mme[|f(X_i)|]\log(n)}{\lambda n^{\min\{1/2,1/p\}}}  \right)+ O\left(\sqrt{\frac{dp}{\lambda^2 n}} \mme[|f(X_i)|] \right) \right)   + O(\delta \mme[|f(X_i)|]).
\end{align*}
The desired result then follows by optimizing over $\delta$.
    
\end{proof}

\subsection{Proofs for Section~\ref{sec:computation}}

In this section we prove the results appearing in \Cref{sec:computation} of the main text. We start with a proof of \Cref{thm:pred_set_is_mon}.

\begin{proof}[Proof of \Cref{thm:pred_set_is_mon}] 
We begin by giving a more careful derivation of the dual program. Recall that our primal optimization problem is 
\begin{equation*}
\begin{aligned}
   \text{minimize}_{g \in \mathcal{F}} \quad & (1 - \alpha)\cdot \mathbf{1}^\top p + \alpha \cdot \mathbf{1}^\top q +  (n + 1) \cdot \mathcal{R}(g) \\
    \text{s.t.}  \quad &S_i -  g(X_i)  - p_i + q_i = 0 \\
     &S -  g(X_{n+1})   - p_{n + 1} + q_{n + 1} = 0 \\
     & p_{i}, q_i \geq 0,\ 1 \leq i \leq n+1.
     \end{aligned}
\end{equation*}
The Lagrangian for this program is
\begin{equation}\label{eq:app_lagrangian}
\begin{split}
    & (1-\alpha) \cdot \mathbf{1}^\top p + \alpha \cdot \mathbf{1}^\top q  +  (n + 1) \cdot \mathcal{R}(g) + \sum_{i = 1}^n \eta_i \left(S_i -  g(X_i)  - p_i + q_i \right) \\
    & \hspace{2cm} + \eta_{n + 1} \left( S -  g(X_{n+1}) - p_{n + 1} + q_{n + 1} \right) 
     - \sum_{i = 1}^{n + 1} (\gamma_i p_i  + \xi_i q_i).
\end{split}
\end{equation}
For ease of notation, let $\mathcal{R}^*(\eta) := -\min_{g \in \mathcal{F}}  (n+1)\mathcal{R}(g) - \sum_{i=1}^{n+1}\eta_i g(X_i) $. Then, minimizing with respect to $g$ gives,
\begin{align*}
    (1 - \alpha)\cdot \mathbf{1}^\top p + \alpha \cdot \mathbf{1}^\top q  -  \mathcal{R}^*\left(\eta\right) + \sum_{i = 1}^n \eta_i S_i +\eta_{n + 1} S  - \eta^\top p + \eta^\top q - \gamma^\top p - \xi^\top q.
\end{align*}
So, taking derivatives of this function with respect to $p$, and $q$, we arrive at the constraints,
\begin{align*}
    \gamma &= (1 - \alpha) \cdot \mathbf{1} - \eta \\
    \xi &= \alpha \cdot \mathbf{1} + \eta.
\end{align*}
Since the only restriction on $\xi$ and $\gamma$ is that they are non-negative, this can be simplified to,
\begin{align*}
    \eta &\leq (1 - \alpha) \cdot \mathbf{1} \\
    \eta &\geq -\alpha \cdot \mathbf{1}.
\end{align*}
Thus, we arrive at the desired dual formulation,
\begin{equation}\label{eq:app_dual_form}
\begin{split}
    & \text{maximize}_{\eta} \quad  \sum_{i = 1}^n \eta_i S_i + \eta_{n+1} S - \mathcal{R}^*\left( \eta \right)   \\
    & \text{subject to } \quad -\alpha \leq \eta_i \leq 1 - \alpha,\ 1 \leq i\leq n+1.
\end{split}
\end{equation}
    
Now, recall that we used the notation $\eta^{S}$ to denote the dual-optimal $\eta$ for a particular choice of $S$. Assume for the sake of contradiction that there exists $\tilde{S} > S$ such that $\eta^{\tilde{S}}_{n + 1} < \eta^{S}_{n + 1}$. Observe that we can write the dual objective as
\begin{align*}
    h(\eta^S) + S \cdot \eta^S_{n + 1},
\end{align*}
where $h$ does not depend on $S$. 
Our assumption implies that
\begin{align*}
    (\tilde{S} - S) \cdot \left( \eta^{\tilde{S}}_{n + 1} - \eta^S_{n + 1} \right) < 0,
\end{align*}
or equivalently,
\begin{align*}
    \tilde{S} \cdot \left( \eta^{\tilde{S}}_{n + 1} - \eta^S_{n + 1} \right) < S \cdot \left( \eta^{\tilde{S}}_{n + 1} - \eta^S_{n + 1} \right).
\end{align*}
On the other hand, by the optimality of $\eta^S$, we have that
\begin{align*}
    h(\eta^{\tilde{S}}) + \tilde{S} \cdot \eta^{\tilde{S}}_{n + 1} \geq h(\eta^{S}) + \tilde{S} \cdot \eta^{S}_{n + 1} \quad \iff \quad \tilde{S} \cdot \left(\eta^{\tilde{S}}_{n + 1} - \eta^{S}_{n + 1} \right) \geq h(\eta^{S}) - h(\eta^{\tilde{S}}).
\end{align*}
Applying our assumption, we conclude that
\begin{align*}
    S \cdot \left( \eta^{\tilde{S}}_{n + 1} - \eta^S_{n + 1} \right) > h(\eta^{S}) - h(\eta^{\tilde{S}}),
\end{align*}
which by rearranging yields the desired contradiction
\begin{align*}
    h(\eta^{\tilde{S}}) + S \cdot \eta^{\tilde{S}}_{n + 1} > h(\eta^{S}) + S \cdot \eta^S_{n + 1} .
\end{align*}
\end{proof}

We now turn to the proof of \Cref{prop:cov_of_dual}, which states that the coverage properties of $\hat{C}_{\text{dual}}(X_{n+1})$ are the same as $\hat{C}(X_{n+1})$. 

\begin{proof}[Proof of  \Cref{prop:cov_of_dual}]
    The proof of this Proposition is nearly identical to the proof of \Cref{thm:infinite_dim_result}, with the exception that now instead of looking at the first order conditions of the primal, we will instead investigate the first order conditions of the Lagrangian \eqref{eq:app_lagrangian} at the optimal dual variables. We keep all the notation the same as in the proof of \Cref{thm:pred_set_is_mon}. 
    
    We begin by proving the first statement pertaining to the coverage properties of $\hat{C}_{\text{dual}}(X_{n+1})$. Let $(\hat{g}_{S_{n+1}},p^{S_{n+1}},q^{S_{n+1}},\eta^{S_{n+1}},\gamma^{S_{n+1}}_i,\xi_i^{S_{n+1}})$ denote an optimal primal-dual solution at the input $S = S_{n+1}$. Recall from the proof of \Cref{thm:pred_set_is_mon} that the Lagrangian for the optimization is
    \begin{align*}
        & (1 - \alpha)\cdot \mathbf{1}^\top p^{S_{n+1}} + \alpha \cdot \mathbf{1}^\top q^{S_{n+1}} + (n+1)\mathcal{R}(\hat{g}_{S_{n+1}}) + \sum_{i=1}^{n+1} \eta_i^{S_{n+1}}(S_i - \hat{g}_{S_{n+1}}(X_i) - p^{S_{n+1}}_i + q^{S_{n+1}}_i)\\
        & \quad \quad - \sum_{i=1}^{n+1} (\gamma^{S_{n+1}} p^{S_{n+1}}_i + \xi_i^{S_{n+1}} q_i^{S_{n+1}}).
    \end{align*}
    Fix any re-weighting function $f \in \mathcal{F}$. By assumption we know that strong duality (and thus the KKT conditions) hold. So, by considering the derivative of the Lagrangian in the direction $f$ and applying the KKT stationarity condition we find that 
    \begin{equation}\label{eq:lang_dir_deriv}
        0 = \frac{d}{d\epsilon} (n+1)\mathcal{R}(\hat{g}_{S_{n+1}} + \epsilon f) \bigg|_{\epsilon = 0} - \sum_{i=1}^{n+1} \eta^{S_{n+1}}_i f(X_i). 
    \end{equation}
    To further unpack this equality, note that complementary slackness in the KKT condition necessitates that $p^{{S_{n+1}}}_i \gamma^{S_{n+1}}_i = 0$ and  $q^{S_{n+1}}_i \xi^{S_{n+1}}_i = 0$. Thus, when $S_i -  \hat{g}_{S_{n+1}}(X_i) > 0$, we must have $\gamma^{S_{n+1}}_i = 0$, or equivalently, $\eta^{S_{n+1}}_i = 1 - \alpha$ and when $S_i -  \hat{g}_{S_{n+1}}(X_i) < 0$, we must have $\eta^{S_{n+1}}_i = -\alpha$. Last, when the residual is exactly $0$, the corresponding $\eta^{S_{n+1}}_i$ can take any value in $\left[-\alpha,  1 - \alpha \right]$. Plugging these observations into \eqref{eq:lang_dir_deriv}, we obtain
    \begin{align*}
    0 & = \frac{d}{d\epsilon} (n + 1) \cdot \mathcal{R}(\hat{g}_{S_{n+1}} + \epsilon f) \bigg|_{\epsilon = 0} + \sum_{i : S_i < \hat{g}_{S_{n+1}}(X_i)} \alpha \cdot f(X_i)\\
    & \quad \quad - \sum_{i : S_i >\hat{g}_{S_{n+1}}(X_i)} (1 - \alpha) f(X_i) -\sum_{i : S_i = \hat{g}_{S_{n+1}}(X_i)} \eta_i^{S_{n+1}} f(X_i)\\
    & = \frac{d}{d\epsilon} (n + 1) \cdot \mathcal{R}(\hat{g}_{S_{n+1}} + \epsilon f) \bigg|_{\epsilon = 0} + \sum_{i : \eta^{S_{n+1}}_i < 1 - \alpha} \alpha f(X_i)\\
    & \quad \quad - \sum_{i : \eta^{S_{n+1}}_i = 1 - \alpha} (1 - \alpha) f(X_i) -\sum_{i : S_i = \hat{g}_{S_{n+1}}(X_i),\ \eta_i^{S_{n+1}} < 1 - \alpha } \left(\eta_i^{S_{n+1}} + \alpha \right)f(X_i)\\
    & = \frac{d}{d\epsilon} (n+1)\mathcal{R}(\hat{g}_{S_{n+1}} + \epsilon f) \bigg|_{\epsilon = 0} + \sum_{i=1}^{n+1}\left(\alpha - \bone\left\{\eta^{S_{n+1}}_i = 1 - \alpha\right\}\right) f(X_i)\\
    & \quad \quad -  \sum_{i : S_i = \hat{g}_{S_{n+1}}(X_i),\ \eta_i^{S_{n+1}} < 1 - \alpha } \left(\eta_i^{S_{n+1}} + \alpha \right)f(X_i).
    \end{align*}
To relate this stationary condition to the coverage note that 
\begin{align*}
    & \mme[f(X_{n+1})(\bone\{Y_{n+1} \in \hat{C}_{\text{dual}}(X_{n+1})\}- (1-\alpha))] = \mme[f(X_{n+1})(\alpha - \bone\{Y_{n+1} \notin \hat{C}_{\text{dual}}(X_{n+1})\})]\\
    & = \mme\left[f(X_{n + 1})\left(\alpha - \bone\left\{\eta^{S_{n+1}}_{n+1} = 1 - \alpha \right\} \right)\right]\\
    & = \mme\left[\frac{1}{n+1} \sum_{i=1}^{n+1} f(X_i)\left(\alpha -  \bone\left\{\eta^{S_{n+1}}_{i} = 1 - \alpha\right\} \right)\right]\\
    & = -\mme \left[ \frac{d}{d\epsilon} \mathcal{R}(\hat{g}_{S_{n+1}} + \epsilon f) \bigg|_{\epsilon = 0}\right] + \mme\left[\frac{1}{n + 1} \sum_{i : S_i = \hat{g}_{S_{n+1}}(X_i),\ \eta_i^{S_{n+1}} < 1 - \alpha } \left(\eta_i^{S_{n+1}} + \alpha \right)f(X_i) \right].
\end{align*}
Finally, since $\eta_i^{S_{n+1}} \in [- \alpha,1 - \alpha]$ the second term above can be bounded as 
\begin{align*}
\left| \mme\left[\frac{1}{n + 1} \sum_{i : S_i = \hat{g}_{S_{n+1}}(X_i),\ \eta_i^{S_{n+1}} < 1 - \alpha } \left(\eta_i^{S_{n+1}} + \alpha \right)f(X_i) \right] \right| & \leq \mme\left[\frac{1}{n+1}\sum_{i=1}^{n+1} \bone\{S_i = \hat{g}_{S_{n+1}}(X_i)\} |f(X_i)|\right]\\
& = \mme[\bone\{S_i = \hat{g}_{S_{n+1}}(X_i)\} |f(X_i)|],
\end{align*}
while when $f \geq 0$, we additionally have the lower bound
\begin{align*}
& \mme\left[\frac{1}{n + 1} \sum_{i : S_i = \hat{g}_{S_{n+1}}(X_i),\ \eta_i^{S_{n+1}} < 1 - \alpha } \left(\eta_i^{S_{n+1}} + \alpha \right)f(X_i) \right]\\
& \geq \mme\left[\frac{1}{n + 1} \sum_{i : S_i = \hat{g}_{S_{n+1}}(X_i),\ \eta_i^{S_{n+1}} < 1 - \alpha } \left(-\alpha + \alpha \right)f(X_i) \right] \geq 0.
\end{align*}
This concludes the proof of the first part of the proposition. For the second part of the proposition one simply notes that for any $f \in \mathcal{F}$,
\begin{align*}
\mme[f(X_{n+1})(\bone\{Y_{n+1} \in \hat{C}_{\text{dual, rand.}}(X_{n+1})\} - (1-\alpha))] & = \mme[f(X_{n+1})(\bone\{\eta^{S_{n+1}}_{n+1} <U\} - (1-\alpha))]\\
& = \mme[f(X_{n+1})(\mme[\bone\{\eta^{S_{n+1}}_{n+1} <U\} \mid X_{n +1}, \eta^{S_{n+1}}_{n+1}] - (1-\alpha))]\\
& = -\mme[f(X_{n+1})\eta^{S_{n+1}}_{n+1} ]\\
& = -\mme\left[\frac{1}{n+1}\sum_{i=1}^{n+1}f(X_{n+1})\eta^{S_{n+1}}_{i}  \right]\\
& = -\mme\left[\frac{d}{d\epsilon} \mathcal{R}(\hat{g}_{S_{n+1}} + \epsilon f) \right],
\end{align*}
where the last equality is simply our first order condition \eqref{eq:lang_dir_deriv}.

\end{proof}

\subsection{Two-sided fitting}\label{sec:app_two-sided}

Recall that in the main text we defined the two-sided prediction set
\[
\hat{C}_{\textup{two-sid.}}(X_{n+1}) = \{y : \hat{g}^{\alpha/2}_{S(X_{n+1},y)}(X_{n+1}) \leq S(X_{n+1},y) \leq \hat{g}^{1-\alpha/2}_{S(X_{n+1},y)}(X_{n+1}) \}.
\]
Analogues to our work on one-sided prediction sets in the main text, in this section we outline the coverage properties and computationally efficient implementation of $\hat{C}_{\textup{two-sid.}}(X_{n+1})$.

To begin, let $\eta^{S,\tau}$ denote an optimal solution to the dual program \eqref{eq:generic_dual} when $\alpha$ is replaced by $1-\tau$, i.e. recalling the definition of $\mathcal{R}^* (\eta) := - \min_{g \in \mathcal{F}} (n+1)\mathcal{R}(g) - \sum_{i=1}^{n+1} \eta_i g(X_i)$, let $\eta^{S,\tau}$ be a solution to
\begin{equation*}
\begin{split}
    \underset{\eta \in \mmr^{n+1}}{\text{maximize}} \quad & \sum_{i = 1}^n \eta_i S_i + \eta_{n+1} S - \mathcal{R}^*\left( \eta \right)   \\
    \text{subject to} \quad & -(1-\tau) \leq \eta_i \leq \tau.
 \end{split}
\end{equation*} 
Then, similar to our one-sided sets, $\hat{C}_{\textup{two-sid.}}(X_{n+1})$ also admits the analogous dual formulation
\begin{equation}\label{eq:two-sided_dual_set}
\hat{C}_{\textup{dual, two-sid.}}(X_{n+1}) = \{y : \eta^{S(X_{n+1},y),\alpha/2} > -(1-\alpha/2) \text{ and } \eta^{S(X_{n+1},y),1-\alpha/2} < 1-\alpha/2\},
\end{equation}
and the analogous randomized prediction set
\begin{equation}\label{eq:rand-two-sided_dual_set}
\hat{C}_{\textup{dual, two-sid., 
 rand.}}(X_{n+1}) = \{y : \eta^{S(X_{n+1},y),\alpha/2} > U_1 \text{ and } \eta^{S(X_{n+1},y),1-\alpha/2} < U_2\},
\end{equation}
where $U_1 \sim \text{Unif}([-(1-\alpha/2),\alpha/2])$ and independently $U_2 \sim \text{Unif}([-\alpha/2,1-\alpha/2])$.

As the next theorem states formally, these two-sided predictions set have identical coverage properties to their one-sided analogs.

\begin{theorem}\label{thm:two-sided_cov}
    Let $\mathcal{F}$ be any vector space, and assume that for all $f,g \in \mathcal{F}$, the derivative of $\epsilon \mapsto \mathcal{R}(g + \epsilon f)$ exists. If $f$ is non-negative with $\mme_P[f(X)] > 0$, then the unrandomized prediction set given by \eqref{eq:two-sided_dual_set} satisfies the lower bound
    \begin{align*}
    \mmp_f(Y_{n+1} \in \hat{C}_{\textup{dual, two-sid.}}(X_{n+1})) \geq 1- \alpha  & - \frac{1}{\mme_P[f(X)]} \mme\left[  \frac{d}{d\epsilon} \mathcal{R}(\hat{g}^{1-\alpha/2}_{S_{n+1}} + \epsilon f) \bigg|_{\epsilon = 0} \right]\\
    & - \frac{1}{\mme_P[f(X)]} \mme\left[  \frac{d}{d\epsilon} \mathcal{R}(\hat{g}^{\alpha/2}_{S_{n+1}} + \epsilon f) \bigg|_{\epsilon = 0} \right]. 
    \end{align*}
    On the other hand, suppose $(X_1,Y_1),\dots,(X_{n+1},Y_{n+1}) \stackrel{i.i.d.}{\sim} P$. Then, for all $f \in \mathcal{F}$, we additionally have the two-sided bound,
    \begin{equation}    \label{eq:infinite_dim_cov_2}
    \begin{split}
    & \mme[f(X_{n+1})(\bone\{Y_{n+1} \in \hat{C}_{\textup{dual, two-sid.}}(X_{n+1})\} - (1-\alpha))]\\
    & = - \frac{1}{\mme_P[f(X)]} \mme\left[  \frac{d}{d\epsilon} \mathcal{R}(\hat{g}^{1-\alpha/2}_{S_{n+1}} + \epsilon f) \bigg|_{\epsilon = 0} \right] - \frac{1}{\mme_P[f(X)]} \mme\left[  \frac{d}{d\epsilon} \mathcal{R}(\hat{g}^{\alpha/2}_{S_{n+1}} + \epsilon f) \bigg|_{\epsilon = 0} \right] + \epsilon_{\textup{int}},
    \end{split}
    \end{equation}
    where $\epsilon_{\textup{int}}$ is an interpolation error term satisfying $|\epsilon_{\textup{int}}| \leq \mme[|f(X_i)| \bone\{S_i = \hat{g}^{1-\alpha/2}_{S_{n+1}}(X_i)\}] + \mme[|f(X_i)| \bone\{S_i = \hat{g}^{\alpha/2}_{S_{n+1}}(X_i)\}]$. Furthermore, the same results also hold for the randomized set \eqref{eq:rand-two-sided_dual_set} with $\epsilon_{\text{int}}$ replaced by $0$.
\end{theorem}

\begin{proof}
    Note that
    \begin{align*}
    & \mme[f(X_{n+1})(\bone\{Y_{n+1} \in \hat{C}_{\textup{dual, two-sid.}}(X_{n+1})\} - (1-\alpha))]\\
    & = \mme[f(X_{n+1})(\alpha - \bone\{Y_{n+1} \notin \hat{C}_{\textup{dual, two-sid.}}(X_{n+1})\})]\\
    & = \mme[f(X_{n+1})(\alpha/2 - \bone\{\eta^{\alpha/2,S_{n+1}}_{n+1} = -(1-\alpha/2)\})] + \mme[f(X_{n+1})(\alpha/2 - \bone\{\eta^{1-\alpha/2,S_{n+1}}_{n+1} = (1-\alpha/2)\})].
    \end{align*}
    The result then follows by repeating the steps of \Cref{prop:cov_of_dual} twice to bound the two terms above separately. A similar argument demonstrates the coverage of $\hat{C}_{\textup{dual, two-sid., rand.}}$
\end{proof}

\subsection{Proofs of additional technical lemmas} \label{sec:technical_proofs}

    \begin{lemma}[Lipschitz property of the pinball loss]\label{lem:pinball_loss_is_lip}
        The pinball loss is 1-Lipschitz in the sense that for any $y_1, y_2, y_3, y_4 \in \mmr$,
        \[
        |\ell_{\alpha}(y_1,y_2) - \ell_{\alpha}(y_3,y_4)| \leq |(y_1 - y_2)  - (y_3 - y_4)|.
        \]
    \end{lemma}
    \begin{proof}
    We will show that $\ell_{\alpha}(y_1,y_2) - \ell_{\alpha}(y_3,y_4) \leq |(y_1 - y_2)  - (y_3 - y_4)|$. The reverse inequality will then follow by symmetry. There are four cases.

    \paragraph{Case 1:} $y_1 \geq y_2,\ y_3 \geq y_4$.
    \[
    \ell_{\alpha}(y_1,y_2) - \ell_{\alpha}(y_3,y_4) = \alpha(y_1 - y_2) - \alpha(y_3 - y_4) \leq |(y_1 - y_2)  - (y_3 - y_4)|.
    \]
    \paragraph{Case 2:} $y_1 < y_2,\ y_3 < y_4$.
    \[
    \ell_{\alpha}(y_1,y_2) - \ell_{\alpha}(y_3,y_4) = (1-\alpha)(y_2 - y_1) - (1-\alpha)(y_4 - y_3) \leq |(y_1 - y_2)  - (y_3 - y_4)|.
    \]
    \paragraph{Case 3:} $y_1 \geq y_2,\ y_3 < y_4$.
    \begin{align*}
    \ell_{\alpha}(y_1,y_2) - \ell_{\alpha}(y_3,y_4) & = \alpha(y_1- y_2) - (1-\alpha)(y_4 - y_3)\\
    &= \alpha(y_1- y_2 - (y_3 - y_4)) + (y_3 - y_4) \leq |(y_1 - y_2)  - (y_3 - y_4)| .
    \end{align*}
    \paragraph{Case 3:} $y_1 < y_2,\ y_3 \geq y_4$.
    \begin{align*}
    \ell_{\alpha}(y_1,y_2) - \ell_{\alpha}(y_3,y_4) & = (1-\alpha)(y_2- y_1) - \alpha(y_3 - y_4)\\
    &= (1-\alpha)(y_2- y_1 - (y_4 - y_3)) + (y_4 - y_3) \leq |(y_1 - y_2)  - (y_3 - y_4)| .
    \end{align*}
    \end{proof}

    \begin{lemma}\label{lem:lower_concentration_of_phi}\sloppy
        Let $\Phi(X_1),\dots,\Phi(X_{n+1}) \in \mmr^d$ be i.i.d.~and assume that there exists constants $C_2, c_2, \rho > 0$ such that $\sup_{\beta : \|\beta\|_2 = 1} \mme[|\Phi(X_i)^\top \beta|^2]^{1/2} \leq c_2$, $\mme[\|\Phi(X_i)\|^2] \leq C_2 d$, and $\inf_{\beta : \|\beta\|_2 = 1} \mme[|\Phi(X_i)^\top \beta|] \geq \rho$. Then there exists constants $c,c'>0$ (depending only on $C_2$, $c_2$, and $\rho$) such that,
    \[
    \mmp\left( \inf_{\beta : \|\beta\|_2 = 1}  \frac{1}{n+1}\sum_{i=1}^{n+1} |\Phi(X_i)^\top \beta| \geq c \right) \geq 1- c'\frac{d^2}{(n+1)^2}    
    \]
    \end{lemma}
    \begin{proof}
        The main idea of this proof is to apply Theorem 5.4 of \citet{Mendelson2014} and thus conclude that for some constant $a>0$, $|\{i : |\Phi(X_i)^\top \beta | > a(n+1)\}|$ is large uniformly in $\beta$. To apply this theorem we need to check two technical conditions. Namely, we need to show that the class of functions $x \mapsto |\Phi(x)^\top\beta|$ has bounded Rademacher complexity and that  $\mmp(|\Phi(X_i)^\top \beta | > \Omega(a))$ is not too small. We now check these conditions.
        
        Let $\sigma_1,\dots,\sigma_{n+1}$ denote i.i.d.~Rademacher random variables. Then, the Rademacher complexity of $x \mapsto |\Phi(x)^\top\beta|$ can be bounded as 
        \begin{align*}
        \mme\left[ \sup_{\beta : \|\beta\|_2 = 1} \left| \frac{1}{n+1} \sum_{i=1}^{n+1} \sigma_i |\Phi(X_i)^\top \beta| \right| \right] & = \mme\left[ \left\| \frac{1}{n+1} \sum_{i=1}^{n+1} \sigma_i \Phi(X_i)\right\|_2 \right]\\
        & \leq \mme\left[\left\|\frac{1}{n+1} \sum_{i=1}^{n+1} \sigma_i \Phi(X_i)\right\|_2^2\right]^{1/2} \leq \sqrt{\frac{C_2 d}{n+1}}.
        \end{align*}
        This verifies our the first technical condition. For the second condition note that by the Paley-Zygmund inequality
        \begin{align*}
        \mmp\left(|\Phi(X_i)^\top \beta | > \frac{1}{2}\rho\right) \geq \frac{\mme[|\Phi(X_i)^\top \beta |]^2}{4\mme[|\Phi(X_i)^\top \beta |^2]} \geq \frac{\rho^2}{4c^2_2}.
        \end{align*}
        Thus, by Theorem 5.4 in \citet{Mendelson2014} we have that there exists constants $a, b>0$ such that with probability at least $1-c'd^2/(n+1)^2$
        \[
        \inf_{\beta : \|\beta\| = 1} |\{i : |\Phi(X_i)^\top \beta | > a(n+1)\} \geq (n+1) b,
        \]
        and thus in particular,
        \[
        \inf_{\beta : \|\beta\| = 1}  \frac{1}{n+1} \sum_{i=1}^{n+1} |\Phi(X_i)^\top \beta | \geq ab.
        \]
        Taking $c=ab$ gives the desired result.
    \end{proof}

\end{document}